\documentclass[11pt]{article}

\RequirePackage[OT1]{fontenc}
\RequirePackage{amsthm,amsmath}
\RequirePackage[numbers]{natbib}
\RequirePackage[colorlinks,citecolor=blue,urlcolor=blue]{hyperref}

\usepackage[margin=1.3truein]{geometry}

\usepackage{amsthm}
\usepackage{amsmath}
\usepackage{amssymb}
\usepackage{graphicx}
\usepackage{xcolor}
\usepackage{array}
\usepackage{multirow}
\usepackage{caption}
\usepackage{enumitem}
\usepackage{natbib}
\usepackage{float}
\usepackage{tikz}
\usepackage{pgfplots}
\usepackage{fontenc}
\usepackage{amsthm}
\usepackage{multirow}

\newcommand{\mbf}[1]{\mbox{\boldmath $#1$}}
\newcommand{\ba}{{\mbf \beta}}

{\catcode `\@=11 \global\let\AddToReset=\@addtoreset}
\AddToReset{equation}{section}

\AddToReset{Theorem}{section}

\newtheorem{cor}{Corollary}[section]
\newtheorem{lem}{Lemma}[section]
\newtheorem{rem}{Remark}[section]
\newtheorem{thm}{Theorem}[section]

\newtheorem{Def}{Definition}[section]

\newcommand{\cC}{{\cal C}}

\newcommand{\cG}{{\cal G}}

\newcommand{\cL}{{\cal L}}
\newcommand{\cM}{{\cal M}}
\newcommand{\cN}{{\cal N}}

\newcommand{\cU}{{\cal U}}

\newcommand{\Np}{{\mathbb{N}}^{+}}
\newcommand{\Nn}{{\mathbb{N}}^{-}}

\def\ba{\begin{array}}
	\def\bc{\begin{center}}
		\def\bd{\begin{description}}
			\def\be{\begin{enumerate}}
				\def\ea{\end{array}}
			\def\ec{\end{center}}
		\def\ed{\end{description}}
	\def\edt{\end{document}}
\def\ee{\end{enumerate}}
\def\ben{\begin{equation}}
\def\benn{\begin{equation*}}
\def\een{\end{equation}}
\def\eenn{\end{equation*}}
\def\benr{\begin{eqnarray}}
\def\eenr{\end{eqnarray}}
\def\benrr{\begin{eqnarray*}}
\def\eenrr{\end{eqnarray*}}

\def\al{\alpha}

\def\edt{\end{document}}
\def\ep{\epsilon}
\def\g{\gamma}

\def\h{\hat}
\def\ka{\kappa}

\def\e{\mathrm e}

\def\iny{\infty}
\def\ka{\kappa}

\def\la{\lambda}

\def\noi{\noindent}

\def\nn{\nonumber}

\def\om{\omega}

\def\si{\sigma}

\def\Si{\Sigma}

\def\vep{\varepsilon}

\def\vs{\vskip}

\def\R{{\mathbb R}}
\def\Z{{\mathbb Z}}

\def\z{\zeta}

\DeclareMathOperator*{\argmin}{arg\,min}

\DeclareMathOperator*{\argmax}{arg\,max}

\usetikzlibrary{shapes,arrows}

\usepackage{titlesec}

\titleformat{\section}
{\normalfont\fontsize{11}{15}\bfseries}{\thesection}{1em}{}

\titleformat{\subsection}
{\normalfont\fontsize{11}{15}\bfseries}{\thesubsection}{1em}{}

\tikzstyle{block} = [rectangle, draw, fill=blue!20,
text width=5em, text badly centered, rounded corners, minimum height=4em,node distance=3cm]
\tikzstyle{line} = [draw, -latex']
\tikzstyle{linena} = [draw]
\tikzstyle{cloud} = [draw, ellipse,fill=red!20, minimum height=2em]

\usetikzlibrary{arrows,calc,angles,positioning,intersections,quotes,decorations.markings,backgrounds,patterns}
\usetikzlibrary{decorations.pathreplacing}
\pgfplotsset{compat=1.11}
\tikzset{
	mynode/.style={fill,circle,inner sep=1pt,outer sep=0pt}
}

\begin{document}

\bc
{\Large {\bf Inference on the change point in high dimensional time series models via plug in least squares}}\\[.5cm]
Abhishek Kaul$^a$\footnote{Email: abhishek.kaul@wsu.edu.}, Stergios B. Fotopoulos$^b$,\\ Venkata K. Jandhyala$^a$, and Abolfazl Safikhani$^c$\\[0.25cm]

$^a$Department of Mathematics and Statistics,\\ 
Washington State University, Pullman, WA 99164, USA.\\[0.25cm]

$^b$Department of Finance and Management Science,\\ 
Washington State University, Pullman, WA 99164, USA.\\[0.25cm]

$^c$Department of Statistics,\\University of Florida, Gainsville, FL 32611-8545, USA.

\ec
\vs .1in
{\renewcommand{\baselinestretch}{1}
	\begin{abstract}
We study a plug in least squares estimator for the change point parameter where change is in the mean of a high dimensional random vector under subgaussian or subexponential distributions. We obtain sufficient conditions under which this estimator possesses sufficient adaptivity against plug in estimates of mean parameters in order to yield an optimal rate of convergence $O_p(\xi^{-2})$ in the integer scale. This rate is preserved while allowing high dimensionality as well as a potentially diminishing jump size $\xi,$ provided $s\log (p\vee T)=o(\surd(Tl_T))$ or $s\log^{3/2}(p\vee T)=o(\surd(Tl_T))$ in the subgaussian and subexponential cases, respectively. Here $s,p,T$ and $l_T$ represent a sparsity parameter, model dimension, sampling period and the separation of the change point from its parametric boundary, respectively. Moreover, since the rate of convergence is free of $s,p$ and logarithmic terms of $T,$ it allows the existence of limiting distributions under high dimensional asymptotics. These distributions are then derived as the {\it argmax} of a two sided negative drift Brownian motion or a two sided negative drift random walk under vanishing and non-vanishing jump size regimes, respectively, thereby allowing inference on the change point parameter. Feasible algorithms for implementation of the proposed methodology are provided. Theoretical results are supported with monte-carlo simulations.
\end{abstract} }
\noi {\it Keywords: change point, inference, high dimensions, limiting distribution.}

\section{Introduction}\label{sec:intro}

In many applications of current scientific interest assuming stationarity of the mean of a time series over an extended sampling period may be unrealistic and may lead to flawed inference. Dynamic time series characterized via mean changes across unknown time points form a simplistic yet useful tool to model non-stationarity of large streams of data. In this article we consider the model,
\benr\label{model:rvmcp}
x_t=\begin{cases}\theta_1^0+\vep_t & t=1,...,\tau^0\\
	\theta_2^0+\vep_t & t=\tau^0+1,...,T.\end{cases}
\eenr
The observed variables here are $x_t=(x_{t1},x_{t2},...,x_{tp})^T\in\R^{p},$ $t=1,...,T.$ The variables $\vep_t=(\vep_{t1},...,\vep_{tp})^T\in\R^p$ are unobserved zero mean random variables, which are allowed to be subgaussian or subexponential. The unknown parameters are the mean vectors $\theta_1^0,\theta_2^0\in\R^p,$ and the change point parameter $\tau^0\in \{0,1,2,...,T\},$ with the latter being of main interest. The model dimension $p$ is allowed to be fixed or diverging potentially much faster than the sampling period $T.$ The boundary points $\tau^0=0,T$ characterize the `no change' case, or a static model where no realizations from the corresponding distribution are observed. These boundary points are considered to present additional theoretical insights in the estimation of $\tau^0$ later in the manuscript, however these shall not be pursued from an inference perspective. Our objective throughout the article is that of inference on $\tau^0$ when it exists, i.e., construction of asymptotically valid confidence intervals for the change point parameter when it is not at the boundary of its parametric space. We mention here that several solutions for the boundary problem (detection) of testing the null hypothesis $H_0:\,\tau^0=T,$ in the high dimensional setting are available in the literature, see, e.g.,  \cite{jirak2015uniform}, \cite{wang2019inference}, \cite{enikeeva2013high}, \cite{cho2016change} and \cite{steland2018inference} amongst others.

To proceed further, define the jump vector and the jump size that are fundamentally related to the properties of any change point estimator. Let,
\benr\label{def:jumpsize}
\eta^0=(\theta_1^0-\theta_2^0),\quad{\rm and}\quad \xi=\|\eta^0\|_2.\footnotemark
\eenr
\footnotetext{These quantities depend on the sampling period $T,$ however this dependence is notationally suppressed for clarity of exposition.}
The problem of change point estimation in the high dimensional setting has received significant attention in the recent past and several different estimators have been proposed. A large proportion of this literature provides near optimal localization error bounds of the form $|\h\tau-\tau^0|\le O(\xi^{-2}a_T),$ where $a_T\to\iny,$ with probability (w.p.) $\to 1.$ For e.g.  in the case of a single change point, the results of \cite{harchaoui2010multiple} yield $a_T=\log T,$ with a least squares estimator together with a total variation penalty, \cite{wang2018high} provide $a_T=\log\log(T)$ with a projected cusum estimator, and those of \cite{cho2015multiple} yields $a_T\ge \log^{2}(T)$ in the case of a single change point. While near optimal rates of the approximation are informative from an estimation perspective, however, from an inference perspective one requires a change point estimator to obey an optimal rate of convergence of $O_p(\xi^{-2})$ in order to allow the existence of limiting distributions and in turn allow inference on $\tau^0.$ The literature on this inference perspective is very sparse. In a setting where $p$ is increasing with $T,$  \cite{bhattacharjee2018change} and \cite{bhattacharjee2019change} develop limiting distribution results while assuming  $\xi^{-1}\surd(p/T)\to 0.$ These results yield non-degenerate limiting distributions provided $p<<T.$ However, due to the assumption made on rate of the jump size, these do not extend to the high dimensional case where $p$ may be diverging faster than $T.$ In this case the assumption on the jump size made in these articles necessitates $\xi\to\iny$ in the high dimensional setting and consequently allows only a degenerate limiting distribution to remain valid. The article \cite{bai2010common} provides a limiting distribution result while assuming a further stronger assumption of $\xi^{-1}\surd{p}\to 0.$ More generally, in the high dimensional setting the question of an optimal rate of estimation and that of inference on the change point parameter in the non-degenerate case where the jump size is not assumed to be diverging remains unaddressed. The viability of the question itself comes from the recent work of \cite{liu2019minimax} who show that assuming sparsity of the jump vector, much weaker signals in the jump size are detectable. Specifically, they show that the region of detectability of the change point satisfies a rate of $\xi^{-1}\surd\big\{s\log (p\vee T)/T\big\}\le c,$ in a minimax sense, upto other logarithmic terms in $s$ and $T,$ and under restrictions on the sparsity parameter $s.$ We refer to their article for the precise rate which involves a tripe iterated log expression. A corresponding result in the univariate setting has been provided in \cite{wang2020univariate}.

A more precise description of the purpose of this article requires additional notation. For any $T\ge 2,$ $p\ge 3,$\footnote{We assume $p\ge 3$ throughout the article so that $\log p\ge 1.$ This is not a necessary condition and is only assumed to ease notational complexity of the results and proofs.} $\tau\in\{0,...,T\},$ and $\theta_1,\theta_2\in\R^p,$ consider,
\benr\label{def:Q}
Q(\tau,\theta_1,\theta_2)=\frac{1}{T}\sum_{t=1}^{\tau} \|x_t-\theta_1\|_2^2 +\frac{1}{T}\sum_{t=\tau+1}^{T} \|x_t-\theta_2\|_2^2.\footnotemark
\eenr
\footnotetext{The sum $\sum_{t=1}^{\tau}$ is defined to be zero when $\tau=0,$ and similar for the other sum on the boundary $\tau=T.$}
Assume for the time being, the availability of some estimates $\h\theta_1$ and $\h\theta_2$ of the mean parameters of the model (\ref{model:rvmcp}) and consider the following plug in estimator utilizing these nuisance estimates,
\benr\label{est:optimalcp}
\tilde\tau:=\tilde\tau(\h\theta_1,\h\theta_2)=\argmin_{0\le \tau\le T} Q(\tau,\h\theta_1,\h\theta_2).
\eenr

The overarching objective of this article is to study the inference properties of the estimator $\tilde\tau$ in the assumed setting allowing high dimensionality and weak requirements on the jump size that allow non-degenerate limiting distributions, for e.g. to allowing a potentially diminishing jump size. Establishing existence of limiting distributions requires first suitable estimation properties to hold, which forms the first main contribution of this article.

In particular, we shall show that $\tilde\tau$ yields an optimal rate of convergence $O_p(\xi^{-2}),$ under a subgaussian or subexponential setting and any nearly arbitrary spatial dependence structure. New arguments are developed to obtain this result, including a novel application of the Kolmogorov's inequality on partial sums (see, Theorem \ref{thm:kolmogorov}). Moreover, we obtain sufficient conditions on the nuisance estimates $\h\theta_1$ and $\h\theta_2$ required to achieve the optimal rate $O_p(\xi^{-2}),$ or a near optimal rate $O_p\big(s\log (p\vee T)\big).$ These sufficient conditions on nuisance estimation are stated as an inter-relationship between the $\ell_2$ error of nuisance estimates and the jump size (Condition C.1 and Condition C.2). Formulating sufficient conditions on nuisance parameters with respect to the jump size provides some surprising insights. For e.g, they allow us to show that the estimation of a change point parameter in itself does not require many assumptions that are typically thought of as necessary conditions in the literature, including a rate condition on the separation of $\tau^0$ from the parametric boundary, a rate condition on the jump size, and a rate of divergence of the model dimension. Instead, these assumptions arise from the nuisance estimation aspect of the overall process of change point estimation. This is best observed for the case where the nuisance parameters are known. Here these sufficient conditions shall be trivially satisfied with $\h\theta_1=\theta_1^0,$ and $\h\theta_2=\theta_2^0.$ In this case, $\tilde\tau$ yields an optimal rate $O_p(\xi^{-2}),$ where $\xi>0$ may be converging arbitrarily fast towards zero and the model dimension may be diverging arbitrarily fast with respect to $T,$ and even when the change point does not actually exist ($\tau^0=0\,{\rm or}\,T$)\footnote{The boundary case of $\tau^0=0,T$ and $\xi>0$ can be simultaneously assumed since we allow $\theta_1^0,$ $\theta_2^0$ and $\tau^0$ to be free parameters. Effectively, $\tau^0=0,$ and $\xi>0,$ assumes there is some mean vector $\theta_2^0$ different from $\theta_1^0,$ however we do not observe any realizations from $\theta_1^0+\vep,$ and symmetrically for $\tau^0=T.$}. This case of known nuisance parameters is clearly infeasible in practice and is only meant to illustrate the above subtlety. The main requirement to obtain an $O_p(\xi^{-2})$ rate for $\tilde\tau$ in the usual case of unknown nuisance parameters shall effectively take the form
\benr\label{eq:40}
\|\h\theta_1-\theta_1^0\|_2\vee\|\h\theta_2-\theta_2^0\|_2\le c_u\si\Big\{\frac{s\log (p\vee T)}{Tl_T}\Big\}^{\frac{1}{2}},
\eenr
under the following weak condition on the rates of model parameters,
\benr
\Big(\frac{c_u\si}{\xi}\Big)\Big\{\frac{s\log (p\vee T)}{\surd(Tl_T)}\Big\}\le c_{u1},
\quad{\rm or}\quad \Big(\frac{c_u\si}{\xi}\Big)\Big\{\frac{s\log^{3/2} (p\vee T)}{\surd (Tl_T)}\Big\}\le c_{u1}\label{eq:41}
\eenr
for the subgaussian and subexponential cases, respectively, and for a suitably chosen small enough constant $c_{u1}>0.$ Here $s$ is a sparsity parameter defined later in (\ref{def:setS}),  $\si$ is a variance proxy parameter (Condition A) and $l_T$ is a sequence separating $\tau^0$ from the parametric boundary (Condition D).

Despite irregular $p$-dimensional nuisance estimates in the construction of $\tilde\tau,$ it shall yield an optimal $O_p(\xi^{-2})$ rate of convergence. This indicates that under the assumed mild conditions largely described above, the estimator $\tilde\tau$ statistically behaves as if the nuisance parameters are known. This property of an estimator is typically referred to as adaptation as described in \cite{bickel1982adaptive}, but is observed here in a high dimensional sense. An indirect but informative comparison is with recent results on inference on regression coefficients in high dimensional regression models. For estimation of a component of the regression vector, it is known that the least squares estimator itself is not sufficiently adaptive against nuisance parameter estimates (estimates of remaining regression vector components) to allow for an optimal rate of convergence. Instead, certain corrections to the least squares loss or its first order moment equations, such as debiasing (\cite{van2014asymptotically}) or orthogonalization (\cite{belloni2011inference},   \cite{chernozhukov2015valid}, \cite{belloni2017confidence} and \cite{ning2017general}) induce sufficient adaptivity against nuisance estimates and thereby allow optimal estimation of the target regression parameter. The results of this article show that in the context of change point estimation, the plugin least squares estimator (\ref{est:optimalcp}) itself possesses the required adaptivity against potentially high dimensional nuisance estimates, in order to allow for $O_p(\xi^{-2})$ estimation of the change point $\tau^0$ provided the nuisance parameters are estimated with sufficient precision.

It may be observed that taking advantage of sparsity yields conditions (\ref{eq:41}) that are weaker than those assumed in \cite{bai2010common}, \cite{bhattacharjee2018change} and \cite{bhattacharjee2019change}. This is best seen by noting that conditions (\ref{eq:41}) allow a diminishing jump despite high dimensionality, under the restrictions $s\log (p\vee T)=o\big(\surd(Tl_T)\big),$ and $s\log^{3/2} (p\vee T)=o\big(\surd(Tl_T)\big),$ under the subgaussian and subexponential cases, respectively. We also mention here that while the estimators studied in \cite{bai2010common} and \cite{bhattacharjee2019change} are also based on a squared loss, however, they consider a grid search least squares where estimation of nuisance parameters $\theta_1^0$ and $\theta_2^0$ is carried out internally in the change point estimation mechanism. This is in contrast to the plug in least squares estimator (\ref{est:optimalcp}) where the nuisance estimation has been separated from the change point estimation. This separation is crucial since it allows nuisance estimates to be computed separately and be made suitable for high dimensional approximations of the mean vectors via regularization, whereas a grid search least squares by construction disallows this capability.

Another observation here is that these conditions impose a stronger requirement on the jump size in the subexponential case in comparison to the subgaussian case. While we do not prove that these are necessary assumptions for an $O_p(\xi^{-2})$ rate of $\tilde\tau,$ however the tail probability bounds (e.g. Bernstein's inequality) that force the conditions (\ref{eq:41}) are known to be sharp bounds. It is thus reasonable to speculate that this relationship between the jump size, dimensionality and the underlying distribution is inherent to achieving an optimal rate of convergence $O_p(\xi^{-2})$ of the change point estimator and not an artifact of our argument. Further circumstantial evidence towards this also comes from the following additional result. We show that a near optimal rate $O_p\big(\xi^{-2}s\log (p\vee T)\big)$ of $\tilde\tau$ can be obtained under a weaker condition than (\ref{eq:40}) on the nuisance estimates, which shall in turn requires the weaker restriction,
\benr\label{eq:43}
\Big(\frac{c_u\si}{\xi}\Big)\Big\{\frac{s\log (p\vee T)}{Tl_T}\Big\}^{\frac{1}{2}}\le c_{u1},
\eenr
for both the subgaussian and subexponential settings.

Notably, the distinction in the required conditions (\ref{eq:41}) for the two classes of distributions is no longer present when only a near optimal rate is of interest. An intuitive explanation for this behavior is as follows. One among a few quantities that controls the rate of $\tilde\tau$ is the tail behavior of a stochastic term of the form $\big\|\sum_{t=(\tau^0+1)}^{(\tau^0+k)}\vep_t\big\|_{\iny}$ uniformly over $k\le k'.$ Note that when $k'\ge \log (p\vee T)$ is diverging with $T,$ then for a sufficiently large $T,$ this tail behavior is of the same order under both the subgaussian and subexponential cases. When only a near optimal rate is of interest, it is sufficient to examine this case with a diverging $k'.$ However, this is no longer true in the case where $k'$ is finite. In this case, the heavier tail of the subexponential distribution is realized in the corresponding tail bound of the underlying stochastic term, and in turn on the assumption required to retain an optimal $O_p(\xi^{-2})$ rate of convergence.

The second main contribution of this article is about inference on the change point parameter. Note that in the case where $\xi\to \iny,$ the rate $O_p(\xi^{-2})$ directly yields a degenerate limiting distribution. It is thus sufficient to restrict this analysis to $\xi=O(1).$ We show that the optimal rate of $\tilde\tau,$ together with peripheral results allows for the existence of limiting distributions of $\tilde\tau$ in both vanishing and non-vanishing jump size regimes, the forms of which are then derived. More precisely, under the vanishing jump regime $\xi\to 0,$ we obtain,
\benr\label{eq:limiting.d.vanishing}
\xi^2\si^{-2}_{\iny}(\tilde\tau-\tau^0)\Rightarrow \argmax_{\z\in\R}\big(2W(\z)-|\z|\big),
\eenr
where $\si^2_{\iny}=\lim_{T\to\iny}(\eta^{0T}\Sigma\eta^0\big)/\xi^2,$ $\Sigma=E\vep_t\vep_t^T,$ and $W(\cdot)$ is a two-sided Brownian motion on $\R.$ It may be observed that the form of the limiting distribution obtained here is the same as that obtained in a one dimensional setting, (\cite{bai1994}). The distribution of $\argmax_{\z\in\R}\big(2W(\z)-|\z|\big)$ is well studied in the literature and its cdf and thus its quantiles are readily available, (\cite{yao1987approximating}).

The limiting distribution in the non-vanishing case of $\xi\to\xi_{\iny}>0$ necessitates a further parametric assumption (Condition A$'$) on the form of the underlying distribution, the reason for which is discussed in Section \ref{sec:inference} later in the manuscript. The literature on this case even in the classical fixed $p$ setting is quite sparse. Some relevant articles in this direction are of \cite{jandhyala1999capturing} and \cite{fotopoulos2010exact}. However, the results of these articles do not allow an extension to the case where the dimension $p$ is function of $T.$ When $p$ is allowed to move with $T,$ but $p<<T,$ the only articles we are aware of who consider the non-vanishing case are of \cite{bhattacharjee2018change} and \cite{bhattacharjee2019change}. However it may be observed that our result to follow is quite different in comparison to theirs and is additionally valid in the high dimensional setting. To describe this distribution, let $\cL$ represent the parametric form of the distribution of the random variable $\big(2\eta^{0T}\vep_t-\xi^2\big)$ and define the following negative drift two sided random walk initializing at the origin,
\benr\label{eq:neg.drift.rw}
\cC_{\iny}(\z)=
\begin{cases}\sum_{t=1}^{\z} z_t, & \z\in \Np=\{1,2,3,...\} \\ 	
	0,				  &	\z=0 \\
	\sum_{t=1}^{-\z}z_t^*,		  &	\z\in \Nn=\{-1,-2,-3,...\},
\end{cases}
\eenr
where $z_t,z_t^*$ are independent copies of a $\cL(-\xi_{\iny}^2,4\xi_{\iny}^2\si^2_{\iny})$ distribution, which are also independent over all $t.$ The notation in the arguments of $\cL(\cdotp,\cdotp)$ is representative of the mean and variance of this distribution, where the limits $\xi_{\iny}$ and $\si_{\iny}^2$ are as described earlier. Then, we obtain the following result,
\benr\label{eq:limiting.d.non.vanishing}
(\tilde\tau-\tau^0)\Rightarrow \argmax_{\z\in \Z}\cC_{\iny}(\z),
\eenr
where $\Z$ represents the collection of integers. Quantiles of this distribution can be approximated numerically thereby enabling the construction of asymptotically valid confidence intervals. We emphasize  that asymptotics here are in a high dimensional sense, where the sampling period $T\to \iny$ and the dimension $p$ may be fixed or be allowed to diverge, potentially at an exponential rate of $T.$

Clearly all of the above discussion on estimation and inference on $\tau^0$ relies critically on the nuisance estimates  $\h\theta_1$ and $\h\theta_2$ that satisfy suitable conditions, that have not yet been explicitly defined. We postpone this discussion to Section \ref{sec:algorithm} where the construction of these nuisance estimates is discussed, along with validity of the assumed sufficient conditions. Section \ref{sec:mainresults} and Section \ref{sec:inference} study the proposed plug in least squares estimator $\tilde\tau$  and provide a rigorous description of the estimation and inference results discussed above. Section \ref{sec:numerical} provides monte-carlo simulations numerically supporting the theoretical results developed in this article. We conclude this section with a short note on the notation used throughout the article.

%%%%%%%%%%%%%%%%%%%%%%%%%%%%%%%%%%%%%%%%%%%%%%%%%%%%%%%%%%%%%%%%%%%%%%%%%%%%%%%%%%%%%%%%%%%%%%%%%%%%%%%%%%%%%%%%%%%%%%%%%%%%%%%%%%%%%%%%%%%%%%%%%%%%%%%%%%%%%%%%%%%%%%%%%%%
\vspace{1.5mm}
\noi{\it Notation}: Throughout the paper, $\R$ represents the real line. For any vector $\delta\in\R^p,$ $\|\delta\|_1,$ $\|\delta\|_2,$ $\|\delta\|_{\iny}$ represent the usual 1-norm, Euclidean norm, and sup-norm respectively. For any set of indices $U\subseteq\{1,2,...,p\},$ let $\delta_U=(\delta_j)_{j\in U}$ represent the subvector of $\delta$ containing the components corresponding to the indices in $U.$ Let $|U|$ and $U^c$ represent the cardinality and complement of $U.$ We denote by $a\wedge b=\min\{a,b\},$ and $a\vee b=\max\{a,b\},$ for any $a,b\in\R.$ We use a generic notation $c_u>0$ to represent universal constants that do not depend on $T$ or any other model parameter. All limits are with respect to the sample size $T\to\iny.$ The notation $\Rightarrow$ represents convergence in distribution.
%%%%%%%%%%%%%%%%%%%%%%%%%%%%%%%%%%%%%%%%%%%%%%%%%%%%%%%%%%%%%%%%%%%%%%%%%%%%%%%%%%%%%%%%%%%%%%%%%%%%%%%%%%%%%%%%%%%%%%%%%%%%%%%%%%%%%%%%%%%%%%%%%%%%%%%%%%%%%%%%%%%%%%%%%%%

\section{Assumptions and estimation properties}\label{sec:mainresults}

In this section we state sufficient conditions and theoretical results regarding estimation properties of the plug in least squares estimator $\tilde\tau$ of (\ref{est:optimalcp}).

%%%%%%%%%%%%%%%%%%%%%%%%%%%%%%%%%%%%%%%%%%%%%%%%%%%%%%%%%%%%%%%%%%%%%%%%%%%%%%%%%%%%%%%%%%%%%%%%%%%%%%%%%%%%%%%%%%%%%%%%%%%%%%%%%%%%%%%%%%%%%%%%%%%%%%%%%%%%%%%%%%%%%%%%%%%
\vspace{1.5mm}
{\it {{\noi{\bf Condition A (on underlying distributions):}} We assume that the underlying distribution in model (\ref{model:rvmcp}) obeys one of the following two conditions.\\~
		{\rm (I)} {\bf (subgaussian):} The vectors $\vep_t=(\vep_{t1},...,\vep_{tp})^T,$ $t=1,..,T,$ are independent and identically distributed subgaussian random vectors with variance proxy $\si^2<\iny$ (see, Definition \ref{def:subg} and \ref{def:submult}).\\~
		{\rm (II)} {\bf (subexponential):}The vectors $\vep_t=(\vep_{t1},...,\vep_{tp})^T,$ $t=1,..,T,$ are independent and identically distributed subexponential random vectors with variance proxy $\si^2<\iny$ (see, Definition \ref{def:sube} and \ref{def:submult})}}
%%%%%%%%%%%%%%%%%%%%%%%%%%%%%%%%%%%%%%%%%%%%%%%%%%%%%%%%%%%%%%%%%%%%%%%%%%%%%%%%%%%%%%%%%%%%%%%%%%%%%%%%%%%%%%%%%%%%%%%%%%%%%%%%%%%%%%%%%%%%%%%%%%%%%%%%%%%%%%%%%%%%%%%%%%%

\vspace{1.5mm}

Subgaussian and subexponential are well known classes of distributions with the latter being heavier tailed than the former. Distributions included in class (I) are: Gaussian distribution, any bounded distribution, asymmetric mixtures of Gaussian distributions etc. Examples of distributions included in class (II) are: Laplace distribution, mean centered Exponential distribution, mean centered Chi-square, centered mixtures of these distributions, amongst several others. The monograph \cite{vershynin2019high} provides a detailed study on the behavior of these classes of distributions. This assumption is significantly weaker than assuming a Gaussian distribution such as that in \cite{wang2018high}, but it requires lighter tail behavior in comparison to \cite{bhattacharjee2019change} who assume a finite fourth moment of the underlying distribution. However, as discussed in Section \ref{sec:intro}, the inference results of \cite{bhattacharjee2018change} and \cite{bhattacharjee2019change} do not extend to the $p>>T$ setting as considered here. Moreover, our results indicate that achieving an optimal rate $O_p(\xi^{-2})$ in the high dimensional setting leads to the rate required of jump signal indeed being influenced by the tail behavior of the underlying distribution (see, (\ref{eq:41})). It is thus expected that an assumption of heavier tails will lead to further stringent requirements on this rate, although we do not pursue this further in this article.

%%%%%%%%%%%%%%%%%%%%%%%%%%%%%%%%%%%%%%%%%%%%%%%%%%%%%%%%%%%%%%%%%%%%%%%%%%%%%%%%%%%%%%%%%%%%%%%%%%%%%%%%%%%%%%%%%%%%%%%%%%%%%%%%%%%%%%%%%%%%%%%%%%%%%%%%%%%%%%%%%%%%%%%%%%%
\vspace{1.5mm}
{\it {{\noi{\bf Condition B (on the covariance structure):}} The covariance $\Sigma:=E\vep_t\vep_t^T$ has bounded eigenvalues, i.e., $0<\ka^2\le\rm{min eigen}(\Si) <\rm{max eigen}(\Si)\le\phi^2<\iny.$}}
%%%%%%%%%%%%%%%%%%%%%%%%%%%%%%%%%%%%%%%%%%%%%%%%%%%%%%%%%%%%%%%%%%%%%%%%%%%%%%%%%%%%%%%%%%%%%%

\vspace{1.5mm}
Condition B assumes a positive definite spatial dependence structure. The assumption of eigenvalues to be bounded by constants is necessary only from an inference perspective. The effect of allowing the maximum eigenvalue of $\Si$ to diverge with $T$ (via $p$) will be the inability of $\tilde\tau$ to yield an optimal rate of convergence $O_p(\xi^{-2}),$ and in turn disallowing the existence of limiting distributions. A more intuitive reasoning for the necessity of this assumption is that allowing $\phi^2$ to diverge may allow the asymptotic variance  $\si^2_{\iny}$ of (\ref{eq:limiting.d.vanishing}) and (\ref{eq:limiting.d.non.vanishing}) to be infinite, causing the corresponding limiting processes to be degenerate, and in turn causing the {\it argmax} of these processes to be undefined. Similarly, allowing $\ka^2$ to converge to zero may lead $\si^2_{\iny}$ to be zero leading to degeneracy of the limiting processes and thus disallowing existence of the limiting distributions.

If the objective is only that of estimation, then these assumptions can be relaxed. In this case, $\ka^2$ may be allowed to converge to zero (or identically zero), i.e., potentially rank deficient. The upper bound $\phi^2$ may be allowed to diverge with $T.$ The bounds for the localization error of $\tilde\tau$ and thereby its rate of convergence provided later in this section are obtained upto universal constants. Consequently the effect of this relaxation will be directly observable in these bounds. Specifically, the rank deficient case will have no impact of the rate of convergence, whereas, the case where $\phi^2$ is allowed to diverge will lead to a deceleration of the rate of convergence of $\tilde\tau.$

We also mention here that Condition A is inherently related to Condition B in that the former by construction imposes the restriction ${\rm maxeigen}\Si=O(\si^2),$ where $\si^2$ is the variance proxy parameter. This can be observed as follows: for all $\delta\in\R^p,$ $\|\delta\|_2=1$ from Definition \ref{def:submult}, we have, $\delta^T\vep_t\sim {\rm subG}(\si^2)$ (or ${\rm subE}(\si^2)$) and thus $\delta^T\Si\delta=E(\delta^T\vep_t)^2\le c_u\si^2.$\footnote{This follows since if $x\sim {\rm subG} (\si^2)$ (or ${\rm subE(\si^2)}$ then $E|x|^k\le 3k\si^kk^{k/2}$ (or $4\si^kk^k$)), $k\ge 1,$ see, e.g. \cite{vershynin2019high}.} Consequently, $\rm{max eigen}(\Si):=\sup_{\|\delta\|_2=1}\delta^T\Si\delta\le c_u\si^2.$ Thus, if the maximum eigenvalue of $\Si$ is allowed to diverge with $T,$ then the variance proxy parameter $\si^2$ must necessarily be allowed to diverge at the same rate. Consequently, without loss of generality, one may assume that $\phi^2$ and $\si^2$ are of the same order, i.e., $\phi^2\asymp \si^2.$ The reason we mention this is because the effect of a diverging $\phi^2$ will be observed  via $\si^2$ in the bounds to be presented later in this section.
%%%%%%%%%%%%%%%%%%%%%%%%%%%%%%%%%%%%%%%%%%%%%%%%%%%%%%%%%%%%%%%%%%%%%%%%%%%%%%%%%%%%%%%%%%%%%%%%%%%%%%%%%%%%%%%%%%%%%%%%%%%%%%%%%%%%%%%%%%%%%%%%%%%%%%%%%%%%%%%%%%%%%%%%%%%%

%%%%%%%%%%%%%%%%%%%%%%%%%%%%%%%%%%%%%%%%%%%%%%%%%%%%%%%%%%%%%%%%%%%%%%%%%%%%%%%%%%%%%%%%%%%%%%%%%%%%%%%%%%%%%%%%%%%%%%%%%%%%%%%%%%%%%%%%%%%%%%%%%%%%%%%%%%%%%%%%%%%%%%%%%%%%
Next consider the following sets of non-zero indices of $\theta_1^0$ and $\theta_2^0,$
\benr\label{def:setS}
S_{1}=\{k;\,\theta^0_{1k}\ne 0\},\quad {\rm and}\quad S_{2}=\{k;\,\theta^0_{2k}\ne 0\},
\eenr
and let $S_1^c$ and $S_2^c$ be the complement sets. Define the maximum cardinality $|S_{1}|\vee |S_{2}|=s\ge 1.$ The parameter $s$ measures sparsity in the model (\ref{model:rvmcp}). To allow the viability of this assumption one may center the observed data with column-wise means, i.e., consider $x_t$ of model (\ref{model:rvmcp}) where instead of the means $\theta_1^0,\theta_2^0$ the jump $\eta^0$ is $s$-sparse, i.e., there is a mean change in at most $s$ components. Upon centering $x_t$ with column-wise empirical means, $x_t^*=x_t-\bar x,$ $t=1,...,T,$ with $\bar x=\sum_{t=1}^T x_t\big/T,$ the sparsity of $\eta^0$ is transferred onto the new mean vectors $\theta_1^*=Ex_t^*,$ $t\le \tau^0,$ and $\theta_2^*=Ex_t^*,$ $t>\tau^0.$ Heuristically, this centering operation is same as that carried out in linear regression models to get rid of the intercept parameter, which is implicitly assumed in the high dimensional linear regression literature. The sparsity assumption is typically made on the jump vector $\eta^0,$ as done in \cite{wang2018high} and \cite{enikeeva2013high}. In contrast we make this assumption directly on the mean vectors $\theta_1^0$ and $\theta_2^0$ and in Appendix \ref{app:equivalence} we show  that this assumption holds without loss of generality with respect to a sparsity assumption on jump vector $\eta^0,$ in context of the problem under consideration.
%%%%%%%%%%%%%%%%%%%%%%%%%%%%%%%%%%%%%%%%%%%%%%%%%%%%%%%%%%%%%%%%%%%%%%%%%%%%%%%%%%%%%%%%%%%%%%%%%%%%%%%%%%%%%%%%%%%%%%%%%%%%%%%%%%%%%%%%%%%%%%%%%%%%%%%%%%%%%%%%%%%%%%%%%%%

The remainder of this section is divided into two subsections. These subsections present near optimal and optimal rates of convergence of $\tilde\tau$ and the sufficient conditions on the nuisance estimates $\h\theta_1$ and $\h\theta_2$ required for the same. As noted in Section \ref{sec:intro}, near optimal rates in themselves do not allow inference. However, these results are still relevant since they shall provide new insight into the distinctions between the sufficient conditions required to achieve optimality over near optimality. Moreover, these shall also serve as a stepping stone in the construction of a feasible methodology to obtain an optimal estimator considered in Section \ref{sec:algorithm}.

\subsection{Near optimal $O_p(\xi^{-2}s\log(p\vee T))$ estimation of $\tau^0$}

We begin with the following condition on the nuisance estimates $\h\theta_1$ and $\h\theta_2.$

%%%%%%%%%%%%%%%%%%%%%%%%%%%%%%%%%%%%%%%%%%%%%%%%%%%%%%%%%%%%%%%%%%%%%%%%%%%%%%%%%%%%%%%%%%%%%%%%%%%%%%%%%%%%%%%%%%%%%%%%%%%%%%%%%%%%%%%%%%%%%%%%%%%%%%%%%%%%%%%%%%%%%%%%%%%
\vspace{1.5mm}
{\it {{\noi{\bf Condition C.1 (on nuisance estimates $\h\theta_1,\h\theta_2$ for near optimality of $\tilde\tau$):}}  Let $\pi_T\to 0$ be a positive sequence and assume that the following two properties hold with probability at least $1-\pi_T.$\\~
		{\rm (I)} The nuisance estimates $\h\theta_1$ and $\h\theta_2$ satisfy $\|(\h\theta_1)_{S_1^c}\|_1\le 3\|(\h\theta_1-\theta_1^0)_{S_1}\|,$ and $\|(\h\theta_2)_{S_2^c}\|_1\le 3\|(\h\theta_2-\theta_2^0)_{S_2}\|,$ where $S_1,$ and $S_2$ are as defined in (\ref{def:setS}).\\~
		{\rm (II)} Assume these nuisance estimates satisfy the following bound in the $\ell_2$ error,
		\benr
		\|\h\theta_1-\theta_1^0\|_2\vee\|\h\theta_2-\theta_2^0\|_2\le c_{u1}\xi,\nn
		\eenr
		where $c_{u1}>0$ is an appropriately chosen small enough constant.}}
%%%%%%%%%%%%%%%%%%%%%%%%%%%%%%%%%%%%%%%%%%%%%%%%%%%%%%%%%%%%%%%%%%%%%%%%%%%%%%%%%%%%%%%%%%%%%%%%%%%%%%%%%%%%%%%%%%%%%%%%%%%%%%%%%%%%%%%%%%%%%%%%%%%%%%%%%%%%%%%%%%%%%%%%%%%

\vspace{1.5mm}
Condition C.1 is an exceptionally weak condition on the quality of nuisance estimates. All it requires is the $\ell_2$ error in the estimation of the mean parameters to be of order of the jump size and may potentially be weaker than assuming ordinary consistency, i.e., an $o_p(1)$ approximation. To see this, consider the case where jump size $\xi$ is bounded below by a constant, then these nuisance estimates are allowed to be inconsistent. Perhaps surprisingly, these nuisance estimates shall still be sufficient for near optimal estimation $O_p\big(s\log(p\vee T)\big)$ of the change point parameter. We can now state our first result which bounds the localization error of $\tilde\tau,$ thereby also yielding a near optimal rate of convergence in both subgaussian and subexponential settings.

%%%%%%%%%%%%%%%%%%%%%%%%%%%%%%%%%%%%%%%%%%%%%%%%%%%%%%%%%%%%%%%%%%%%%%%%%%%%%%%%%%%%%%%%%%%%%%%%%%%%%%%%%%%%%%%%%%%%%%%%%%%%%%%%%%%%%%%%%%%%%%%%%%%%%%%%%%%%%%%%%%%%%%%%%%%
\begin{thm}\label{thm:nearoptimalcp} Suppose the model (\ref{model:rvmcp}) and assume $\tau^0\wedge(T-\tau^0)\ge 0$ and that $\xi>0.$ Additionally assume Condition A(I) (subgaussian), B and C.1 holds. Then for any $T\ge 2,$ and $c_u>2,$ we have,
	\benr
	(i)\,\,\big|\tilde\tau-\tau^0\big|\le 72c_u\si^2\xi^{-2}s\log (p\vee T),\nn
	\eenr
	with probability at least $1-2\exp\{-c_{u1}\log (p\vee T)\}-\pi_T,$ with $c_{u1}=(c_u-2).$ In other words $\si^{-2}\xi^2(\tilde\tau-\tau^0)=O\big(s\log (p\vee T)\big),$ with probability $1-o(1).$ Alternatively, under Condition A(II) (subexponential setting), B and C.1, assuming $T\ge \log (p\vee T),$ and $c_u>8,$ we have,
	\benr
	(ii)\,\,\big|\tilde\tau-\tau^0\big|\le \max\Big\{72c_u\si^2\xi^{-2}s\log (p\vee T),\,\,\log(p\vee T)\Big\}\nn
	\eenr
	with probability at least $1-\exp\big\{-c_{u2}\log (p\vee T)\big\}-\pi_T,$ with  $c_{u2}=(\surd(c_u/2)-2)>0.$ In other words, when $\xi=O(\surd s),$ we have $\si^{-2}\xi^2(\tilde\tau-\tau^0)=O\big(s\log (p\vee T)\big),$ else, we have  $\si^{-2}(\tilde\tau-\tau^0)=O\big(\log (p\vee T)\big),$ both with probability $1-o(1).$
\end{thm}
%%%%%%%%%%%%%%%%%%%%%%%%%%%%%%%%%%%%%%%%%%%%%%%%%%%%%%%%%%%%%%%%%%%%%%%%%%%%%%%%%%%%%%%%%%%%%%%%%%%%%%%%%%%%%%%%%%%%%%%%%%%%%%%%%%%%%%%%%%%%%%%%%%%%%%%%%%%%%%%%%%%%%%%%%%%

Although Theorem \ref{thm:nearoptimalcp} only provides a near optimal rate of convergence and not the optimal rate, it does so under a very mild condition on the relationship between the nuisance estimates and the jump size (Condition C.1).

\begin{rem}\label{rem:near.optimal.rem1} {\rm It may be observed from Theorem \ref{thm:nearoptimalcp} that when $\xi=O(\surd{s})$ and $T\ge \log(p\vee T),$ then under both subgaussian and subexponential cases we have the same rate of convergence of $\tilde\tau,$ i.e. $(\tilde\tau-\tau^0)=O_p\big(\xi^{-2}s\log(p\vee T)\big),$ under the same assumption (Condition C.1) on the nuisance estimates. This illustrates that when only a near optimal rate of convergence is of interest the heavier tail of a subexponential distribution does not influence $\tilde\tau$ in its rate of convergence, or the quality of nuisance estimates required to achieve the same. This shall no longer be true when instead an optimal rate is of interest. Remark \ref{rem:near.optimal} in Section \ref{sec:algorithm} provides further insight in this direction.}
\end{rem}

Another observable consequence of Theorem \ref{thm:nearoptimalcp} is that when $\big[\{s\log (p\vee T)\}^{1/2}\big/\xi\big]\to 0,$ then under the subgaussian case we have perfect identification of the change point parameter in probability, i.e., $pr(\tilde\tau=\tau^0)\to 1.$ However the same cannot be obtained from Theorem \ref{thm:nearoptimalcp} in the subexponential case. This is because under more general conditions than Remark \ref{rem:near.optimal.rem1}, the localization bound under a subexponential distribution is either less precise than its subgaussian counterpart, or alternatively, requires a more rigid condition to match the estimation precision obtained under subgaussianity. This is illustrated in the following result where a slightly weaker bound allows perfect identifiability for the subexponential case.

\begin{thm}\label{thm:subE.nearoptimal.special} Suppose the model (\ref{model:rvmcp}) and assume $\tau^0\wedge(T-\tau^0)\ge 0,$ $\xi>0.$ Additionally assume Condition A(II) (subexponential), B and C.1 hold. Then for $T\ge 2$ and any $c_u>2,$ we have,
	\benr\label{eq:8}
	\big|\tilde\tau-\tau^0\big|\le \big(12c_u\si\big)^2\xi^{-2}s\log^2 (p\vee T)
	\eenr
	with probability at least $1-2\exp\{-c_{u1}\log (p\vee T)\}-\pi_T,$ with $c_{u1}=(c_u-2).$ Consequently, if additionally the jump size is large enough to satisfy $\xi\ge c_{u2} s^{1/2}\log (p\vee T),$ for some $c_{u2}>0.$ Then, we have, (i) $(\tilde\tau-\tau^0)=O(1),$ with probability $1-o(1).$ Furthermore, if the jump size diverges any faster,  i.e., $\big\{s^{1/2}\log (p\vee T)\big/\xi\big\}\to 0,$ then, $pr(\tilde\tau=\tau^0)\to 1.$
\end{thm}

Theorem \ref{thm:subE.nearoptimal.special} is valid when $T\ge 2$ whereas the subexponential case of Theorem \ref{thm:nearoptimalcp} requires $T\ge \log(p\vee T).$ The reason as to why this distinction arises shall have significant consequences on optimal estimation and inference in the context of distinctions between assumptions for subgaussian and subexponential distributions. This is pursued in the following subsection. Simply stated, when the underlying stochastic term comprises of only a finite number of random variables, the heavier tail of the subexponential distribution is realized in the tail bound of this underlying stochastic term, else the behavior is similar to that in the subgaussian case. For intuition purposes, note that when an optimal rate of convergence $|\tilde\tau-\tau^0|\le c\xi^{-2},$ is of interest and $\xi\ge c_u,$ then there are only a finite number of indices between $\tilde\tau$ and $\tau^0.$

\subsection{Optimal $O_p(\xi^{-2})$ estimation of $\tau^0$}\label{subsec:optimal}

This subsection illustrates that $\tilde\tau$ can achieve an optimal rate of convergence, $O_p(\xi^{-2})$ while allowing potentially diminishing jump sizes. The only price one needs to pay to get this advantage is to ensure that the nuisance estimates $\h\theta_1,$ and $\h\theta_2$ are of a higher quality as compared that in the previous subsection. To describe this behavior we begin with a stronger version of Condition C.1 on the nuisance estimates.

%%%%%%%%%%%%%%%%%%%%%%%%%%%%%%%%%%%%%%%%%%%%%%%%%%%%%%%%%%%%%%%%%%%%%%%%%%%%%%%%%%%%%%%%%%%%%%%%%%%%%%%%%%%%%%%%%%%%%%%%%%%%%%%%%%%%%%%%%%%%%%%%%%%%%%%%%%%%%%%%%%%%%%%%%%%
\vspace{1.5mm}
{\it {{\noi{\bf Condition C.2 (assumption on nuisance estimates for optimality of $\tilde\tau$):}}  Let $\pi_T\to 0$ be a positive sequence and assume that either one of the two pairs of properties {\rm (I,II)} or {\rm (I,III)} holds with probability at least $1-\pi_T.$\\~
		{\rm (I)} The nuisance estimates $\h\theta_1$ and $\h\theta_2$ satisfy $\|(\h\theta_1)_{S_1^c}\|_1\le 3\|(\h\theta_1-\theta_1^0)_{S_1}\|,$ and $\|(\h\theta_2)_{S_2^c}\|_1\le 3\|(\h\theta_2-\theta_2^0)_{S_2}\|,$ where $S_1,$ and $S_2$ are as defined in (\ref{def:setS}). \\~
		{\rm (II)} {\bf For the subgaussian case:}  Assume that there exists a sequence $r_T>0,$ such that these nuisance estimates satisfy,
		\benr
		\|\h\theta_1-\theta_1^0\|_2\vee\|\h\theta_2-\theta_2^0\|_2\le r_T\le \frac{c_{u1}\xi}{\{s\log (p\vee T)\}^{1/2}},\nn
		\eenr
		for an appropriately chosen small enough constant $c_{u1}>0.$\\~
		{\rm (III)} {\bf For the subexponential case:} Assume that there exists a sequence $r_T>0,$ such that these nuisance estimates satisfy,
		\benr
		\|\h\theta_1-\theta_1^0\|_2\vee\|\h\theta_2-\theta_2^0\|_2\le r_T\le \frac{c_{u1}\xi}{s^{1/2}\log (p\vee T)},\nn
		\eenr
		for an appropriately chosen small enough constant $c_{u1}>0.$}}
%%%%%%%%%%%%%%%%%%%%%%%%%%%%%%%%%%%%%%%%%%%%%%%%%%%%%%%%%%%%%%%%%%%%%%%%%%%%%%%%%%%%%%%%%%%%%%%%%%%%%%%%%%%%%%%%%%%%%%%%%%%%%%%%%%%%%%%%%%%%%%%%%%%%%%%%%%%%%%%%%%%%%%%%%%%

\vspace{1.5mm}
The only distinction between Condition C.2 and Condition C.1 of the previous subsection is that we have assumed a tighter bound on the nuisance estimates. This tightening has consequences on both the rate of convergence of $\tilde\tau,$ and the assumptions on $s,p$ required for the feasibility of this assumption. These aspects shall be discussed in detail after the following first main result providing an optimal rate of convergence of $\tilde\tau.$

\begin{thm}\label{thm:cpoptimal} Suppose the model (\ref{model:rvmcp}) and assume  $\tau^0\wedge(1-\tau^0)\ge 0,$ $\xi>0,$ and $T\ge 2.$ Additionally assume either one of the following two sets of conditions.\\~
	(a) Suppose Condition A(I) (subgaussian), B and C.2 (I,II) hold.\\~
	(b) Suppose Condition A(II) (subexponential), B and C.2 (I, III) hold.\\~
	Then, for any $0<a<1,$ choosing $c_{a}\ge \surd{(1/a)},$ we have,
	\benr
	|\tilde\tau-\tau^0|\le 36c_a^2\si^2\xi^{-2}\nn
	\eenr
	with probability at least $1-a-2\exp\{-\log(p\vee T)\}-\pi_T.$ Equivalently, we have, $\si^{-2}\xi^{2}(\tilde\tau-\tau^0)=O_p(1).$
\end{thm}

Theorem \ref{thm:cpoptimal} provides the optimal rate of convergence of $\tilde\tau.$ A first look on the sufficient conditions required for this result may lead one to suspect that Theorem \ref{thm:cpoptimal} provides an optimal bound without any rate conditions on the model parameters $s,p,\xi,l_T.$ This is indeed true only in a very special case but false in general, as discussed in the following.

Consider the case where mean parameters $\theta_1^0$ and $\theta_2^0$ are known. Here, setting $\h\theta_1=\theta_1^0$ and $\h\theta_2=\theta_2^0$ allows Condition C.2 to be trivially satisfied irrespective of the rate of divergence of $s,p.$ Consequently, even if $\tau^0$ is at a boundary ($\tau^0=0\,{\rm or}\, T$) and the dimensions $s,p$ are diverging arbitrarily fast,
$\tilde\tau$ will still estimate $\tau^0$ at an optimal rate. The only assumption required for this case is $\xi>0,$ i.e. $\theta_1^0\ne\theta_2^0.$ This case is clearly infeasible in practice and is only discussed to provide the following perhaps surprising insight. The estimation of a change point in itself does not require many assumptions that are usually thought of as necessary in the literature, including separation from boundary, minimum jump size and restrictions on dimensionality. These assumptions instead arise solely from the nuisance estimation aspect of the overall process.

In the more realistic setup where $\theta_1^0$ and $\theta_2^0$ are unknown, the key in Theorem \ref{thm:cpoptimal} is Condition C.2. Effectively, the use of Condition C.2 has passed the burden of assumptions on model parameters to the nuisance estimates $\h\theta_1$ and $\h\theta_2.$ To discuss this further we require the following boundary condition on $\tau^0.$

%%%%%%%%%%%%%%%%%%%%%%%%%%%%%%%%%%%%%%%%%%%%%%%%%%%%%%%%%%%%%%%%%%%%%%%%%%%%%%%%%%%%%%%%%%%%%%%%%%%%%%%%%%%%%%%%%%%%%%%%%%%%%%%%%%%%%%%%%%%%%%%%%%%%%%%%%%%%%%%%%%%%%%%%%%
\vspace{1.5mm}
{\it {{\noi{\bf Condition D (on separation of $\tau^0$ from its parametric boundary):}} Assume the existence of a change point $\tau^0$ for the model (\ref{model:rvmcp}), i.e., it satisfies $\tau^0\wedge(T-\tau^0)\ge Tl_T\ge 1,$ for some positive sequence $l_T\to 0.$ }}
%%%%%%%%%%%%%%%%%%%%%%%%%%%%%%%%%%%%%%%%%%%%%%%%%%%%%%%%%%%%%%%%%%%%%%%%%%%%%%%%%%%%%%%%%%%%%%%%%%%%%%%%%%%%%%%%%%%%%%%%%%%%%%%%%%%%%%%%%%%%%%%%%%%%%%%%%%%%%%%%%%%%%%%%%%%

\vspace{1.5mm}
Clearly, all this condition requires is at least one realization from both of the two distributions characterizing model (\ref{model:rvmcp}) and is usually implicit in the literature. Under Condition D, one can obtain regularized mean estimates $\h\theta_1,$ and $\h\theta_2,$ that satisfy at best the bound (see, Section \ref{sec:algorithm}),
\benr\label{eq:15}
\|\h\theta_1-\theta_1^0\|_2\vee\|\h\theta_2-\theta_2^0\|_2\le r_T=c_u\si\Big(\frac{s\log (p\vee T)}{Tl_T}\Big)^{\frac{1}{2}},
\eenr
with probability $1-o(1).$ Now comparing (\ref{eq:15}) with Condition C.2 (II) and C.2 (III) for the subgaussian and subexponential cases, respectively, yields the following requirements that must be satisfied for Condition C.2 to be feasible and in turn Theorem \ref{thm:cpoptimal} to remain valid,
\benr
&&\Big(\frac{c_u\si}{\xi}\Big)\Big\{\frac{s\log (p\vee T)}{\surd(Tl_T)}\Big\}\le c_{u1},\quad{\rm for\,\,the\,\,subgaussian\,\, setting,}\label{eq:16}\\
&&\Big(\frac{c_u\si}{\xi}\Big)\Big\{\frac{s\log^{3/2} (p\vee T)}{\surd (Tl_T)}\Big\}\le c_{u1},\quad{\rm for\,\,the\,\,subexponential\,\, setting,}\label{eq:17}
\eenr
for a suitably chosen small enough constant $c_{u1}>0.$ Relations (\ref{eq:16}) and (\ref{eq:17}) describes interplay between model parameters $\xi,l_T,s,p,T$ \big(which are all sequences in $T$\big) and the underlying class of distribution, that are then sufficient for the estimator $\tilde\tau$ to achieve the optimal rate of convergence $O_p(\xi^{-2}).$

As a direct consequence of Theorem \ref{thm:cpoptimal} one may observe that when $\xi\to \iny,$ then the estimator $\tilde\tau$ perfectly identifies the change point parameter $\tau^0,$ in probability, i.e.,  the limiting distribution of $\tilde\tau$ in this case is degenerate. While this perfect identifiability under a diverging jump size is also provided by the grid search least squares estimator as studied in \cite{bai2010common} and \cite{bhattacharjee2019change}, however the assumption made here on the diverging jump size is weaker given high dimensionality. This can be observed in the rate at which $\xi$ is required to diverge, for e.g. for the same result to hold true in \cite{bhattacharjee2019change} one requires $\xi\to \iny$ at a fast enough rate so that additionally $\xi^{-1}\surd{p/T}\to 0$ is satisfied. The assumption required in \cite{bai2010common} to achieve the same perfect identifiability in the high dimensional setting is more stringent than that of \cite{bhattacharjee2019change}.

From an estimation perspective, the optimal rate $O_p(\xi^{-2})$ of $\tilde\tau$ may not seem a significant improvement in comparison to near optimal rates available in the literature for estimators in the high dimensional setting, for e.g. the projected cusum estimator of \cite{wang2018high} with a presented rate of $O_p(\xi^{-2}\log \log T),$ thus the improvement offered by $\tilde\tau$ being only of order $\log\log T.$ However, this slight improvement is critical from an inference perspective, it is only the availability of an $O_p(\xi^{-2})$ rate that allows the existence of a limiting distribution.

The above discussion also highlights that in order for $\tilde\tau$ to have a non-degenerate limiting distribution in the high dimensional setting, conditions (\ref{eq:16}) and (\ref{eq:17}) must allow $\xi\le c_u,$ despite high dimensionality and while preserving the optimal rate of convergence $O_p(\xi^{-2})$ presented in Theorem \ref{thm:cpoptimal}. This feasibility is summarized in the following corollary.

\begin{cor}\label{cor:cpoptimal.diminishing.jump} Suppose the model (\ref{model:rvmcp}) and assume one of the following two sets of conditions.\\~
	(a) Condition A(I) (subgaussian), B and C.2 (I,II) hold with $r_T=c_u\si\big\{s\log (p\vee T)\big/{Tl_T}\big\}^{1/2}.$ Additionally assume $s\log (p\vee T)=o\big(\surd(Tl_T)\big),$ and $\xi$ (potentially diminishing) satisfies (\ref{eq:16}). \\~
	(b) Condition A(II) (subexponential), B and C.2 (I,III) hold with $r_T=c_u\si\big\{s\log (p\vee T)\big/{Tl_T}\big\}^{1/2}.$ Additionally assume $s\log^{3/2} (p\vee T)=o\big(\surd(Tl_T)\big),$ and $\xi$ (potentially diminishing) satisfies (\ref{eq:17}).\\~
	Then, we have, $\si^{-2}\xi^{2}(\tilde\tau-\tau^0)=O_p(1).$
\end{cor}

We conclude this section with another perspective on the discussion in this subsection. Recall that the construction of $\tilde\tau$ utilizes $p$-dimensional nuisance estimates whose rate of convergence involve the dimensional parameters $s,p$ (see, \ref{eq:15}). However the rate of convergence of the change point estimator itself is $O_p(\xi^{-2}),$ which is free of dimensionality parameters $s,p,$ the sampling period $T$ and is valid despite high dimensionality and a potentially diminishing jump size. This alludes towards the estimator $\tilde\tau$ behaving as if the nuisance parameter vectors utilized in its construction are known. This property of an estimator is typically referred to as adaptation as described in \cite{bickel1982adaptive}, but is observed here in a high dimensional sense and in the context of change point estimation. In the fixed $p$ setting, this property of a change point estimator behaving as if the nuisance parameters are known has also been studied in \cite{hinkley1972time}. There are also more recent precedent's to similar behavior but in the context of inference on regression coefficients in high dimensional linear regression models. For estimation of a component of the regression vector, where certain corrections to the least squares loss or its first order moment equations, such as debiasing (\cite{van2014asymptotically}) or orthogonalization (\cite{belloni2011inference},   \cite{chernozhukov2015valid}, \cite{belloni2017confidence} and \cite{ning2017general}) induce sufficient adaptivity against nuisance estimates and thereby allow optimal estimation of the target regression parameter. The results of this subsection show that in the context of change point estimation, the plugin least squares estimator (\ref{est:optimalcp}) itself possesses adaptivity against potentially high dimensional nuisance estimates, in order to allow for $O_p(\xi^{-2})$ estimation of the change point $\tau^0,$ provided the nuisance parameters are estimated with sufficient precision. This adaptation shall become further visible in the following section where limiting distributions of $\tilde\tau$ are established.

\section{Limiting distributions of $\tilde\tau$ in vanishing and non-vanishing jump size regimes}\label{sec:inference}

This section investigates the asymptotic distributional properties of $\tilde\tau.$ Critically, here asymptotics are in a high dimensional sense where $s,p$ are allowed to be fixed or diverge with $T,$ with $p$ diverging potentially exponentially with $T.$ As noted before, the case of $\xi\to\iny$ yields a degenerate limiting distribution of $\tilde\tau.$ Thus, in what follows we restrict our analysis to $\xi\le c_u,$ where the limiting distribution of $\tilde\tau$ is non-trivial. This case is further subdivided into two distinct regimes described in the following condition.

%%%%%%%%%%%%%%%%%%%%%%%%%%%%%%%%%%%%%%%%%%%%%%%%%%%%%%%%%%%%%%%%%%%%%%%%%%%%%%%%%%%%%%%%%%%%%%%%%%%%%%%%%%%%%%%%%%%%%%%%%%%%%%%%%%%%%%%%%%%%%%%%%%%%%%%%%%%%%%%%%%%%%%%%%%%
\vspace{1.5mm}
{\it {{\noi{\bf Condition E (on the jump size for stability of limiting distributions):}} Assume that the jump size is bounded above, i.e, $0<\xi\le c_u.$ Let $\Si$ and $\eta^0$ be as defined in Condition B and (\ref{def:jumpsize}), respectively, and additionally assume that either one of the following two conditions hold. \\~
		(i) (vanishing jump)  Let $\xi\to 0$ and $\xi^{-2}\big(\eta^{0T}\Si\eta^0\big)\to \si^2_{\iny},$ for some $0<\si_{\iny}^2<0.$\\~	
		(ii) (non-vanishing jump) Let $\xi\to \xi_{\iny},$ and $\big(\eta^{0T}\Si\eta^0\big)\to \xi_{\iny}^2\si^2_{\iny},$  for some $0<\xi_{\iny},\,\si_{\iny}^2<\iny.$}}
%%%%%%%%%%%%%%%%%%%%%%%%%%%%%%%%%%%%%%%%%%%%%%%%%%%%%%%%%%%%%%%%%%%%%%%%%%%%%%%%%%%%%%%%%%%%%%%%%%%%%%%%%%%%%%%%%%%%%%%%%%%%%%%%%%%%%%%%%%%%%%%%%%%%%%%%%%%%%%%%%%%%%%%%%%%

\vspace{1.5mm}
The existence of the deterministic limit assumed in Condition E(i) and E(ii) is a mild assumption since Condition B already guarantees that the sequences under consideration are bounded above and below, i.e., we have, $0<\ka^2\xi^2\le \eta^{0T}\Si\eta^0\le \phi^2\xi^2<\iny.$ This limit measures the variance of the underlying limiting process which then characterizes the distribution of $\tilde\tau,$ thus the need for an assumption of its existence.

The vanishing and non-vanishing jump size regimes described in Condition E play a fundamental role in the distributional behavior of a change point estimator. The reason for this inherent characteristic can be directly observed by noting that the stochastic term that controls the change point estimator $\tilde\tau$ has a distribution of the form $\sum_{t=1}^{\z\xi^{-2}}u_t^T\eta^0,$ where $u_t\sim i.i.d(0,\Si),$ and constant $\z>0.$ The regime $\xi\to 0,$ enables the $\z\xi^{-2}\to \iny$ and thus upon suitable normalization allows the functional central limit theorem to kick in, and yield a Brownian motion as the resulting process over $\z.$ This neat property has been exploited in the classical literature under fixed dimension $(p)$ to obtain distributional results under this vanishing jump size regime, see, e.g. \cite{bai1994}. Unfortunately, when $\xi\not\to 0,$ the stochastic term described earlier is no longer an infinite sum, and it is clear that the Brownian motion approximation is no longer feasible. Infact it is also observable that any distributional result under this non-vanishing case will necessitate a further parametric assumption on the underlying distribution, since in this case the stochastic term under consideration is a finite sum.

The first result below considers the vanishing case $\xi\to 0.$ It obtains the limiting distribution of $\tilde\tau$ as the distribution of the {\it argmax} of a symmetric two sided Brownian motion with a negative drift, under suitable conditions on the quality of the nuisance estimates used in the construction of $\tilde\tau.$

\begin{thm}[Limiting distribution under vanishing jump regime]\label{thm:wc.vanishing} Suppose Condition A, B, D and E(i) hold and assume that the sequence $l_T$ of Condition D satisfies $T\l_T\to\iny.$ Let the mean parameters $\theta_1^0$ and $\theta_2^0$ be known and let $\tilde\tau^*=\tilde\tau(\theta_1^0,\theta_2^0).$ Then, we have,
	\benr\label{eq:wc.vanishing}
	\xi^{2}(\tilde\tau^*-\tau^0)\Rightarrow \argmax_{\z\in\R}\big\{2\si_{\iny}W(\z)-|\z|\},
	\eenr
	where $W(\z)$ is a two sided Brownian motion\footnote{A two-sided Brownian motion $W(\z)$ is defined as $W(0) = 0,$ $W(\z) = W_1(\z),$ $\z > 0$ and $W(\z) = W_2(-\z),$ $\z < 0,$ where $W_1(\z)$ and $W_2(\z)$ are two independent Brownian motions defined on the non-negative half real line}. Alternatively, when $\theta_1^0$ and $\theta_2^0$ are unknown, suppose $\tilde\tau=\tilde\tau(\h\theta_1,\h\theta_2),$ where the estimates $\h\theta_1$ and $\h\theta_2$ satisfy   Condition C.2. Additionally assume that the sequence $r_T$ of Condition C.2 satisfies,
	\benr\label{eq:30}
	r_T=\begin{cases}\frac{o(1)\xi}{\{s\log(p\vee T)\}^{1/2}}, & {\rm under\,\,subgaussian\,\,case},\\
		\frac{o(1)\xi}{s^{1/2}\log(p\vee T)}, & {\rm under\,\,subexponential\,\,case}.\end{cases}
	\eenr
	Then, the convergence in distribution (\ref{eq:wc.vanishing}) also holds when $\tilde\tau^*$ is replaced with $\tilde\tau.$
\end{thm}

Following are observations regarding the sufficient conditions required for this result and comparisons with those in Theorem \ref{thm:cpoptimal} which provides the optimal rate of convergence. As before, the burden of rate assumptions on $s,p$ and $\xi,$ have been passed onto Condition C.2 and additionally here the requirement (\ref{eq:30}), which in turn requires an inter-relationship between $s,p,T,\xi,l_T$ to be satisfied similar to as discussed before in (\ref{eq:16}) and (\ref{eq:17}). In this case however, condition (\ref{eq:30}) forces a marginally stronger requirement, specifically, comparing the desired rate of $r_T$ in condition (\ref{eq:30}) to the best attainable rate (\ref{eq:15}) of mean estimation under high dimensionality yields,
\benr\label{eq:31}
\frac{s\log (p\vee T)}{\xi\surd(Tl_T)}= o(1),\,\,\,{\rm or}\,\,\, \frac{s\log^{3/2} (p\vee T)}{\xi\surd (Tl_T)}= o(1),
\eenr
for the subgaussian and subexponential cases, respectively. Comparing (\ref{eq:31}) to the requirements (\ref{eq:16}) and (\ref{eq:17}) we observe that the additional assumption made here is only to tighten the rate restriction to $o(1)$ from $O(1).$ This illustrates the price paid in order to obtain the limiting distribution in comparison to only optimal rate of estimation. This tighter restriction is also in coherence with classical results in the fixed $p$ setting, where the condition reduces to only a relationship between $\xi,\,T$ and $l_T.$ Additionally, since here we are restricted by the regime $\xi\to 0$ under consideration, consequently these sufficient conditions must be further restricted as,
\benr\label{eq:39}
\frac{s\log (p\vee T)}{\surd(Tl_T)}= o(1),\,\,\,{\rm or}\,\,\, \frac{s\log^{3/2} (p\vee T)}{\surd (Tl_T)}= o(1),
\eenr
for the subgaussian and subexponential cases, respectively.

Another slightly stronger assumption made here in comparison to Theorem \ref{thm:cpoptimal} is on sequence $l_T.$ While the result of Theorem \ref{thm:cpoptimal} is valid without the actual existence of the change point, i.e. $\tau^0\wedge (T-\tau^0)\ge 0,$ the limiting distribution of Theorem \ref{thm:wc.vanishing} assumes that the change point exists and is separated from the boundaries of its parametric space, i.e., $\tau^0\wedge(T-\tau^0)\ge Tl_T\to \iny.$ This additional assumption is required in order to allow both ends of the two sided random walk to stabilize to the given Brownian motion process.

It can be observed that a change of variable to $\z=\si_{\iny}^2\z',$ yields that $\argmax_{\z\in\R}\big\{2\si_{\iny}W(\z)-|\z|\}=^d\si_{\iny}^2\argmax_{\z'\in\R}\big\{2W(\z')-|\z'|\},$ which in turn yields the relation (\ref{eq:limiting.d.vanishing}) provided in Section \ref{sec:intro}. This distribution is well studied in the literature and its cdf was first provided by \cite{yao1987approximating}, which enables computation of quantiles and in turn an asymptotically valid confidence interval for $\tau^0.$

We now proceed to the non-vanishing regime of Condition E(ii). The literature on distributional properties of $\tilde\tau$ in this case is quite sparse. Even under the classical fixed $p$ setting, a comprehensive understanding on the same remains unfulfilled. In this context, the articles \cite{jandhyala1999capturing} and \cite{fotopoulos2010exact} provide generalized results on the distribution of the maximum likelihood change point estimators. These results provide key connections of the desired limiting distribution to a two sided random walk. However, the results require mean parameters to be known and constant, i.e. where the sequence $\xi$ is assumed to be constant over the sampling period $T,$ consequently also requiring dimension $p$ to fixed. To the best of our knowledge, the only results in the literature that discuss this non-vanishing regime in a diverging $p$ setting are those of \cite{bhattacharjee2018change} and \cite{bhattacharjee2019change}. However, these result are also limited to $p<<T,$ i.e, $p$ is diverging slower than $T.$ The second main result of this section provides this limiting distribution for the estimator $\tilde\tau,$ valid under both fixed $p$ and high dimensional asymptotic. Moreover, it does not require underlying mean parameters to be known apriori.

Recall from the earlier discussion on Condition E that under this non vanishing regime, the stochastic term controlling the distributional properties of the change point estimator is a finite sum (in $t$), and with finite variance for each random variable in this sum. This disallows the use of central limit theorems and thereby makes Gaussian approximations of the limiting process infeasible (without exact normality assumptions on the data generating process). It is due to this reason that the analysis of this regime requires further parametric assumptions on the underlying distribution which are stated in the following.

%%%%%%%%%%%%%%%%%%%%%%%%%%%%%%%%%%%%%%%%%%%%%%%%%%%%%%%%%%%%%%%%%%%%%%%%%%%%%%%%%%%%%%%%%%%%%%%%%%%%%%%%%%%%%%%%%%%%%%%%%%%%%%%%%%%%%%%%%%%%%%%%%%%%%%%%%%%%%%%%%%%%%%%%%%%
\vspace{1.5mm}
{\it {{\noi{\bf Condition A$'$ (additional distributional assumptions):}} Suppose Condition A holds and additionally assume for any constants $c_1,c_2\in\R,$ the r.v.'s $c_2\vep^T_t\eta^0+c_1\sim^{i.i.d} \cL\big(c_1,c_2^2\eta^{0T}\Si\eta^0\big),$ $t=1,...,T,$ for some  distribution $\cL,$ which is continuous and supported in $\R.$ }}
%%%%%%%%%%%%%%%%%%%%%%%%%%%%%%%%%%%%%%%%%%%%%%%%%%%%%%%%%%%%%%%%%%%%%%%%%%%%%%%%%%%%%%%%%%%%%%%%%%%%%%%%%%%%%%%%%%%%%%%%%%%%%%%%%%%%%%%%%%%%%%%%%%%%%%%%%%%%%%%%%%%%%%%%%%%

\vspace{1.5mm}
The arguments in the notation $\cL(\mu,\si^2)$ are used to represent the mean and variance of the distribution $\cL,$ i.e,  $E\cL(\mu,\si^2)=\mu,$ and ${\rm var}\big(\cL(\mu,\si^2)\big)=\si^2.$ Note that the mean and variance parameters of the distribution $\cL$ are notated only to present the limiting distribution result to follow, this notation does not necessarily imply that $\cL$ is characterized by only its mean and variance.

The additional assumptions made in Condition A$'$ over those in Condition A are that of assuming an explicit form $\cL$ of the underlying distribution and assuming this distribution to be continuous. The  requirement of this distribution being supported in $\R$ is also implicitly assumed in Condition A. To proceed further we require the following stochastic process that shall serve to characterize the desired limiting distribution of the change point estimator in the current non-vanishing regime. Let $\Np=\{1,2,....\}$ and $\Nn=\{-1,-2,....\},$ and define the following negative drift two-sided random walk initializing at the origin,
\benr\label{def:cCz}
\cC_{\iny}(\z)=
\begin{cases}\sum_{t=1}^{\z} z_t, & \z\in \Np \\ 	
	0,				  &	\z=0 \\
	\sum_{t=1}^{-\z}z_t^*,		  &	\z\in \Nn,
\end{cases}
\eenr
Here $z_t,z_t^*$ are independent copies of a $\cL(-\xi_{\iny}^2,4\xi_{\iny}^2\si^2_{\iny})$ distribution, which are also independent over all $t.$ The parameters $\xi_{\iny}$ and $\si_{\iny}^2$ are defined in Condition E(ii). In the case of unit variances and spatial uncorrelated-ness of the data generating process, where $\Si=E\vep_t\vep_t^T=I_{p\times p},$ we have $\si^2_{\iny}=1$ and consequently $z_t,z_t^*\sim^{i.i.d}\cL(-\xi^2_{\iny},4\xi^2_{\iny}).$ Under these notations we can now state the second main result of this section.

\begin{thm}[Limiting distribution under non-vanishing jump regime]\label{thm:wc.non.vanishing} Suppose Condition A$'$, B, D and E(ii) hold and assume that $l_T$ of Condition D satisfies $T\l_T\to\iny.$ Let the mean parameters $\theta_1^0$ and $\theta_2^0$ be known and let $\tilde\tau^*=\tilde\tau(\theta_1^0,\theta_2^0).$ Then
	\benr\label{eq:wc.non.vanishing}
	(\tilde\tau^*-\tau^0)\Rightarrow \argmax_{\z\in \Z}\cC_{\iny}(\z),
	\eenr
	where $\cC_{\iny}(\z)$ is as defined in (\ref{def:cCz}). Alternatively, when $\theta_1^0$ and $\theta_2^0$ are unknown, suppose $\tilde\tau=\tilde\tau(\h\theta_1,\h\theta_2),$ where estimates $\h\theta_1$ and $\h\theta_2$ satisfy Condition C.2. Additionally assume the sequence $r_T$ of Condition C.2 satisfies (\ref{eq:30}). Then, the convergence in distribution (\ref{eq:wc.non.vanishing}) also holds when $\tilde\tau^*$ is replaced with $\tilde\tau.$
\end{thm}

The map $\argmax_{\z\in \Z}\cC_{\iny}(\z)$ is a.s. unique and possesses a distribution supported on $\Z.$ This has been shown in the proof of Theorem \ref{thm:wc.non.vanishing}, although it is also quite intuitive upon observing that the two sided random walk $\cC_{\iny}(\z)$ is negative drift with  continuously distributed increments, which in turn implies that $\max_{\z\in \Z}\cC_{\iny}(\z)$ is supported on $[0,\iny),$ where it is continuously distributed on $(0,\iny)$ and has an additional probability mass at the singleton zero.

The sole distinction between the assumptions of Theorem \ref{thm:wc.vanishing} and Theorem \ref{thm:wc.non.vanishing} is the change of regime from a vanishing jump size (Condition E(i)) to the non-vanishing jump regime (Condition E(ii)), respectively. Consequently, the observations made in the discussion after Theorem \ref{thm:wc.vanishing} on the inter-relationship between the quality of nuisance estimates, the dimensional parameters $s,p$ and the jump size $\xi,$ retain their validity under this non-vanishing regime as well. In particular, these inter-related requirements instead can be replaced with the rate restrictions (\ref{eq:31}) and in turn (\ref{eq:39}), while maintaining the validity of Theorem \ref{thm:wc.non.vanishing}. Since the analytical form of the distribution $\argmax_{\z\in\Z}\cC_{\iny}(\z)$ is unavailable, one may resort to obtaining quantiles of this distribution via a monte-carlo simulation, i.e., simulating the two sided random walk process and in turn obtaining realizations from the distribution under consideration.

\section{Construction of a feasible $O_p(\xi^{-2})$ estimator of $\tau^0$}\label{sec:algorithm}

The results of Section \ref{sec:mainresults} and Section \ref{sec:inference} allow $ \tilde\tau$ to provide an $O_p(\xi^{-2})$ approximation of $\tau^0,$ and provide limiting distributions to perform inference on the unknown change point. However, these results rely on the apriori availability of nuisance estimates $\h\theta_1,$ and $\h\theta_2,$ satisfying Condition C.2. In this section we develop an algorithmic estimator to obtain these nuisance estimates that are theoretically guaranteed to satisfy Condition C.2, which in turn shall yield a feasible  $O_p(\xi^{-2})$ estimate of the change point parameter.

To proceed further we require more notation. For any $\tau\in\{1,...,(T-1)\},$ let
\benr\label{def:empmeans}
\bar x_{(0:\tau]}=\frac{1}{\tau}\sum_{t=1}^{\tau} x_t,\quad{\rm and}\quad \bar x_{(\tau:T]}=\frac{1}{(T- \tau)}\sum_{t=\tau+1}^{T} x_t,
\eenr
be the piece-wise sample means. Consider the soft-thresholding operator, $k_{\la}(x)={\rm sign}(x)(|x|-\la)_{+},$ $\la>0,$ $x\in\R^p,$ where ${\rm sign}(\cdotp),$ $|\cdotp|,$ and $(\cdotp)_{+}$\footnote{For $x\in \R,$ $(x)_{+}=x,$ if $x\ge 0,$ and $x=0$ if $x<0.$} are applied component-wise. Then for any $\la_1,\la_2>0,$ define $\ell_1$ regularized mean estimates,
\benr\label{est:softthresh}
\tilde\theta_1(\tau)=k_{\la_1}\big(\bar x_{(0:\tau]}\big),\quad{\rm and}\quad \tilde\theta_2(\tau)=k_{\la_2}\big(\bar x_{(\tau:T]}\big),
\eenr
It is well known in the literature (\cite{donoho1995noising}, \cite{donoho1995wavelet}) that the soft-thresholding operation in (\ref{est:softthresh}) is equivalent to the following $\ell_1$ regularization.
\benr\label{est:softL1construction}
\tilde\theta_1(\tau)&=&\argmin_{\theta\in\R^p}\big\|\bar x_{(0:\tau]}-\theta\big\|^2_2+\la_1\|\theta\|_1,
\eenr
and similar for $\tilde\theta_2(\tau).$

In order to develop a feasible estimator for $\tau^0,$ recall the following two aspects from Section \ref{sec:mainresults}. (a) The missing links required to implement the estimator of Section \ref{sec:mainresults} are the nuisance (mean) estimates. (b) These mean estimates require either Condition C.1 (milder) to obtain a near optimal estimate or Condition C.2 (stronger) to obtain an optimal estimate of $\tau^0.$ We shall fulfill requirement (a) using soft thresholded means (\ref{est:softthresh}), and furthermore utilize the distinctions between Condition C.1 and Condition C.2 to build an algorithmic estimator that improves a nearly arbitrarily chosen $\check\tau,$ first to a near optimal estimate $\h\tau$ in a first iteration, and then to an optimal estimate $\tilde\tau$ in a second iteration. We remind the reader here that the specific choice of soft-thresholding as a regularization mechanism on the empirical means is superficial, the eventual objective is only to obtain mean estimates that are well behaved in the high dimensional setting in the $\ell_2$ norm (see, (\ref{eq:15})). Alternatively, one may consider using any suitable choice of the regularization mechanism that may also be problem specific, e.g. group $\ell_1$ regularization which assumes a partially known sparsity structure.

The stepwise approach of the estimator to be considered is as follows. Condition C.1 is weak enough that it is satisfied by the estimates $\check\theta_1=\tilde\theta(\check\tau)$ and $\check\theta_2=\tilde\theta_2(\check\tau)$ of (\ref{est:softthresh}), computed with any nearly arbitrarily chosen $\check\tau\in\{1,...,(T-1)\}$ that is marginally away from its boundaries. (see, Condition F below). Thus, Theorem \ref{thm:nearoptimalcp} and Theorem \ref{thm:subE.nearoptimal.special} now guarantee the update $\h\tau=\tilde\tau \big(\check\theta_1,\check\theta_2\big)$ of (\ref{est:optimalcp}) computed using these mean estimates $\check\theta_1,\check\theta_2,$ shall be a near optimal estimate of $\tau^0.$ With the availability of this near optimal estimate $\h\tau,$ it can be shown that the updates $\h\theta_1=\tilde\theta_1(\h\tau),$ and $\h\theta_2=\tilde\theta_2(\h\tau),$ satisfy Condition C.2. This allows us to perform another update $\breve\tau=\tilde\tau(\h\theta_1,\h\theta_2),$ and Theorem \ref{thm:cpoptimal} now guarantees optimality of $\breve\tau.$ Thus, in performing these updates (two each of the change point and the mean) we have taken a $\check\tau$ from a nearly arbitrary neighborhood of $\tau^0,$ and deposited it in an optimal neighborhood of $\tau^0,$ with an intermediate $\h\tau$ that lies in a near optimal neighborhood. This process is stated as Algorithm 1 below and is presented visually in Figure \ref{fig:schematic}.

%%%%%%%%%%%%%%%%%%%%%%%%%%%%%%%%%%%%%%%%%%%%%%%%%%%%%%%%%%%%%%%%%%%%%%%%%%%%%%%%%%%%%%%%%%%%%%%%%%%%%%%%%%%%%%%%%%%%%%%%%%%%%%%%%%%%%%%%%%%%%%%%%%%%%%%%%%%%%%%%%%%%%%%%%%%
\begin{figure}[]
	\centering{
		\resizebox{1\textwidth}{!}{
			\begin{tikzpicture}[node distance = 3.75cm,auto]
			% Place nodes
			\node [cloud] (ctau) {$\check\tau$};
			\node [block,right of=ctau] (ctheta) {$\check\theta_1=\tilde\theta_1(\check\tau)$\\ $\check\theta_2=\tilde\theta_2(\check\tau)$};
			\node [cloud, right of=ctheta] (htau) {$\h\tau=\tilde\tau(\check\theta_1, \check\theta_2)$};
			\node [block,right of=htau] (htheta) {$\h\theta_1=\tilde\theta_1(\h\tau),$ $\h\theta_2=\tilde\theta_2(\h\tau)$};
			\node [cloud,right of=htheta] (btau) {$\breve\tau=\tilde\tau(\h\theta_1, \h\theta_2)$};
			\node [block, below of=ctau, node distance=2.5cm] (arb) {\footnotesize{Condition F satisfied (nearly arbitrary choice)}};
			\node [block, below of=ctheta, node distance=2.5cm] (c1) {\footnotesize{Condition C.1 satisfied}};
			\node [block, below of=htau, node distance=2.5cm] (nopt) {\footnotesize{Near\\ optimal}};
			\node [block, below of=btau, node distance=2.5cm] (opt) {\footnotesize{Optimal}};
			\node [block, below of=htheta, node distance=2.5cm] (c2) {\footnotesize{Condition C.2 satisfied}};
			%\node [block, below right of=arb] (ts) {\footnotesize{To be shown}};
			% Draw edges
			\path [line] (ctau) -- (ctheta);
			\path [line] (ctheta) -- (htau);
			\path [line] (htau) -- (htheta);
			\path [line] (htheta) -- (btau);
			\path [linena,dashed] (ctau) -- (arb);
			\path [linena,dashed] (htau) -- (nopt);
			\path [linena,dashed] (btau) -- (opt);
			\path [linena,dashed] (ctheta) -- (c1);
			\path [linena,dashed] (htheta) -- (c2);
			\path[line](arb)  --  (c1);
			\path[line](c1)  -- (nopt);
			\path[line](nopt)  --  (c2);
			\path[line](c2)  --  (opt);
			\end{tikzpicture}}
		\caption{\footnotesize{A schematic of the underlying working mechanism of Algorithm 1.}}
		\label{fig:schematic}}
\end{figure}
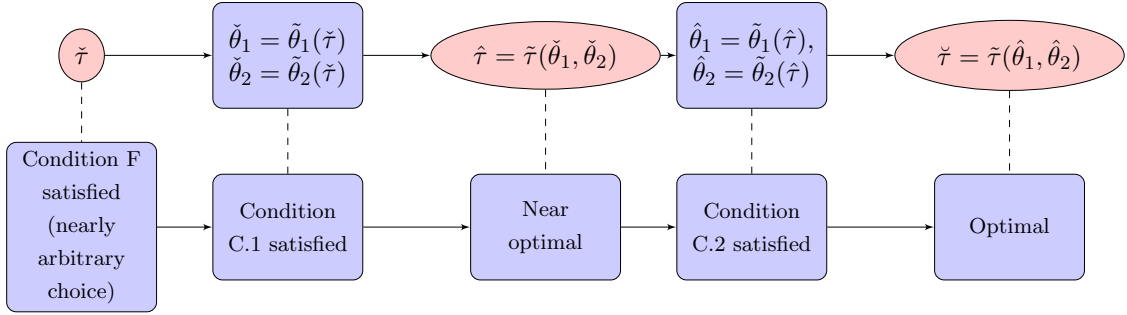
%%%%%%%%%%%%%%%%%%%%%%%%%%%%%%%%%%%%%%%%%%%%%%%%%%%%%%%%%%%%%%%%%%%%%%%%%%%%%%%%%%%%%%%%%%%%%%%%%%%%%%%%%%%%%%%%%%%%%%%%%%%%%%%%%%%%%%%%%%%%%%%%%%%%%%%%%%%%%%%%%%%%%%%%%%%

%%%%%%%%%%%%%%%%%%%%%%%%%%%%%%%%%%%%%%%%%%%%%%%%%%%%%%%%%%%%%%%%%%%%%%%%%%%%%%%%%%%%%%%%%%%%%%%%%%%%%%%%%%%%%%%%%%%%%%%%%%%%%%%%%%%%%%%%%%%%%%%%%%%%%%%%%%%%%%%%%%%%%%%%%%%
\begin{figure}[]
	\noi\rule{\textwidth}{0.5pt}
	
	\vspace{-2mm}
	\flushleft {\bf Algorithm 1:} Optimal estimation of $\tau^0$
	
	\vspace{-1.25mm}
	\noi\rule{\textwidth}{0.5pt}
	
	\vspace{-1.25mm}
	\flushleft{\bf (Initialize):} Choose any $\check\tau\in\{1,...,(T-1)\}$ satisfying Condition F.
	
	\vspace{-1.25mm}
	\flushleft{\bf Step 1:} Obtain estimates $\check\theta_1=\tilde\theta_1(\check\tau),$ and $\check\theta_2=\tilde\theta_2(\check\tau),$ and update change point estimate as,
	
	\vspace{-3.25mm}
	\benr\label{eq:44}
	\h\tau=\argmin_{\tau\in \{1,...,(T-1)\}}Q(\tau,\check\theta_1,\check\theta_2)
	\eenr
	
	\vspace{-3.25mm}
	\flushleft{\bf Step 2:} Update mean estimates to $\h\theta_1=\tilde\theta_1(\h\tau),$ and $\h\theta_2=\tilde\theta_2(\h\tau)$ and perform another update of the change point estimate as,
	
	\vspace{-3.25mm}
	\benr
	\breve\tau=\argmin_{\tau\in \{1,...,(T-1)\}}Q(\tau,\h\theta_1,\h\theta_2)\nn
	\eenr
	
	\vspace{-3.25mm}
	\flushleft{\bf (Output):} $\breve\tau$
	
	\vspace{-1.25mm}
	\noi\rule{\textwidth}{0.5pt}
\end{figure}
%%%%%%%%%%%%%%%%%%%%%%%%%%%%%%%%%%%%%%%%%%%%%%%%%%%%%%%%%%%%%%%%%%%%%%%%%%%%%%%%%%%%%%%%%%%%%%%%%%%%%%%%%%%%%%%%%%%%%%%%%%%%%%%%%%%%%%%%%%%%%%%%%%%%%%%%%%%%%%%%%%%%%%%%%%%

To complete the description of Algorithm 1, we provide the mild sufficient condition required from the initializing choice $\check\tau.$

%%%%%%%%%%%%%%%%%%%%%%%%%%%%%%%%%%%%%%%%%%%%%%%%%%%%%%%%%%%%%%%%%%%%%%%%%%%%%%%%%%%%%%%%%%%%%%%%%%%%%%%%%%%%%%%%%%%%%%%%%%%%%%%%%%%%%%%%%%%%%%%%%%%%%%%%%%%%%%%%%%%%%%%%%%%
\vspace{1.5mm}
{\it {{\noi{\bf Condition F (initializing assumption):}} Let $\psi=\|\eta^0\|_{\iny}$ and assume that the initializer $\check\tau$ of Algorithm 1 satisfies the following relations.
		\benr
		(i)\,\,\check\tau\wedge(T-\check\tau)\ge c_uTl_T,\quad{\rm and}\quad (ii) |\check\tau-\tau^0|\le \frac{c_{u1}Tl_T}{\big(\surd{(2s)\psi\big/\xi}\big)}.\nn
		\eenr		
		Here $l_T$ is as defined in Condition D, $c_{u}>0$ is any constant and $c_{u1}>0$ is an appropriately chosen small enough constant.}}
%%%%%%%%%%%%%%%%%%%%%%%%%%%%%%%%%%%%%%%%%%%%%%%%%%%%%%%%%%%%%%%%%%%%%%%%%%%%%%%%%%%%%%%%%%%%%%%%%%%%%%%%%%%%%%%%%%%%%%%%%%%%%%%%%%%%%%%%%%%%%%%%%%%%%%%%%%%%%%%%%%%%%%%%%%%

The first requirement of Condition F is clearly innocuous, all it requires is a marginal separation of the chosen $\check\tau$ from the boundaries of the parametric space of the change point. It is satisfied with $\check\tau=\lfloor Tk\rfloor,$ with any $k\in [c_{u1},c_{u2}]\subset(0,1).$

The second requirement is discussed in the following, first from a theoretical and then followed by a practical perspective. For simplicity consider the case when $l_T\ge c_u<1,$ i.e., the true change point $\tau^0/T$ in the fractional scale is in some bounded subset of $(0,1),$ and that $\big(\surd{(2s)\psi\big/\xi}\big)=O(1),$ i.e., the entries of the change vector $\eta^0$ are roughly evenly spread across its non-zero components and not with uneven diverging spikes, this is also satisfied if one assumes $\psi\le c_u.$ Then, requirement (ii) of Condition F is satisfied for all $\check\tau$ in an $o(T)$ neighborhood of $\tau^0,$ i.e., any $\check\tau$ satisfying $|\check\tau-\tau^0|=o(T).$

We shall show that despite choosing any starting value in this $o(T)$ neighborhood, Step 1 of Algorithm 1 shall then move it into a near optimal neighborhood. Following which, the next iteration of Step 2 will then move it to an optimal neighborhood of $\tau^0,$ i.e., $o(T)$-nbd.$\longrightarrow^{\rm Step 1}$ near optimal-nbd., $O_p(\xi^{-2}s\log p)$ $\longrightarrow^{\rm Step 2}$ optimal-nbd., $O_p(\xi^{-2}).$ Note here the sequential improvement in the rate of convergence from initializing to Step 2. Moreover, the improvement to optimality in exactly two iterations. Another important consequence of these results is that it shows the redundancy of any further iterations, in the sense that since an optimal rate has been obtained at Step 2, performing further iterations will not yield any statistical improvement in the estimation of $\tau^0.$ Upon viewing the above discussion from a rate perspective provides the theoretical argument in support of the mildness of Condition F.

From a practical perspective, choosing a theoretically valid initializer $\check\tau$ in an $o(T)$ neighborhood of $\tau^0$ is quite straightforward, for e.g. one may choose any slowly diverging sequence  (say $\log T$) and choose $\log T$ equally separated values in $\{1,...,T\}$ forming a coarse grid of possible initializer values. Upon choosing the best fitting value $\check\tau$ for Algorithm 1 from this coarse initializer grid (minimizing squared loss) and assuming that the best fitting value is closest to $\tau^0,$ amongst the chosen grid points (this follows fairly naturally and can also be verified analytically by arguments similar to those of Theorem \ref{thm:nearoptimalcp}). Then by the pigeonhole principle this choice of $\check\tau$ must be in an $T/\log T=o(T)$ neighborhood of $\tau^0.$ Thereby this $\check\tau$ shall form a theoretically valid initializer. A similar preliminary coarse grid search has also been heuristically utilized in \cite{roy2017change} in a different model setting.

However, based on extensive numerical experiments we observe that this preliminary coarse grid search is numerically redundant. It is observed that any arbitrarily chosen $\check\tau$ separated from the boundaries of the parametric space yields statistically indistinguishable updates of Algorithm 1 when $T$ is large. The reader may numerically confirm these observations using the software associated with this article. An illustration of this behavior is provided in Figure \ref{fig:insensitivity1} with a single data set realization. In Section \ref{sec:numerical} we present results with the initializer fixed at $\check\tau=\lfloor T/2\rfloor,$ irrespective of the location of the true change point $\tau^0.$  Note  that in the absence of any information on $\tau^0,$ this choice of  $\check\tau=\lfloor T/2\rfloor$ forms the worst or farthest initializer in a mean distance sense. All other values of $\check\tau$ shall only serve to make estimation easier. Despite this worst possible choice, numerical results remain indistinguishable compared to those obtained when $\check\tau$ is chosen with a preliminary coarse grid search. A version of this condition has also been provided in \cite{kaul2019efficient} in the context of near optimal estimation of a change point in linear models together with evidence in its support.

\begin{figure}[]
	\centering
	\resizebox{\textwidth}{!}{
		\begin{minipage}[b]{0.45\textwidth}
			\includegraphics[width=\textwidth]{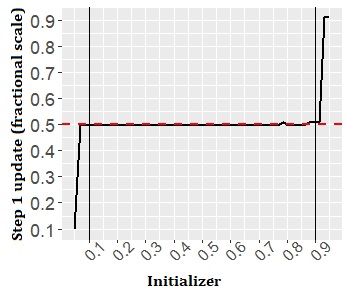}
		\end{minipage}
		\hspace{1in}
		\begin{minipage}[b]{0.45\textwidth}
			\includegraphics[width=\textwidth]{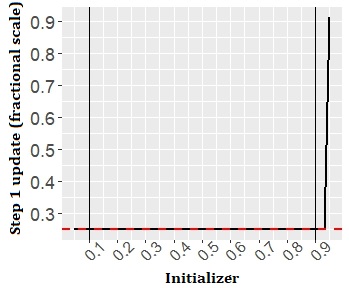}
	\end{minipage}}
	% note that files may not be rotated
	\caption{\footnotesize{Illustration of robustness of Algorithm 1 to the initializer $\check\tau.$ x-axis: initializer $\check\tau,$ y-axis: estimated change point $\h\tau$ of Step 1 of Algorithm 1. This illustration is based on a single dataset $x,$ with $\tau^0=\lfloor T/2\rfloor$ (Left panel: indicated by red line) and $\tau^0=\lfloor T/4 \rfloor$ (Right panel:indicated by red line). Additional parameters: $T=225,$ $p=100,$ $\theta_1^0=(1_{1\times 5}, 0_{1\times p-5})^T,$ $\theta_2^0=(0_{1\times 5},1_{1\times 5}, 0_{1\times p-10})^T$ and $\vep_t\sim^{i.i.d.}\cN(0,\Sigma),$ with $\Sigma_{ij}=\rho^{|i-j|}.$}}
	\label{fig:insensitivity1}
\end{figure}

In the following we provide a precise description of the statistical performance of Algorithm 1, starting with a result that obtains the near optimal rate of convergence of $\h\tau$ of Step 1 of Algorithm 1.

\begin{thm}\label{thm:al1.near.optimal} Suppose the model (\ref{model:rvmcp}) and assume the following,
	\benr\label{eq:23}
	\frac{c_{u}\si}{\xi}\Big\{\frac{s\log (p\vee T)}{Tl_T}\Big\}^{\frac{1}{2}}\le c_{u1},
	\eenr	
	for an appropriately chosen small enough constant $c_{u1}>0.$ Additionally assume $T\ge 2,$ the regularizers $\la_1$ and $\la_2$ for Step 1 of Algorithm 1 are chosen as in (\ref{eq:la.step1.choice}), and assume either one of the following two sets of conditions.\\~
	(a) Condition A(I) (subgaussian), B and D hold.\\~
	(b) Condition A(II) (subexponential), B and D hold and $c_uTl_T \ge  \log (p\vee T).$\\~
	Then, $\h\tau=\tilde\tau(\check\theta_1,\check\theta_2)$ of Step 1 of Algorithm 1 satisfies the following.
	\benr\label{eq:24}
	|\h\tau-\tau^0|\le \begin{cases}c_{u}\si^2\xi^{-2}s\log(p\vee T) & {\rm under\,\, Conditions\, (a)}\\ c_{u}\si^2\xi^{-2}s\log^2(p\vee T) &  {\rm under\,\, Conditions\, (b)}\end{cases}
	\eenr	
	with probability at least $1-o(1).$
\end{thm}

The result of Theorem \ref{thm:al1.near.optimal} shows that $\h\tau$ of Step 1 of Algorithm 1 will satisfy near optimal bounds despite the algorithm initializing with a nearly arbitrary $\check\tau.$
The conclusion of this theorem is effectively same as that of Theorem \ref{thm:nearoptimalcp} and Theorem \ref{thm:subE.nearoptimal.special}, with the distinction being that here the $\h\tau$ is a implementable estimate in comparison to Theorem \ref{thm:nearoptimalcp}, where the availability of nuisance estimates satisfying Condition C.1 was assumed. Following are two important remarks regarding Algorithm 1 and Theorem \ref{thm:al1.near.optimal}.

\begin{rem}\label{rem:near.optimal} {\rm This remark is a continuation of Remark \ref{rem:near.optimal.rem1}.  Under conditions (b) (subexponential case) of Theorem \ref{thm:al1.near.optimal}, using the result of Theorem \ref{thm:nearoptimalcp} it can also be shown that
		\benr\label{eq:24b}
		|\h\tau-\tau^0|\le  c_u\si^{2}\max\big\{\xi^{-2}s\log(p\vee T),\,\log(p\vee T)\big\}
		\eenr	
		with probability $1-o(1).$ The bound presented in (\ref{eq:24}) is chosen since it is required for the results to follow. The reason we bring this up is because in the case where $\xi=O(\surd{s}),$ it may be observed that (\ref{eq:24b}) reduces to $|\h\tau-\tau^0|\le  c_u\xi^{-2}s\log (p\vee T),$ which is the same bound as that in the subgaussian case, i.e., in this case where only a near optimal rate of convergence is of interest, the heavier tail of the subexponential distribution does not impact the change point estimation neither through rate assumptions (\ref{eq:23}) nor through the rate of the change point estimator itself.}
\end{rem}

\begin{rem}[Boundary case of $\tau^0=0\,\,{\rm or}\,\,T$]\label{rem:boundary} {\rm It may be observed that Theorem \ref{thm:al1.near.optimal} assumes existence of a change point (Condition D), which is not required in Theorem \ref{thm:nearoptimalcp}. As discussed in Section \ref{sec:mainresults}, this distinction is not due to change point estimation itself but instead because one requires at least one realization from both underlying distributions before and after $\tau^0$ to obtain any estimate of both nuisance parameters $\theta_1^0$ and $\theta_2^0.$ If these mean parameters are known apriori one may also estimate the boundary points with the squared loss itself. Nevertheless, a $0$-norm regularization approach can be utilized to relax this assumption and include one boundary point, $\tau^0=T$\footnote{It is clear that when $\theta_1^0,$ $\theta_2^0$ are unknown both boundary values $\tau^0=0,T$ are not simultaneously identifiable since no realizations from one of the distributions are observed.}. This can be achieved by replacing Step 1 of Algorithm 1 with a regularized version,
		\benr
		\h\tau^*=\argmin_{\tau\in \{1,...,T\}}\big\{Q(\tau,\check\theta_1,\check\theta_2)+\g{\bf 1}[\tau\ne T]\big\}, \qquad \g>0.\nn
		\eenr
		Here $\g$ is a tuning parameter. One may observe that $\h\tau^*$ can equivalently be written as,
		\benr\label{eq:45}
		\h\tau^*=\begin{cases} 	T			  &	 {\rm If}\,\, \{Q(T,\check\theta_1,\check\theta_2)-Q(\h\tau,\check\theta_1,\check\theta_2)\}<\g, \\
			\h\tau & {\rm else},
		\end{cases} \nn
		\eenr	
		where $\h\tau$ is as in (\ref{eq:44}). This representation is more common in the change point literature, see, e.g. \cite{fryzlewicz2014wild} and \cite{wang2018high}, where it is typically utilized to extend a single change point methodology to a multiple change point setting via variants of binary segmentation. Selection consistency \big($pr(\h\tau^*=T)\to 1$ when $\tau^0=T$\big) yielded by this regularization can be additionally verified via conventional arguments. This is quite intuitive since when $\tau^0=T,$ the mean parameter on both sides of any arbitrary cutoff $\check\tau$ is $\theta_1^0.$ Thus if $\g$ is chosen as an upper bound on residual noise, the boundary squared loss will be at most $\g$ larger than that obtained with any value in $\{1,...,(T-1)\},$ in probability. A rigorous proof is omitted since it is largely a  reproduction of existing arguments from the literature.}
\end{rem}

The construction of Algorithm 1 is modular in the sense that for it to yield an estimate $\breve\tau$ that is optimal in its rate of convergence, it does not require the estimator of Step 1 to be specifically the one that has currently been chosen. Instead, all that is required from Step 1 is that it provides some estimate $\h\tau$ that satisfies the bound (\ref{eq:24}) of Theorem \ref{thm:al1.near.optimal} with probability $1-o(1).$ Consequently, one may instead modify Algorithm 1 to use any other near optimal estimator in Step 1. This is described below as Algorithm 2.

%%%%%%%%%%%%%%%%%%%%%%%%%%%%%%%%%%%%%%%%%%%%%%%%%%%%%%%%%%%%%%%%%%%%%%%%%%%%%%%%%%%%%%%%%%%%%%%%%%%%%%%%%%%%%%%%%%%%%%%%%%%%%%%%%%%%%%%%%%%%%%%%%%%%%%%%%%%%%%%%%%%%%%%%%%%
\begin{figure}[H]
	\noi\rule{\textwidth}{0.5pt}
	
	\vspace{-2mm}
	\flushleft {\bf Algorithm 2:} Optimal estimation of $\tau^0:$
	
	\vspace{-1.25mm}
	\noi\rule{\textwidth}{0.5pt}
	
	\vspace{-1.25mm}
	\flushleft{\bf Step 1:} Implement any estimator $\h\tau$ from the literature that satisfies the near optimal bounds (\ref{eq:24}) with probability $1-o(1).$
	
	\vspace{-1.25mm}
	\flushleft{\bf Step 2:} Compute mean estimates $\h\theta_1=\tilde\theta_1(\h\tau),$ and $\h\theta_2=\tilde\theta_2(\h\tau)$ and perform the update,
	
	\vspace{-3.25mm}
	\benr
	\breve\tau=\argmin_{\tau\in \{1,...,(T-1)\}}Q(\tau,\h\theta_1,\h\theta_2)\nn
	\eenr
	
	\vspace{-3.25mm}
	\flushleft{\bf (Output):} $\breve\tau$
	
	\vspace{-1.25mm}
	\noi\rule{\textwidth}{0.5pt}
\end{figure}
%%%%%%%%%%%%%%%%%%%%%%%%%%%%%%%%%%%%%%%%%%%%%%%%%%%%%%%%%%%%%%%%%%%%%%%%%%%%%%%%%%%%%%%%%%%%%%%%%%%%%%%%%%%%%%%%%%%%%%%%%%%%%%%%%%%%%%%%%%%%%%%%%%%%%%%%%%%%%%%%%%%%%%%%%%%

An example of an estimator that can be used in Step 1 of Algorithm 2 is of \cite{wang2018high}, which obeys a tighter bound than that of Theorem \ref{thm:al1.near.optimal} under similar rate conditions on model parameters, and consequently also satisfies (\ref{eq:24}). However, this estimator would be limited to the Gaussian setting. To the best of our knowledge, there is no estimator that is currently available in the literature that would serve as a replacement for Step 1 of Algorithm 1 while allowing high dimensionality and under the assumed conditions of Theorem \ref{thm:al1.near.optimal}. The following result provides the optimal rate of convergence for the estimate $\breve\tau$ obtained from either Algorithm 1 or Algorithm 2 and shows that the limiting distributions of Section \ref{sec:inference} remain valid for these feasible estimators.

\begin{thm}\label{thm:alg1.optimal} Suppose model (\ref{model:rvmcp}) and assume $(\psi/{\xi}\big)\le c_u\surd\{\log(p\vee T)\}.$ Additionally assume the regularizers $\la_1$ and $\la_2$ for Step 1 of Algorithm 1 are chosen as in (\ref{eq:la.step1.choice}) and those for Step 2 of Algorithm 1 or Algorithm 2 are chosen as in (\ref{eq:la.step2.choice}) and assume either one of the following two sets of conditions.\\~
	(a) Condition A(I) (subgaussian), B, D, and the relation (\ref{eq:16}) hold, and $c_uTl_T\ge s\log (p\vee T).$\\~
	(b) Condition A(II) (subexponential), B, D and the relation (\ref{eq:17}) hold, and $c_uTl_T\ge s\log^2 (p\vee T).$  \\~
	Then the estimate $\breve\tau$ of Algorithm 1 or Algorithm 2 satisfies,  $\si^{-2}\xi^{2}(\breve\tau-\tau^0)=O_p(1).$ Additionally suppose Condition A$'$, E and (\ref{eq:31}) hold, then $\breve\tau$ of Algorithm 1 or Algorithm 2 obeys the limiting distributions of Theorem \ref{thm:wc.vanishing} and Theorem \ref{thm:wc.non.vanishing}, in the vanishing and non-vanishing regimes, respectively.
\end{thm}

The result of Theorem \ref{thm:alg1.optimal} concludes the task that was put forth in the problem setup of Section \ref{sec:intro}, i.e., to obtain feasible estimators that achieve an optimal rate of convergence and possess well defined limiting distributions which in turn allows inference on the change point parameter $\tau^0,$ despite high dimensionality of the underlying mean structure. Finally, we mention here the computational simplicity of Algorithm 1. It may be noted that computationally all that is required is two computations of sample means and two one dimensional discrete minimizations. The only notable computational cost arises from data based tuning process of choosing the regularizers $\la_1,\la_2$ for soft-thresholding. Effectively, this makes Algorithm 1 highly scalable and implementable on large scale data.

\section{Numerical results}\label{sec:numerical}

This section evaluates the numerical performance of the estimation and inference results developed in the preceding sections. The two main objectives of this section are to evaluate the estimation performance of the proposed Algorithm 1 ({\bf AL1}) and benchmark its performance with the estimator ({\bf WS}) of \cite{wang2018high}. While illustrating this objective we shall also compute the first step estimator ({\bf Step 1}) of Algorithm 1, which although is not optimal but still yields a near optimal rate of convergence. The second objective is to evaluate the empirical inference performance of Algorithm 1 when utilized in conjunction with the result of Theorem \ref{thm:alg1.optimal}. An auxiliary simulation examining uniformity of the proposed methodology over the mean parametric space is provided in Appendix \ref{app:numerical} of the supplementary materials. In all simulations to follow, no underlying parameter is assumed to be known.

We consider two simulation designs in the following. Simulation A considers the subgaussian setting with an underlying Gaussian distribution and Simulation B considers the subexponential setting with an underlying Laplace (double exponential) distribution. In all cases considered, the mean vectors are set to be $\theta_1^0=(\theta_{1\times s},0_{p-s})^T_{p\times 1}$ and $\theta_2^0=(0_{1\times s},\theta_{1\times s},0_{p-2s})^T_{p\times 1},$ and $s=5.$ Here $\theta_{1\times s}=\{1,...,0.25\},$ with $s=5$ equally spaced entries, this yields a jump size $\xi=2.14.$ The covariance matrix $\Si$ is chosen to be a toeplitz type matrix defined as $\Si_{ij}=\rho^{|i-j|},$ $i,j=1,...,p$ and $\rho=0.5.$ We consider all combinations of the sampling period $T\in\{200,275,350,425\},$ model dimension $p\in\{50,250,500,750\}$ and the change point $\tau^0\in\big\{\lfloor 0.2\cdotp T\rfloor,\lfloor 0.4\cdotp T\rfloor,\lfloor 0.6\cdotp T\rfloor,\lfloor 0.8\cdotp T\rfloor\big\}.$ The remaining specifications of Simulation A and Simulation B are as follows. For Simulation A, the unobserved noise variables $\vep_t$ are generated as independent Gaussian r.v.'s, more precisely we set $\vep_t\sim^{i.i.d} \cN(0,\Si),$ $t=1,...,T.$ For Simulation B the unobserved noise variables $\vep_t$ are generated as $\vep_t=\Si^{\frac{1}{2}}\vep_t^*,$ $t=1,...,T,$ where $\vep_t^*=(\vep_{t1}^*,....,\vep_{tp}^{*}),$ and each component $\vep_{tj}^*\sim^{i.i.d} {\rm Laplace}(0,1),$ $j=1,...,p,$ with zero mean and unit variance. This yields i.i.d random variables $\vep_t,$ $t=1,..,T$ which are subexponential random vectors with a covariance $\Si$ amongst components. Both Simulation A and Simulation B are further subdivided into two cases A(i), A(ii) and B(i), B(ii), the first of each simulation evaluating estimation performance and the second computing inference performance.

For the inference related designs of Simulation A(ii) and B(ii), we construct confidence intervals using both the limiting distributions of Theorem \ref{thm:wc.vanishing} and Theorem \ref{thm:wc.non.vanishing}. Note that by design $\xi$ is fixed throughout, hence the former limiting distribution is mis-specified for the considered cases. The significance level is set to $\al=0.05.$ Confidence intervals are constructed as $\big[(\breve\tau-ME),\, (\breve\tau+ME)\big],$ where $\breve\tau$ is the output of Algorithm 1 and the margin of error ($ME$) is computed as  $ME=q_{\alpha}^v\si^2_{\iny}/\xi^2$ or $ME=q_{\al}^{nv}$ based on the results of Theorem \ref{thm:wc.vanishing} and Theorem \ref{thm:wc.non.vanishing}, respectively. Here $q_{\al}^v$ represents the $\big(1-\alpha/2\big)^{th}$ quantile of the argmax of two sided negative drift Brownian motion of Theorem \ref{thm:wc.vanishing}. This critical value is evaluated as $c_{\alpha}=11.03$ by using its distribution function provided in \cite{yao1987approximating}. The $\big(1-\alpha/2\big)^{th}$ quantile $q_{\al}^{nv}$ of the argmax of the two sided negative drift random walk is computed as its monte carlo approximation by simulating $3000$ realizations of this distribution. Recall that Theorem \ref{thm:wc.non.vanishing} necessitates a parametric assumption on the distribution of the projection of $\vep_t$ (Condition A$'$). As per the assumed data generating process of Simulation A, the distribution $\cL$ here is assumed to be Gaussian for this design. For Simulation B(ii) we assume $\cL$ to also be Laplace distributed, this is clearly a mis-specification since Laplace distribution is not invariant under linear combinations. However this was empirically observed to be the closest parametric form amongst other common subexponential distributions. For implementation of the confidence interval, we utilize plugin estimates of $\si^2_{\iny}$ and $\xi^2,$ pertinent computational details of which are provided in Appendix \ref{app:numerical} of the supplementary materials.

Choice of tuning parameters: The regularizers $\la_1,$ $\la_2$ used to obtain soft thresholded mean estimates in Step 1 and Step 2 are tuned via a BIC type criteria. Specifically we set $\la_1=\la_2=\la,$ and evaluate $\tilde\theta_1(\la),$ and $\tilde\theta_2(\la)$ over an equally spaced grid of twenty five values in the interval $(0,0.5).$ Upon letting $\h S=\{j;\,\,\h\theta_{1j}\ne 0\}\cup \{j;\,\,\h\theta_{2j}\ne 0\}$ we evaluate the criteria,
\benr\label{eq:bic}
BIC(\la,\tau)= \sum_{t=1}^{\tau}\|x_t-\tilde\theta_1(\la)\|_2^2+\sum_{t= \tau+1}^{T}\|x_t-\tilde\theta_2(\la)\|_2^2+ |\h S|\log T.
\eenr
For Step 1 of Algorithm 1 we set $\la$ as the minimizer of $BIC(\la,\check\tau),$ and for Step 2 of Algorithm 1 we choose $\la$ as the minimizer of $BIC(\la,\h\tau).$ In context of the benchmarking estimator of \cite{wang2018high}, due to the absence of a recommended tuning mechanism, we follow a similar approach as above to also tune their estimator. Their estimator is implemented using the author provided r-package {\it InspectChangepoint} \cite{wang2016inspectchangepoint} on a grid of twenty five values in order to obtain a sequence of estimated change points. Each estimated change point is then used to construct corresponding soft-thresholded mean estimates, which are tuned via the BIC criteria as above. Finally, the squared loss criteria is applied to choose the tuned estimate from amongst the pairs of estimated change points and corresponding estimated mean parameters.

To report our results we present the following metrics. For the estimation results of Simulation A(i) and B(i) we report bias ($|E(\h\tau-\tau^0)|$), root mean squared error (RMSE, $E^{1/2}(\h\tau-\tau^0)^2$), and time (average over replications of running time in seconds)\footnote{CPU: Intel Xeon E5-2609 v3 @ 1.9GHz}, computed based on $100$ monte carlo replications. The reported computation time for Algorithm 1 includes all tuning undertaken for its computation, i.e., as it would be implemented in practice. For the benchmark estimator of \cite{wang2018high}, the reported computation time is that of repeating their estimation process over the chosen tuning grid of twenty five values and does not include the time taken to thereafter complete the tuning process as described above. For the inference results of Simulation A(ii) and B(ii), we report coverage (relative frequency of the number of times $\tau^0$ lies in the confidence interval) and the average margin of error (average over replications of the margin of error of each confidence interval) computed based on $500$ monte carlo replications.

Partial results of estimation simulations A(i), B(i), and inference simulations of A(ii) and B(ii) are provided in Table \ref{tab:est.A(i).t02}, Table \ref{tab:est.B(i).t02}, and Table \ref{tab:inf.A(ii).t02} and Table \ref{tab:inf.B(ii).t02}, respectively. Results of all remaining cases of these simulations are
provided in Table \ref{tab:est.A(i).t04} - Table \ref{tab:inf.B(ii).t08} in Appendix \ref{app:numerical} of the supplementary materials. These results provide strong numerical support to our theoretical results regarding estimation and limiting distribution behavior of the proposed Algorithm 1.

%
%
% and is in keeping with the limiting distribution result of Theorem \ref{thm:limitingdist}. Furthermore, the standard error estimates appear to be stable accross increasing values of $p.$
%

%%%%%%GAUSSIAN TABLES%%%%%%%%%%%%%%%%%%%%%%%%%%%%%%%%%%%%%%%%%%%%%%%%%%%%%%
\begin{table}[]
	\caption{\footnotesize{Simulation A(i): estimation performance of Step 1 ($\h\tau$), AL1 $(\breve\tau)$ and WS methods under Gaussian setting with $\tau^0=\lfloor0.2\cdotp T\rfloor.$ Bias ($|E(\h\tau-\tau^0)|$), and RMSE ($E^{1/2}(\h\tau-\tau^0)^2$) and time (in seconds), approximated with $100$ monte carlo replications.}}
	\resizebox{1\textwidth}{!}{
		\begin{tabular}{ccccccccccc}
			\hline
			\multicolumn{2}{c}{$\tau^0=\lfloor0.2\cdotp T\rfloor$} & \multicolumn{3}{c}{Step 1}                    & \multicolumn{3}{c}{AL1}                       & \multicolumn{3}{c}{WS}                        \\ \hline
			$T$                        & $p$                       & \textbf{bias} & \textbf{RMSE} & \textbf{time} & \textbf{bias} & \textbf{RMSE} & \textbf{time} & \textbf{bias} & \textbf{RMSE} & \textbf{time} \\ \hline
			200                        & 50                        & 2.410         & 5.823         & 0.060         & 0.320         & 1.876         & 0.111         & 0.110         & 2.666         & 0.117         \\
			200                        & 250                       & 1.930         & 6.089         & 0.145         & 0.380         & 2.107         & 0.261         & 1.050         & 4.088         & 1.504         \\
			200                        & 500                       & 2.970         & 11.937        & 0.197         & 2.240         & 11.475        & 0.370         & 1.970         & 5.835         & 7.414         \\
			200                        & 750                       & 0.000         & 4.228         & 0.254         & 0.180         & 2.478         & 0.460         & 1.860         & 5.552         & 22.699        \\ \hline
			275                        & 50                        & 2.400         & 5.020         & 0.093         & 0.570         & 3.442         & 0.168         & 0.070         & 2.610         & 0.139         \\
			275                        & 250                       & 1.520         & 3.592         & 0.253         & 0.300         & 2.392         & 0.450         & 0.700         & 5.923         & 1.961         \\
			275                        & 500                       & 1.600         & 4.035         & 0.321         & 0.420         & 1.811         & 0.596         & 2.710         & 8.945         & 7.862         \\
			275                        & 750                       & 0.780         & 4.474         & 0.400         & 0.140         & 2.510         & 0.743         & 1.890         & 4.951         & 23.096        \\ \hline
			350                        & 50                        & 1.850         & 4.836         & 0.098         & 0.190         & 2.347         & 0.179         & 0.290         & 2.076         & 0.162         \\
			350                        & 250                       & 1.180         & 3.552         & 0.268         & 0.110         & 1.700         & 0.480         & 0.350         & 3.294         & 2.025         \\
			350                        & 500                       & 1.700         & 3.680         & 0.409         & 0.350         & 1.947         & 0.743         & 0.920         & 3.990         & 8.699         \\
			350                        & 750                       & 1.470         & 4.339         & 0.494         & 0.100         & 1.703         & 0.911         & 0.810         & 3.838         & 24.261        \\ \hline
			425                        & 50                        & 2.220         & 6.263         & 0.131         & 0.130         & 1.873         & 0.242         & 0.470         & 2.193         & 0.176         \\
			425                        & 250                       & 1.330         & 4.021         & 0.304         & 0.060         & 2.112         & 0.581         & 0.290         & 2.629         & 2.234         \\
			425                        & 500                       & 1.860         & 3.912         & 0.500         & 0.100         & 1.828         & 0.970         & 0.330         & 2.247         & 9.471         \\
			425                        & 750                       & 1.390         & 3.345         & 0.621         & 0.170         & 1.997         & 1.225         & 1.470         & 5.489         & 26.021        \\ \hline
	\end{tabular}}
	\label{tab:est.A(i).t02}
\end{table}

%%%%%%LAPLACE TABLES%%%%%%%%%%%%%%%%%%%%%%%%%%%%%%%%%%%%%%%%%%%%%%%%%%%%%%

\begin{table}[]
	\caption{\footnotesize{Simulation B(i): estimation performance of Step 1 ($\h\tau$), AL1 $(\breve\tau)$ and WS methods under Laplace setting with $\tau^0=\lfloor 0.2\cdotp T\rfloor.$ Bias ($|E(\h\tau-\tau^0)|$), and RMSE ($E^{1/2}(\h\tau-\tau^0)^2$) and time (in seconds), approximated with $100$ monte carlo replications.}}
	\resizebox{1\textwidth}{!}{	\begin{tabular}{cclllllllll}
			\hline
			\multicolumn{2}{c}{$\tau^0=\lfloor0.2\cdotp T\rfloor$} & \multicolumn{3}{c}{Step 1}                                                                                & \multicolumn{3}{c}{AL1}                                                                                   & \multicolumn{3}{c}{WS}                                                                                    \\ \hline
			$T$                        & $p$                       & \multicolumn{1}{c}{\textbf{bias}} & \multicolumn{1}{c}{\textbf{RMSE}} & \multicolumn{1}{c}{\textbf{time}} & \multicolumn{1}{c}{\textbf{bias}} & \multicolumn{1}{c}{\textbf{RMSE}} & \multicolumn{1}{c}{\textbf{time}} & \multicolumn{1}{c}{\textbf{bias}} & \multicolumn{1}{c}{\textbf{RMSE}} & \multicolumn{1}{c}{\textbf{time}} \\ \hline
			200                        & 50                        & 3.910                             & 8.679                             & 0.063                             & 0.640                             & 2.458                             & 0.111                             & 0.240                             & 3.803                             & 0.113                             \\
			200                        & 250                       & 1.300                             & 6.822                             & 0.128                             & 0.700                             & 6.293                             & 0.244                             & 0.800                             & 4.909                             & 1.384                             \\
			200                        & 500                       & 1.740                             & 7.647                             & 0.220                             & 1.430                             & 6.932                             & 0.378                             & 2.330                             & 6.049                             & 7.100                             \\
			200                        & 750                       & 1.390                             & 4.021                             & 0.234                             & 0.560                             & 2.412                             & 0.439                             & 3.110                             & 8.190                             & 21.581                            \\ \hline
			275                        & 50                        & 2.480                             & 6.010                             & 0.086                             & 0.060                             & 1.625                             & 0.161                             & 0.200                             & 1.908                             & 0.139                             \\
			275                        & 250                       & 1.420                             & 4.416                             & 0.240                             & 0.020                             & 1.860                             & 0.427                             & 0.030                             & 2.524                             & 1.900                             \\
			275                        & 500                       & 1.260                             & 4.334                             & 0.293                             & 0.350                             & 2.161                             & 0.541                             & 0.640                             & 3.682                             & 7.627                             \\
			275                        & 750                       & 0.780                             & 3.914                             & 0.396                             & 0.430                             & 3.260                             & 0.728                             & 2.000                             & 5.860                             & 22.935                            \\ \hline
			350                        & 50                        & 2.790                             & 5.925                             & 0.107                             & 0.260                             & 1.871                             & 0.174                             & 0.110                             & 2.907                             & 0.138                             \\
			350                        & 250                       & 1.680                             & 4.637                             & 0.240                             & 0.000                             & 1.789                             & 0.444                             & 0.340                             & 2.550                             & 1.887                             \\
			350                        & 500                       & 1.650                             & 3.869                             & 0.371                             & 0.550                             & 2.105                             & 0.679                             & 1.690                             & 5.074                             & 8.247                             \\
			350                        & 750                       & 2.260                             & 6.482                             & 0.483                             & 0.590                             & 2.373                             & 0.892                             & 1.360                             & 4.630                             & 24.222                            \\ \hline
			425                        & 50                        & 1.770                             & 4.410                             & 0.114                             & 0.100                             & 1.783                             & 0.223                             & 0.300                             & 2.864                             & 0.160                             \\
			425                        & 250                       & 3.300                             & 7.695                             & 0.297                             & 0.700                             & 2.328                             & 0.564                             & 1.210                             & 4.189                             & 2.195                             \\
			425                        & 500                       & 1.340                             & 3.904                             & 0.487                             & 0.160                             & 2.078                             & 0.918                             & 0.860                             & 4.678                             & 9.103                             \\
			425                        & 750                       & 1.090                             & 3.291                             & 0.644                             & 0.010                             & 1.841                             & 1.247                             & 0.590                             & 3.848                             & 25.988                            \\ \hline
	\end{tabular}}
	\label{tab:est.B(i).t02}
\end{table}

%%%%%%%%%%%%%%%INFERENCE GAUSS%%%%%%%%%%%%%%%%%%%%%%%%%%%%%%%%%%%%%%%%%%%%%%%%%

\begin{table}[]
	\caption{\footnotesize{Simulation A(ii): inference using AL1 $(\breve\tau)$ with $\tau^0=\lfloor 0.2\cdotp T\rfloor,$ at significance level $\al=0.05.$ Here, {\bf V:} confidence intervals constructed using Theorem \ref{thm:wc.vanishing} under vanishing regime, {\bf NV:} confidence intervals constructed using Theorem \ref{thm:wc.non.vanishing} under non-vanishing regime (Gaussian parametric assumption). Computation based on 500 monte carlo replications.}}
	\resizebox{1\textwidth}{!}{	\begin{tabular}{ccccccc}
			\hline
			\multirow{2}{*}{} & \multicolumn{6}{c}{\textbf{Coverage (average margin of error)}}                               \\ \cline{2-7}
			& V             & NV            & V             & NV            & V             & NV            \\ \hline
			$p$               & \multicolumn{2}{c}{$n=275$}   & \multicolumn{2}{c}{$n=350$}   & \multicolumn{2}{c}{$n=425$}   \\ \hline
			50                & 0.922 (3.861) & 0.946 (3.808) & 0.936 (3.958) & 0.948 (3.879) & 0.946 (4.004) & 0.960 (3.939) \\
			250               & 0.918 (3.453) & 0.922 (3.384) & 0.914 (3.567) & 0.930 (3.473) & 0.944 (3.734) & 0.952 (3.635) \\
			500               & 0.902 (3.310) & 0.920 (3.252) & 0.912 (3.506) & 0.922 (3.443) & 0.906 (3.517) & 0.926 (3.450) \\
			750               & 0.882 (3.208) & 0.898 (3.107) & 0.928 (3.437) & 0.938 (3.347) & 0.920 (3.533) & 0.934 (3.467) \\ \hline
	\end{tabular}}
	\label{tab:inf.A(ii).t02}
\end{table}

%%%%%%%%%%%%%%INFERENCE LAPLACE%%%%%%%%%%%%%%%%%%%%%%%%%%%%%%%%%%%%%%%%%%%%%%%%%

\begin{table}[]
	\caption{\footnotesize{Simulation B(ii): inference using AL1 $(\breve\tau)$ with $\tau^0=\lfloor 0.2\cdotp T\rfloor,$ at significance level $\al=0.05.$ Here, {\bf V:} confidence intervals constructed using Theorem \ref{thm:wc.vanishing} under vanishing regime, {\bf NV:} confidence intervals constructed using Theorem \ref{thm:wc.non.vanishing} under non-vanishing regime (Laplace parametric assumption). Computation based on 500 monte carlo replications.}}
	\resizebox{1\textwidth}{!}{	\begin{tabular}{cllllll}
			\hline
			\multirow{2}{*}{} & \multicolumn{6}{c}{\textbf{Coverage (average margin of error)}}                                                                                  \\ \cline{2-7}
			& \multicolumn{1}{c}{V} & \multicolumn{1}{c}{NV} & \multicolumn{1}{c}{V} & \multicolumn{1}{c}{NV} & \multicolumn{1}{c}{V} & \multicolumn{1}{c}{NV} \\ \hline
			$p$               & \multicolumn{2}{c}{$n=275$}                    & \multicolumn{2}{c}{$n=350$}                    & \multicolumn{2}{c}{$n=425$}                    \\ \hline
			50                & 0.926 (3.785)         & 0.940 (3.763)          & 0.924 (3.910)         & 0.936 (3.877)          & 0.918 (3.981)         & 0.938 (3.951)          \\
			250               & 0.906 (3.424)         & 0.920 (3.386)          & 0.934 (3.595)         & 0.946 (3.554)          & 0.918 (3.627)         & 0.926 (3.570)          \\
			500               & 0.912 (3.286)         & 0.926 (3.263)          & 0.910 (3.482)         & 0.916 (3.446)          & 0.920 (3.605)         & 0.942 (3.560)          \\
			750               & 0.896 (3.157)         & 0.922 (3.117)          & 0.898 (3.327)         & 0.926 (3.296)          & 0.924 (3.497)         & 0.940 (3.443)          \\ \hline
	\end{tabular}}
	\label{tab:inf.B(ii).t02}
\end{table}

We begin with a discussion on the estimation results of Simulation A(i) and B(i) from Table \ref{tab:est.A(i).t02}, and Table \ref{tab:est.B(i).t02} in the Gaussian and Laplace settings, respectively. The proposed Algorithm 1 is observed to perform uniformly better over all considered model dimension sizes in comparison to the benchmark method WS when the sampling period is large $T\in\{350,425\}.$ In the case where $T\in\{200,275\},$ neither method is observed to be uniformly superior. In all results, the Step 1 estimator is observed to be worst performer, this is not particularly surprising since the near optimal rate of convergence $O_p(\xi^{-2}s\log p)$ of the Step 1 estimator derived in Theorem \ref{thm:al1.near.optimal} is indeed slower than that of WS: $O_p(\xi^{-2}\log \log T)$ and the optimal rate of Algorithm 1: $O_p(\xi^{-2}).$ We note here that the Laplace setting of Simulation B is a mis-specification for the method WS since that method is developed under a Gaussian setting.

Moving on to the inference results of Simulation A(ii) and B(ii) from Table \ref{tab:inf.A(ii).t02} and Table \ref{tab:inf.B(ii).t02}. The proposed Algorithm 1 and the inference methodology provides good control on the nominal significance level with an expected deterioration observed with larger values of $p$ and values of $\tau^0$ closer to the boundary of the parametric space (also see, results of Table \ref{tab:inf.A(ii).t04}-\ref{tab:inf.B(ii).t08}), but importantly the coverage is observed to catchup to the nominal level as $T$ increases. Some observations from these results are following. The confidence intervals constructed using the non-vanishing regime result appear to provide nearly uniformly more precise coverage in comparison to those constructed using the vanishing regime result. While the reader may recall that the latter setting of a vanishing jump size regime is mis-specified under the considered designs, however, this should not be the cause for the above observation since under this mis-specification one would expect conservative coverage as opposed to the observed deficient coverage.  Following are two speculative reasons that could be the root of this observation. There may be finite sample biases in the estimated jump size $\h\xi$ and estimated asymptotic variance $\h\si^2_{\iny}$ which are inherent to regularized estimators. This reason however is not likely since this would also have impacted the non-vanishing regime confidence intervals equally, but is not observed to be the case.
The most probable reason is due to the non-vanishing result itself and the manner in which its quantiles are computed. Specifically, since this result is based on a parametric distributional assumption and moreover its quantiles are evaluated as a monte-carlo approximation from realizations of the limiting distribution generated via the estimated jump size and asymptotic variance, it is consequently more adaptive to the specific data set realization under consideration in a finite sample sense. A final unusual observation is that despite the non-vanishing regime providing a higher coverage, the average margin of error is smaller than that of the vanishing regime. The margin of error being lower and coverage being higher is clearly not possible uniformly over all replications since both utilize the same estimate $\breve\tau$ of Algorithm 1. Instead, upon a careful examination of individual intervals it was observed that the reason here is again that the quantiles of the non-vanishing regime are more adaptive to the specific data set realization under consideration.

\appendix

%\appendixone

\section{Proofs}

\subsection{Proofs of Section \ref{sec:mainresults}}

To present the arguments of this section we require some additional notation. In all to follow define $\h\eta=\h\theta_1-\h\theta_2.$
Also, for any non-negative sequences $0\le v_T\le u_T\le 1$ we define the following collection. Let
\benr\label{def:setG}
\cG(u_T,v_T)=\Big\{\tau\in\{1,2,...,T\};\,\,Tv_T\le |\tau-\tau^0|\le Tu_T \Big\}
\eenr
Finally, for any vectors $\theta_1,\theta_2\in\R^p$ and any $\tau\in\{0,...,T\},$ define,
\benr\label{def:cU}
\cU(\tau,\theta_1,\theta_2)&=&Q(\tau,\theta_1,\theta_2)-Q(\tau^0,\theta_1,\theta_2),\nn\\
&=&\begin{cases}\frac{1}{T}\sum_{t=(\tau^0+1)}^{\tau}\big\{\|x_t-\theta_1\|_2^2-\|x_t-\theta_2\|_2^2\big\}, & \tau=(\tau^0+1),...,T\\
	0,											                             & \tau=\tau^0\\
	-\frac{1}{T}\sum_{t=(\tau+1)}^{\tau^0}\big\{\|x_t-\theta_1\|_2^2-\|x_t-\theta_2\|_2^2\big\},		 &	\tau=0,...,(\tau^0-1)
\end{cases}
\eenr
where $Q(\cdotp,\cdotp,\cdotp)$ is the least squares loss in (\ref{def:Q}) defined for any $T\ge 2.$ The proof of Theorem \ref{thm:nearoptimalcp} shall rely on the following preliminary lemma that provides a uniform lower bound on the expression $\cU(\tau,\h\theta_1,\h\theta_2),$ over the collection $\cG(u_T,v_T).$

%%%%%%%%%%%%%%%%%%%%%%%%%%%%%%%%%%%%%%%%%%%%%%%%%%%%%%%%%%%%%%%%%%%%%%%%%%%%%%%%%%%%%%%%%%%%%%%%%%%%%%%%%%%%%%%%%%%%%%%%%%%%%%%%%%%%%%%%%%%%%%%%%%%%%%%%%%%%%%%%%%%%%%%%%%%%
\begin{lem}\label{lem:lower.b.near.optimal} Suppose the model (\ref{model:rvmcp}) and assume  $\tau^0\wedge(1-\tau^0)\ge 0$ and that $\xi>0.$ Additionally assume Condition A(I) (subgaussian setting), B and C.1 hold and let $0\le v_T\le u_T\le 1,$ be any non-negative sequences. Then for $T\ge 2,$ and any $c_{u}>2,$ we have,
	\benr\label{eq:3}
	\inf_{\tau\in\cG(u_T,v_T)}\cU(\tau,\h\theta_1,\h\theta_2)\ge \frac{\xi^2}{2}\Big[v_T-\frac{6\surd(2c_u)\si}{\xi}\Big\{\frac{u_Ts\log(p\vee T)}{T}\Big\}^{\frac{1}{2}}\Big]
	\eenr
	with probability at least $1-2\exp\{-c_{u1}\log(p\vee T)\}-\pi_T,$ for $c_{u1}=(c_u-2)>0.$ Alternatively, suppose Condition A(II) (subexponential setting), B and C.1 hold. Additionally assume that  $T\ge 2\vee\log (p\vee T),$ and that the sequence $v_T$ satisfies $Tv_T\ge \log(p\vee T).$ Then, for any $c_u>8,$ the same bound (\ref{eq:3}) holds with probability at least $1-\exp\big\{-c_{u2}\log (p\vee T)\big\},$ for $c_{u2}=(\surd(c_u/2)-2)>0.$
\end{lem}
%%%%%%%%%%%%%%%%%%%%%%%%%%%%%%%%%%%%%%%%%%%%%%%%%%%%%%%%%%%%%%%%%%%%%%%%%%%%%%%%%%%%%%%%%%%%%%%%%%%%%%%%%%%%%%%%%%%%%%%%%%%%%%%%%%%%%%%%%%%%%%%%%%%%%%%%%%%%%%%%%%%%%%%%%%%%

%%%%%%%%%%%%%%%%%%%%%%%%%%%%%%%%%%%%%%%%%%%%%%%%%%%%%%%%%%%%%%%%%%%%%%%%%%%%%%%%%%%%%%%%%%%%%%%%%%%%%%%%%%%%%%%%%%%%%%%%%%%%%%%%%%%%%%%%%%%%%%%%%%%%%%%%%%%%%%%%%%%%%%%%%%%%
\begin{proof}[Proof of Theorem \ref{lem:lower.b.near.optimal}] We begin this proof with a few observations that shall be required to obtain the desired bound (\ref{eq:3}). Using Condition C.1 we have the following relations,
	\benr\label{eq:4}
	\|\h\eta-\eta^0\|_2&\le& \|\h\theta_1-\theta_1^0\|_2+\|\h\theta_2-\theta_2^0\|_2\le 2c_{u1}\xi\quad{\rm and\,\, similarly},\nn\\
	\|\h\eta-\eta^0\|_1&\le& 4\surd{s}\|\h\theta_1-\theta_1^0\|_2+4\surd{s}\|\h\theta_2-\theta_2^0\|_2\le 8c_{u1}\surd{s}\xi
	\eenr
	with probability at least $1-\pi_T.$ Here the third inequality follows since the assumption $\|(\h\theta_1)_{S_1^c}\|_1\le 3\|(\h\theta_1-\theta_1^0)_{S_1}\|_1$ of Condition C.1 in turn implies that $\|\h\theta_1-\theta_1^0\|_1\le 4\surd{s}\|\h\theta_1-\theta_1^0\|_2.$ Next, consider,
	\benr\label{eq:5}
	\|\h\eta\|_2&\le& \|\h\eta-\eta^0\|_2+\|\eta^0\|_2\le \xi\{1+2c_{u1}\}\le \frac{3}{2}\xi,\quad{\rm and\,\, similarly},\nn\\
	\|\h\eta\|_1&\le&\|\h\eta-\eta^0\|_1+\|\eta^0\|_1\le 8c_{u1}\surd{s}\xi+\surd{s}\xi\nn\\
	&\le& \surd{s}\xi\{1+8c_{u1}\}\le \frac{3}{2}\surd{s}\xi,
	\eenr
	which holds with probability at least $1-\pi_T.$ Here the second inequality for the $\ell_2$ bound follows from (\ref{eq:4}) and the third follows from Condition C.1 where $c_{u1}>0$ is chosen to be small enough. The $\ell_1$ bound follows analogously. Lastly, consider the following expression,
	\benr\label{eq:10}
	\|\h\eta\|_2^2+2(\h\theta_2-\theta_2^0)^T\h\eta &=&\|\eta^0+(\h\eta-\eta^0)\|_2^2+2(\h\theta_2-\theta_2^0)^T\h\eta\nn\\
	&\ge& \|\eta^0\|_2^2+2(\h\eta-\eta^0)^T\eta^{0}+2(\h\theta_2-\theta_2^0)^T\h\eta\nn\\
	&\ge& \|\eta^0\|_2^2-2\|\h\eta-\eta^0\|_2\|\eta^{0}\|_2-2\|\h\theta_2-\theta_2^0\|_2\|\h\eta\|_2\nn\\
	&\ge&\xi^2(1-4c_{u1}-3c_{u1})\ge \frac{\xi^2}{2},
	\eenr
	which again holds with probability at least $1-\pi_T.$ Here the first inequality is simply an algebraic manipulation. The second follows from the Cauchy-Schwartz inequality and the third follows from Condition C.1, (\ref{eq:4}) and (\ref{eq:5}). The final inequality again follows since the constant $c_{u1}>0$ in Condition C.1 is chosen to be small enough.
	
	We now proceed to the main proof of the bound (\ref{eq:3}). Consider any $\tau\in\cG(u_T,v_T),$ and without loss of generality assume $\tau\ge \tau^0.$ The mirroring case of $\tau\le\tau^0$ can be proved using symmetrical arguments. Consider,
	\benr\label{eq:7}
	\cU(\tau,\h\theta_1,\h\theta_2)&=&Q(\tau,\h\theta_1,\h\theta_2)-Q(\tau^0,\h\theta_1,\h\theta_2)\nn\\
	&=&\frac{1}{T}\sum_{t=1}^{\tau}\|x_t-\h\theta_1\|^2+\frac{1}{T}\sum_{t=\tau+1}^{T}\|x_t-\h\theta_2\|^2\nn\\
	&&-\frac{1}{T}\sum_{t=1}^{\tau^0}\|x_t-\h\theta_1\|^2+\frac{1}{T}\sum_{t=\tau^0+1}^{T}\|x_t-\h\theta_2\|^2\nn\\
	&=&\frac{1}{T}\sum_{t=\tau^0+1}^{\tau}\|x_t-\h\theta_1\|^2-\frac{1}{T}\sum_{t=\tau^0+1}^{\tau}\|x_t-\h\theta_2\|^2\nn\\
	&=&\frac{(\tau-\tau^0)}{T}\|\h\eta\|_2^2 -\frac{2}{T}\sum_{t=\tau^0+1}^{\tau}\vep_t^T\h\eta+\frac{2(\tau-\tau^0)}{T}(\h\theta_2-\theta_2^0)^T\h\eta\nn\\
	&\ge &\frac{v_T\xi^2}{2}- \frac{2}{T}\Big\|\sum_{t=\tau^0+1}^{\tau}\vep_t\Big\|_{\iny}\|\h\eta\|_1\nn\\
	&\ge&\frac{v_T\xi^2}{2}- \frac{3\surd{s}\xi}{T}\Big\|\sum_{t=\tau^0+1}^{\tau}\vep_t\Big\|_{\iny},	\eenr
	with probability at least $1-\pi_T.$ Where the final inequality follows by using (\ref{eq:5}). The uniform bound (\ref{eq:3}) now follows by substituting the uniform bound in  Lemma \ref{lem:nearoptimalcross} for term $\big\|\sum_{t=\tau^0+1}^{\tau}\vep_t\big\|_{\iny}$ in (\ref{eq:7}), for the subgaussian and subexponential cases, under their respective assumptions.
\end{proof}
%%%%%%%%%%%%%%%%%%%%%%%%%%%%%%%%%%%%%%%%%%%%%%%%%%%%%%%%%%%%%%%%%%%%%%%%%%%%%%%%%%%%%%%%%%%%%%%%%%%%%%%%%%%%%%%%%%%%%%%%%%%%%%%%%%%%%%%%%%%%%%%%%%%%%%%%%%%%%%%%%%%%%%%%%%%%

\bc$\rule{4in}{0.1mm}$\ec

%%%%%%%%%%%%%%%%%%%%%%%%%%%%%%%%%%%%%%%%%%%%%%%%%%%%%%%%%%%%%%%%%%%%%%%%%%%%%%%%%%%%%%%%%%%%%%%%%%%%%%%%%%%%%%%%%%%%%%%%%%%%%%%%%%%%%%%%%%%%%%%%%%%%%%%%%%%%%%%%%%%%%%%%%%%%
\begin{proof}[Proof of Theorem \ref{thm:nearoptimalcp}]
	The proof of this result relies on a recursive argument on Lemma \ref{lem:lower.b.near.optimal}, where the desired rate of convergence is obtained by a series of recursions, with this rate being sharpened at each step. We begin by considering any $v_T>0$ and applying Lemma \ref{lem:lower.b.near.optimal} on the set $\cG(1,v_T)$ to obtain,
	\benrr
	\inf_{\tau\in\cG(1,v_T)}\cU(\tau,\h\theta_1,\h\theta_2)\ge \frac{\xi^2}{2}\Big[v_T-\frac{6\si\surd{2c_u}}{\xi}\Big\{\frac{s\log (p\vee T)}{T}\Big\}^{\frac{1}{2}}\Big].
	\eenrr
	with probability at least $1-2\exp\{-c_{u1}\log (p\vee T)\}-\pi_T,$ where $c_{u1}=(c_u-2)$ in the subgaussian case and $c_{u1}=(\surd(c_u/2)-2)$ in the subexponential setting. Now
	upon choosing any,
	\benr
	v_T>v_T^*=\frac{6\si\surd{2c_u}}{\xi}\Big\{\frac{s\log (p\vee T)}{T}\Big\}^{\frac{1}{2}},\nn
	\eenr
	we obtain $\inf_{\tau\in\cG(1,v_T)}\cU(\tau,\h\theta_1,\h\theta_2)>0,$ thus implying that $\tilde\tau\notin\cG(1,v_T^*),$ i.e., $|\tilde\tau-\tau^0|\le  Tv_T^*,$ with probability at least $1-2\exp\{-c_{u1}\log (p\vee T)\}-\pi_T.$\footnote{Since by construction of $\tilde\tau,$ we have $\cU(\tilde\tau,\h\theta_1,\h\theta_2)\le 0.$} Now reset $u_T=v_T^*$ and reapply Lemma \ref{lem:lower.b.near.optimal} for any $v_T>0$ to obtain,
	\benrr
	\inf_{\tau\in\cG(u_T,v_T)}\cU(\tau,\h\theta_1,\h\theta_2)\ge \frac{\xi^2}{2}\Big[v_T-\Big(\frac{6\si\surd{2c_u}}{\xi}\Big)^{1+\frac{1}{2}}\Big\{\frac{s\log (p\vee T)}{T}\Big\}^{\frac{1}{2}+\frac{1}{2}}\Big],\nn
	\eenrr
	with probability at least $1-2\exp\{-c_{u1}\log (p\vee T)\}-\pi_T.$. Again choosing any,
	\benr
	v_T>v_T^*= \Big(\frac{6\si\surd{2c_u}}{\xi}\Big)^{1+\frac{1}{2}}\Big\{\frac{s\log (p\vee T)}{T}\Big\}^{\frac{1}{2}+\frac{1}{2}},\nn
	\eenr
	we obtain $ \inf_{\tau\in\cG(u_T,v_T)}\cU(\tau,\h\theta_1,\h\theta_2)>0,$ thus yielding $\h\tau\notin\cG(u_T,v_T^*),$ i.e.,
	\benr
	|\tilde\tau-\tau^0|\le T\Big(\frac{6\si\surd{2c_u}}{\xi}\Big)^{a_2}\Big\{\frac{s\log (p\vee T)}{T}\Big\}^{b_2},
	\eenr
	with probability at least $1-2\exp\{-c_{u1}\log (p\vee T)\}-\pi_T.$ Where,
	\benr
	a_2=1+\frac{1}{2}=\sum_{j=0}^2\frac{1}{2^j},\quad{\rm and}\quad b_2=\frac{1}{2}+\frac{1}{4}=\sum_{j=1}^2\frac{1}{2^j}.\nn
	\eenr
	Note that the rate of convergence of $\tilde\tau$ has been sharpened at the second recursion in comparison to the first. Continuing these recursions by resetting $u_T$ to the bound of the previous recursion, and applying Lemma \ref{lem:lower.b.near.optimal}, we obtain for the $m^{th}$ recursion,
	\benr
	|\tilde\tau-\tau^0|\le T\Big(\frac{6\si\surd{2c_u}}{\xi}\Big)^{a_m}\Big\{\frac{s\log (p\vee T)}{T}\Big\}^{b_m},
	\eenr
	with probability at least $1-2\exp\{-c_{u1}\log (p\vee T)\}-\pi_T.$ Repeating these recursions an infinite number of times and noting that $a_\iny=\sum_{j=0}^{\iny}(1/2^j)=2,$ and $b_{\iny}=\sum_{j=1}^{\iny}(1/2^j)=1$ we obtain,
	\benr
	|\tilde\tau-\tau^0|\le T\Big(\frac{6\si\surd{2c_u}}{\xi}\Big)^2\frac{s\log p}{T}\nn
	\eenr
	with probability at least $1-2\exp\{-c_{u1}\log (p\vee T)\}-\pi_T.$ Finally, note that despite the recursions in the above argument, the probability of the bound after every recursion is maintained to be at least $1-2\exp\{-c_{u1}\log (p\vee T)\}-\pi_T.$ This follows since the probability statement of Lemma \ref{lem:lower.b.near.optimal} arises from stochastic upper bounds of Lemma \ref{lem:nearoptimalcross} applied recursively with a tighter bound at each recursion. This yields a sequence of events such that the event at each recursion is a proper subset of the event at the previous recursion. This completes the proof of Part (i) of this theorem.
	
	The distinction between the bounds of Part (ii) (subexponential case) and Part (i) (subgaussian case) arises due to the following reason. Recall that for the bound of Lemma \ref{lem:lower.b.near.optimal} to be valid in the subexponential case, we require $Tv_T\ge \log (p\vee T).$ This is in turn due to the same requirement in the corresponding subexponential part of Lemma \ref{lem:nearoptimalcross}. Thus the recursions in the above argument can only be performed so long as the rate is slower than $\log (p\vee T)$ for the subexponential case, thereby yielding the statement of Part (ii) of this theorem.
\end{proof}
%%%%%%%%%%%%%%%%%%%%%%%%%%%%%%%%%%%%%%%%%%%%%%%%%%%%%%%%%%%%%%%%%%%%%%%%%%%%%%%%%%%%%%%%%%%%%%%%%%%%%%%%%%%%%%%%%%%%%%%%%%%%%%%%%%%%%%%%%%%%%%%%%%%%%%%%%%%%%%%%%%%%%%%%%%%%

\bc$\rule{4in}{0.1mm}$\ec

%%%%%%%%%%%%%%%%%%%%%%%%%%%%%%%%%%%%%%%%%%%%%%%%%%%%%%%%%%%%%%%%%%%%%%%%%%%%%%%%%%%%%%%%%%%%%%%%%%%%%%%%%%%%%%%%%%%%%%%%%%%%%%%%%%%%%%%%%%%%%%%%%%%%%%%%%%%%%%%%%%%%%%%%%%%%
\begin{proof}[Proof of Theorem \ref{thm:subE.nearoptimal.special}]
	We begin with a version of Lemma \ref{lem:lower.b.near.optimal} that is valid in this subexponential setting with any non-negative sequences $0\le v_T\le u_T\le 1$ (as opposed to $\log (p\vee T)\le Tv_T\le u_T$ assumed in Lemma \ref{lem:lower.b.near.optimal}). Proceeding with identical as in Lemma \ref{lem:lower.b.near.optimal} and using the deviation bound of Lemma \ref{lem:nearoptimalcross.subE.special} (instead of Lemma \ref{lem:nearoptimalcross}) in (\ref{eq:7}), we obtain,
	\benr
	\inf_{\tau\in\cG(u_T,v_T)}\cU(\tau,\h\theta_1,\h\theta_2)\ge \frac{\xi^2}{2}\Big[v_T-\frac{12c_u\si}{\xi}\Big\{\frac{u_Ts}{T}\Big\}^{\frac{1}{2}}\log(p\vee T)\Big],\nn
	\eenr
	with probability at least $1-2\exp\{-(c_u-2)\log(p\vee T)\}-\pi_T.$ Without loss of generality assume $\tau\ge \tau^0.$ The mirroring case of $\tau\le\tau^0$ can be proved using symmetrical arguments. Now following the recursive tightening proof of Theorem \ref{thm:nearoptimalcp} we have for the $m^{th}$ recursion that, $\inf_{\tau}\in\cG(u_T,v_T^*)>0,$ i.e., $|\tilde\tau-\tau^0|\le Tv_T^*,$ with probability at least $1-2\exp\{-(c_u-2)\log(p\vee T)\}-\pi_T,$ when,
	\benr
	v_T^*=\Big(\frac{12\si c_u\log(p\vee T)}{\xi}\Big)^{a_m}\Big(\frac{s}{T}\Big)^{b_m}.\nn
	\eenr
	Where,
	\benr
	a_m=\sum_{j=0}^m\frac{1}{2^j},\quad{\rm and}\quad b_m=\sum_{j=1}^m\frac{1}{2^j}.\nn
	\eenr
	Continuing these recursions an infinite number of times we obtain,
	\benr
	|\tilde\tau-\tau^0|\le T \Big(\frac{12c_u\si\log(p\vee T)}{\xi}\Big)^2\Big(\frac{s}{T}\Big) \le  12^2\frac{c_u^2\si^2}{\xi^2}s\log^2(p\vee T),\nn
	\eenr
	with probability at least $1-2\exp\{-(c_u-2)\log(p\vee T)\}-\pi_T.$ This completes the proof of the bound (\ref{eq:8}) of this theorem. The remaining assertions are a direct application of the bound (\ref{eq:8}).
\end{proof}
%%%%%%%%%%%%%%%%%%%%%%%%%%%%%%%%%%%%%%%%%%%%%%%%%%%%%%%%%%%%%%%%%%%%%%%%%%%%%%%%%%%%%%%%%%%%%%%%%%%%%%%%%%%%%%%%%%%%%%%%%%%%%%%%%%%%%%%%%%%%%%%%%%%%%%%%%%%%%%%%%%%%%%%%%%%%

\bc$\rule{4in}{0.1mm}$\ec

%%%%%%%%%%%%%%%%%%%%%%%%%%%%%%%%%%%%%%%%%%%%%%%%%%%%%%%%%%%%%%%%%%%%%%%%%%%%%%%%%%%%%%%%%%%%%%%%%%%%%%%%%%%%%%%%%%%%%%%%%%%%%%%%%%%%%%%%%%%%%%%%%%%%%%%%%%%%%%%%%%%%%%%%%%%%
\begin{lem}\label{lem:lower.b.optimal} Suppose the model (\ref{model:rvmcp}) and assume  $\tau^0\wedge(1-\tau^0)\ge 0,$  $\xi>0,$ and let $0\le v_T\le u_T\le 1,$ be any non-negative sequences. Additionally assume either one of the following two sets of conditions.\\~
	(a) Suppose Condition A(I) (subgaussian), B and C.2 (I,II) holds, and let $T\ge 2.$\\~
	(b) Suppose Condition A(II) (subexponential), B and C.2 (I, III) holds and let $T\ge 2.$ \\~
	Then, for any $0<a<1,$ choosing $c_{a}\ge \surd{(1/a)},$ we have the following uniform lower bound.
	\benr
	\inf_{\tau\in\cG(u_T,v_T)}\cU(\tau,\h\theta_1,\h\theta_2)\ge\frac{\xi^2}{2}\Big[v_T- 6c_{a}\frac{(\si\vee\phi)}{\xi}\Big(\frac{u_T}{T}\Big)^{\frac{1}{2}}\Big]\nn
	\eenr
	with probability at least $1-a-2\exp\{-\log(p\vee T)\}-\pi_T.$
\end{lem}
%%%%%%%%%%%%%%%%%%%%%%%%%%%%%%%%%%%%%%%%%%%%%%%%%%%%%%%%%%%%%%%%%%%%%%%%%%%%%%%%%%%%%%%%%%%%%%%%%%%%%%%%%%%%%%%%%%%%%%%%%%%%%%%%%%%%%%%%%%%%%%%%%%%%%%%%%%%%%%%%%%%%%%%%%%%%

%%%%%%%%%%%%%%%%%%%%%%%%%%%%%%%%%%%%%%%%%%%%%%%%%%%%%%%%%%%%%%%%%%%%%%%%%%%%%%%%%%%%%%%%%%%%%%%%%%%%%%%%%%%%%%%%%%%%%%%%%%%%%%%%%%%%%%%%%%%%%%%%%%%%%%%%%%%%%%%%%%%%%%%%%%%%
\begin{proof}[Proof of Lemma \ref{lem:lower.b.optimal}] The structure of this proof is largely similar to that of Lemma \ref{lem:lower.b.near.optimal}, except that it requires a more careful analysis of a residual stochastic term in order to allow for the comparatively sharper bound, however with a slightly weaker probability statement.
	
	First, using Condition C.2, we have,
	\benr\label{eq:11}
	\|\h\eta-\eta^0\|_1&\le& 4\surd{s}\|\h\theta_1-\theta_1^0\|_2+4\surd{s}\|\h\theta_2-\theta_2^0\|_2\nn\\
	&\le& \begin{cases}8\surd{s}\frac{c_{u1}\xi}{\{s\log (p\vee T)\}^{1/2}}, & {\rm under\,\, Condition\, C.2\, (II),}\\ 8\surd{s}\frac{c_{u1}\xi}{s^{1/2}\log (p\vee T)}, & {\rm under\,\, Condition\, C.2\, (III)\,}, \end{cases}
	\eenr
	with probability at least $1-\pi_T.$ Also, using Condition C.2 we also have that,
	\benr\label{eq:12}
	\|\h\eta\|_2^2+2(\h\theta_2-\theta_2^0)^T\h\eta\ge \frac{\xi^2}{2},
	\eenr
	that holds with the same probability. This inequality follows by identical arguments as those in (\ref{eq:10}), which also hold here since Condition C.1 assumed in Lemma \ref{lem:lower.b.near.optimal} is weaker than Condition C.2 assumed here. Now consider any $\tau\in\cG(u_T,v_T),$ and without loss of generality assume $\tau\ge \tau^0.$  Following the arguments used to obtain (\ref{eq:7}) we have,
	\benr\label{eq:9}
	\cU(\tau,\h\theta_1,\h\theta_2)&=&Q(\tau,\h\theta_1,\h\theta_2)-Q(\tau^0,\h\theta_1,\h\theta_2)\nn\\
	&=&\frac{(\tau-\tau^0)}{T}\Big\{\|\h\eta\|_2^2 +2(\h\theta_2-\theta_2^0)^T\h\eta\Big\}-\frac{2}{T}\sum_{t=\tau^0+1}^{\tau}\vep_t^T\h\eta\nn\\
	&\ge &\frac{v_T\xi^2}{2}-\frac{2}{T}\sum_{t=\tau^0+1}^{\tau}\vep_t^T\h\eta\nn\\
	&=&\frac{v_T\xi^2}{2}- \frac{2}{T}\sum_{t=\tau^0+1}^{\tau}\vep_t^T\eta^0-\frac{2}{T}\sum_{t=\tau^0+1}^{\tau}\vep_t^T(\h\eta-\eta^0),
	\eenr
	with probability at least $1-\pi_T.$ The inequality follows from (\ref{eq:12}), and the final equality is obtained by simply an algebraic manipulation, however it is the key step that shall yield the desired bound of this lemma. We now consider each of the residual stochastic terms in the expression (\ref{eq:9}). First, applying Lemma \ref{lem:optimalcross} for any $0<a<1,$ with $c_{a}\ge \surd(1/a),$ we have,
	\benr\label{eq:13}
	\sup_{\substack{\tau\in\cG(u_T,v_T);\\ \tau\ge\tau^0}}\frac{2}{T}\Big|\sum_{t=\tau^0+1}^{\tau}\vep_t^T\eta^0\Big|\le 2c_{a}\phi\xi\Big(\frac{u_T}{T}\Big)^{\frac{1}{2}}
	\eenr
	with probability at least $1-a.$ The second stochastic term in (\ref{eq:9}) can be bounded as follows,
	\benr\label{eq:14}
	\frac{2}{T}\sum_{t=\tau^0+1}^{\tau}\vep_t^T(\h\eta-\eta^0)&\le&\frac{2}{T}\Big\|\sum_{t=\tau^0+1}^{\tau}\vep_t\Big\|_{\iny}\big\|\h\eta-\eta^0\big\|_1\nn\\
	&\le&32c_uc_{u1}\si\xi\Big(\frac{u_T}{T}\Big)^{\frac{1}{2}}\le  \si\xi\Big(\frac{u_T}{T}\Big)^{\frac{1}{2}},
	\eenr
	with probability at least $1-2\exp\{-\log(p\vee T)\}.$ The second inequality follows using the deviation bounds in Lemma \ref{lem:nearoptimalcross} and Lemma \ref{lem:nearoptimalcross.subE.special} for the subgaussian and subexponential cases, respectively, and the use of corresponding $\ell_1$ error bound of (\ref{eq:11}). The third inequality and the probability statement follows by first choosing $c_u=3$ \big(see, Lemma \ref{lem:nearoptimalcross} and Lemma \ref{lem:nearoptimalcross.subE.special}\big), and then recalling that by assumption, $c_{u1}$ of Condition C.2 is chosen to be small enough \big(choosing $c_{u1}\le 1/(32\cdotp 3)$\big). Substituting (\ref{eq:13}) and (\ref{eq:14}) in (\ref{eq:9}), we obtain,
	\benr
	\inf_{\substack{\tau\in\cG(u_T,v_T);\\\tau\ge\tau^0}}\cU(\tau,\h\theta_1,\h\theta_2)\ge \frac{v_T\xi^2}{2}-2c_a\phi\xi\Big(\frac{u_T}{T}\Big)^{\frac{1}{2}}-\si\xi\Big(\frac{u_T}{T}\Big)^{\frac{1}{2}}\nn\\
	\ge\frac{\xi^2}{2}\Big[v_T- 6c_a\frac{(\si\vee\phi)}{\xi}\Big(\frac{u_T}{T}\Big)^{\frac{1}{2}}\Big]\nn
	\eenr
	with probability at least $1-a-2\exp\{-\log(p\vee T)\}-\pi_T.$ The mirroring case of $\tau\le\tau^0$ can be proved using symmetrical arguments. This completes the proof of this lemma.
\end{proof}
%%%%%%%%%%%%%%%%%%%%%%%%%%%%%%%%%%%%%%%%%%%%%%%%%%%%%%%%%%%%%%%%%%%%%%%%%%%%%%%%%%%%%%%%%%%%%%%%%%%%%%%%%%%%%%%%%%%%%%%%%%%%%%%%%%%%%%%%%%%%%%%%%%%%%%%%%%%%%%%%%%%%%%%%%%%%

\bc$\rule{4in}{0.1mm}$\ec

%%%%%%%%%%%%%%%%%%%%%%%%%%%%%%%%%%%%%%%%%%%%%%%%%%%%%%%%%%%%%%%%%%%%%%%%%%%%%%%%%%%%%%%%%%%%%%%%%%%%%%%%%%%%%%%%%%%%%%%%%%%%%%%%%%%%%%%%%%%%%%%%%%%%%%%%%%%%%%%%%%%%%%%%%%%%
\begin{proof}[Proof of Theorem \ref{thm:cpoptimal}]
	The proof of this result follows a recursive argument similar to that of Theorem \ref{thm:nearoptimalcp}, the distinction being that these recursions are made on Lemma \ref {lem:lower.b.optimal} instead of Lemma \ref{lem:lower.b.near.optimal}. We begin by considering any $v_T>0,$ now applying Lemma \ref{lem:lower.b.optimal} on the set $\cG(1,v_T),$ for any $0<a<1,$ and choosing $c_a\ge\surd(1/a),$ we have,
	\benrr
	\inf_{\tau\in\cG(1,v_T)}\cU(\tau,\h\theta_1,\h\theta_2)\ge \frac{\xi^2}{2}\Big[v_T-6c_a\frac{(\si\vee\phi)}{\xi}\Big(\frac{1}{T}\Big)^{\frac{1}{2}}\Big].
	\eenrr
	with probability at least $1-a-2\exp\{-\log (p\vee T)\}-\pi_T.$	Upon choosing any,
	\benr
	v_T>v_T^*=6c_a\frac{(\si\vee\phi)}{\xi}\Big(\frac{1}{T}\Big)^{\frac{1}{2}},\nn
	\eenr
	we obtain that $\inf_{\tau\in\cG(1,v_T)}\cU(\tau,\h\theta_1,\h\theta_2)>0,$ thus implying that $\tilde\tau\notin\cG(1,v_T^*),$ i.e., $|\tilde\tau-\tau^0|\le  Tv_T^*,$ with probability at least $1-a-2\exp\{-\log (p\vee T)\}-\pi_T.$ Now reset $u_T=v_T^*$ and reapply Lemma  \ref{lem:lower.b.optimal} for any $v_T>0$ to obtain,
	\benrr
	\inf_{\tau\in\cG(1,v_T)}\cU(\tau,\h\theta_1,\h\theta_2)\ge \frac{\xi^2}{2}\Big[v_T-\Big\{6c_a\frac{(\si\vee\phi)}{\xi}\Big\}^{1+\frac{1}{2}}\Big(\frac{1}{T}\Big)^{\frac{1}{2}+\frac{1}{4}}\Big].
	\eenrr
	with probability at least $1-a-2\exp\{-\log (p\vee T)\}-\pi_T.$ Again choosing any,
	\benr
	v_T>v_T^*= \Big\{6c_a\frac{(\si\vee\phi)}{\xi}\Big\}^{1+\frac{1}{2}}\Big(\frac{1}{T}\Big)^{\frac{1}{2}+\frac{1}{4}},\nn
	\eenr
	we obtain that $\inf_{\tau\in\cG(u_T,v_T)}\cU(\tau,\h\theta_1,\h\theta_2)>0,$ consequently yielding $\h\tau\notin\cG(u_T,v_T^*),$ i.e.,
	\benr
	|\tilde\tau-\tau^0|\le T\Big\{6c_a\frac{(\si\vee\phi)}{\xi}\Big\}^{b_2}\Big(\frac{1}{T}\Big)^{c_2},
	\eenr
	with probability at least $1-a-2\exp\{-\log (p\vee T)\}-\pi_T.$ Where,
	\benr
	b_2=1+\frac{1}{2}=\sum_{j=0}^2\frac{1}{2^j},\quad{\rm and}\quad c_2=\frac{1}{2}+\frac{1}{4}=\sum_{j=1}^2\frac{1}{2^j}.\nn
	\eenr
	Continuing these recursions by resetting $u_T$ to the bound of the previous recursion, and applying Lemma \ref{lem:lower.b.optimal}, we obtain for the $m^{th}$ recursion,
	\benr
	|\tilde\tau-\tau^0|\le T\Big\{6c_a\frac{(\si\vee\phi)}{\xi}\Big\}^{b_m}\Big(\frac{1}{T}\Big)^{c_m},
	\eenr
	with probability at least $1-a-2\exp\{-\log (p\vee T)\}-\pi_T.$ Repeating these recursions an infinite number of times and noting that $b_\iny=\sum_{j=0}^{\iny}(1/2^j)=2,$ and $c_{\iny}=\sum_{j=1}^{\iny}(1/2^j)=1$ we obtain,
	\benr\label{eq:49}
	|\tilde\tau-\tau^0|\le 36c_a^2(\si\vee\phi)^2\xi^{-2}\nn
	\eenr
	with probability at least $1-a-2\exp\{-\log (p\vee T)\}-\pi_T.$ The statement of Theorem \ref{thm:cpoptimal} follows directly from the bound (\ref{eq:49}) upon recalling $\phi^2\le c_u\si^2$ from the discussion after Condition B. As seen earlier, despite the recursions in the above argument, the probability of the bound after every recursion is maintained to be at least $1-a-2\exp\{-\log (p\vee T)\}-\pi_T,$ due to the same reasoning as in Theorem \ref{thm:nearoptimalcp}.This completes the proof of Part (i) of this theorem.
\end{proof}
%%%%%%%%%%%%%%%%%%%%%%%%%%%%%%%%%%%%%%%%%%%%%%%%%%%%%%%%%%%%%%%%%%%%%%%%%%%%%%%%%%%%%%%%%%%%%%%%%%%%%%%%%%%%%%%%%%%%%%%%%%%%%%%%%%%%%%%%%%%%%%%%%%%%%%%%%%%%%%%%%%%%%%%%%%%%

\bc$\rule{4in}{0.1mm}$\ec

%%%%%%%%%%%%%%%%%%%%%%%%%%%%%%%%%%%%%%%%%%%%%%%%%%%%%%%%%%%%%%%%%%%%%%%%%%%%%%%%%%%%%%%%%%%%%%%%%%%%%%%%%%%%%%%%%%%%%%%%%%%%%%%%%%%%%%%%%%%%%%%%%%%%%%%%%%%%%%%%%%%%%%%%%%%%
\subsection{Proofs of Section \ref{sec:inference}}
%%%%%%%%%%%%%%%%%%%%%%%%%%%%%%%%%%%%%%%%%%%%%%%%%%%%%%%%%%%%%%%%%%%%%%%%%%%%%%%%%%%%%%%%%%%%%%%%%%%%%%%%%%%%%%%%%%%%%%%%%%%%%%%%%%%%%%%%%%%%%%%%%%%%%%%%%%%%%%%%%%%%%%%%%%%%

%%%%%%%%%%%%%%%%%%%%%%%%%%%%%%%%%%%%%%%%%%%%%%%%%%%%%%%%%%%%%%%%%%%%%%%%%%%%%%%%%%%%%%%%%%%%%%%%%%%%%%%%%%%%%%%%%%%%%%%%%%%%%%%%%%%%%%%%%%%%%%%%%%%%%%%%%%%%%%%%%%%%%%%%%%%%
As the reader may have observed, a change of notation has been carried out for the results of this section. These results are presented in more conventional {\it argmax} notation instead of the {\it argmin} notation of the problem setup in Section \ref{sec:intro}. This is purely a notational change and all results can equivalently be stated in the {\it argmin} language. Accordingly we define the following versions. Let $\cU(\tau,\theta_1,\theta_2)$ be as in (\ref{def:cU}) and consider,
\benr\label{def:cC}
\cC(\tau,\theta_1,\theta_2)&=&-T\cU(\tau,\theta_1,\theta_2)\nn\\
&=&\begin{cases}-\sum_{t=(\tau^0+1)}^{\tau}\big\{\|x_t-\theta_1\|_2^2-\|x_t-\theta_2\|_2^2\big\}, & \tau=(\tau^0+1),...,T\\
	0,											                             & \tau=\tau^0\\
	\sum_{t=(\tau+1)}^{\tau^0}\big\{\|x_t-\theta_1\|_2^2-\|x_t-\theta_2\|_2^2\big\}.		 &	\tau=0,...,(\tau^0-1)
\end{cases}
\eenr
The multiplication of $\cU$ with the sampling period $T$ is only meant for notational convenience later on. Then, we can re-express the change point estimator $\tilde\tau(\theta_1,\theta_2)$ defined in (\ref{est:optimalcp}) as,
\benr
\tilde\tau(\theta_1,\theta_2)=\argmax_{0\le \tau\le T}\cC(\tau,\theta_1,\theta_2)\nn
\eenr
The proofs of Theorem \ref{thm:wc.vanishing} and Theorem \ref{thm:wc.non.vanishing} below are applications of the Argmax Theorem (reproduced as Theorem \ref{thm:argmax} in Appendix \ref{app:auxiliary}). The arguments here are largely an exercise in verification of requirements of this theorem.
%%%%%%%%%%%%%%%%%%%%%%%%%%%%%%%%%%%%%%%%%%%%%%%%%%%%%%%%%%%%%%%%%%%%%%%%%%%%%%%%%%%%%%%%%%%%%%%%%%%%%%%%%%%%%%%%%%%%%%%%%%%%%%%%%%%%%%%%%%%%%%%%%%%%%%%%%%%%%%%%%%%%%%%%%%%%

\bc$\rule{4in}{0.1mm}$\ec

%%%%%%%%%%%%%%%%%%%%%%%%%%%%%%%%%%%%%%%%%%%%%%%%%%%%%%%%%%%%%%%%%%%%%%%%%%%%%%%%%%%%%%%%%%%%%%%%%%%%%%%%%%%%%%%%%%%%%%%%%%%%%%%%%%%%%%%%%%%%%%%%%%%%%%%%%%%%%%%%%%%%%%%%%%%%
\begin{proof}[Proof of Theorem \ref{thm:wc.vanishing}] Let the underlying indexing metric space be $\R,$ and consider the two cases of known and unknown mean parameters.
	
	\vspace{1.5mm}	
	\noi{\bf Case I \big($\theta_1^0$ and $\theta_2^0$ known\big):} Following is list of requirement of the Argmax theorem that require verification for this case (see, page 288 of \cite{vaart1996weak}).
	\benr
	&(i)& {\rm The\, sequence}\,\, \xi^{2}(\tilde\tau^*-\tau^0)\,\, {\rm is\, uniformly\, tight}.\nn\\
	&(ii)&	\big\{2\si_{\iny}W(\z)-|\z|\}\,\, {\rm satisfies\, suitable\, regularity\, conditions}\footnotemark.\nn\\
	&(iii)& {\rm For\, any}\,\, \z\in [-c_u,c_u]\,\, {\rm we\, have,}\,\, \cC(\tau^0+\z\xi^{-2},\theta_1^0,\theta_2^0)\Rightarrow \big\{2\si_{\iny}W(\z)-|\z|\big\}.\nn
	\eenr
	\footnotetext{Almost all sample paths $\z\to \big\{2\si_{\iny}W(\z)-|\z|\}$ are upper semicontinuous and posses a unique maximum at a (random) point $\argmax_{\z\in\R}\big\{2\si_{\iny}W(\z)-|\z|\},$ which as a random map in the indexing metric space is tight.}		
	Note that by setting $\h\theta_1=\theta_1^0$ and $\h\theta_2=\theta_2^0,$ Condition C.2 is trivially satisfied. Now using Theorem \ref{thm:cpoptimal} we have that $\xi^{2}(\tilde\tau^*-\tau^0)=O_p(1).$ This directly yields requirement (i). The second requirement follows from well known properties of Brownian motion's. The only remaining requirement is (iii), which is provided in the following. For any fixed $\z\in\R,$ first consider $\lfloor \z\xi^{-2}\rfloor$ and note that under the assumed regime of $\xi\to 0$ we have,
	\benr
	\z \leftarrow (\z-\xi^{2})\le \xi^{2}\lfloor \z\xi^{-2}\rfloor\le (\z+\xi^2)\to \z\nn
	\eenr
	Hence, w.l.o.g. we may assume $\z\xi^{-2}$ is integer valued. Now for any $\z\in(0,c_u],$ consider
	\benr
	\cC(\tau^0+\z\xi^{-2},\theta_1^0,\theta_2^0)&=&-\sum_{t=(\tau^0+1)}^{\tau^0+\z\xi^{-2}}\big\{\|x_t-\theta_1^0\|_2^2-\|x_t-\theta_2^0\|_2^2\big\}\nn\\
	&=&2\sum_{t=(\tau^0+1)}^{\tau^0+\z\xi^{-2}}\vep_t^T\eta^0-\z\Rightarrow 2\si_{\iny}W_1(\z)-\z,\nn
	\eenr
	where the weak convergence follows from the functional central limit theorem. Repeating the same argument with $\z\in[-c_u,0),$ yields  	$\cC(\tau^0+\z\xi^{-2},\theta_1^0,\theta_2^0)\Rightarrow 2\si_{\iny}W_2(\z)-|\z|.$ This completes the proof of requirement (iii) for the Argmax theorem and consequently an application of its results yields $\xi^{2}(\tilde\tau^*-\tau^0)\Rightarrow \argmax_{\z\in\R}\big\{2\si_{\iny}W(-\z)-|\z|\},$ which completes the proof of this case.
	
	\vspace{1.5mm}		
	\noi{\bf Case II \big($\theta_1^0$ and $\theta_2^0$ unknown\big):} In this case the applicability of the argmax theorem requires verification of the following conditions.
	\benr
	&(i)& {\rm The\, sequence}\,\, \xi^{2}(\tilde\tau-\tau^0)\,\, {\rm is\, uniformly\, tight}.\nn\\
	&(ii)&	\big\{2\si_{\iny}W(\z)-|\z|\}\,\, {\rm satisfies\, suitable\, regularity\, conditions}.\nn\\
	&(iii)& {\rm For\, any}\,\, \z\in [-c_u,c_u]\,\, {\rm we\, have,}\,\, \cC(\tau^0+\z\xi^{-2},\h\theta_1,\h\theta_2)\Rightarrow \big\{2\si_{\iny}W(\z)-|\z|\big\}.\nn
	\eenr
	Part (i) again follows from the result of Theorem \ref{thm:cpoptimal} under the assumed Condition C.2 on the nuisance estimates $\h\theta_1$ and $\h\theta_2.$  Part (ii) is identical to the corresponding requirement of Case I. Finally to prove part (iii) note that from Lemma \ref{lem:Capprox} we have that,
	\benr\label{eq:35}
	\sup_{\tau\in\cG(c_uT^{-1}\xi^{-2},0)}|\cC(\tau,\h\theta_1,\h\theta_2)-\cC(\tau,\theta_1^0,\theta_2^0)|=o_p(1).
	\eenr
	The approximation (\ref{eq:35}) and Part (iii) of Case I together imply Part (iii) for this case. This completes the verification of all requirements for this case. The statement of the theorem now follows by an application of the Argmax theorem.
\end{proof}

%%%%%%%%%%%%%%%%%%%%%%%%%%%%%%%%%%%%%%%%%%%%%%%%%%%%%%%%%%%%%%%%%%%%%%%%%%%%%%%%%%%%%%%%%%%%%%%%%%%%%%%%%%%%%%%%%%%%%%%%%%%%%%%%%%%%%%%%%%%%%%%%%%%%%%%%%%%%%%%%%%%%%%%%%%%%
\bc$\rule{4in}{0.1mm}$\ec
%%%%%%%%%%%%%%%%%%%%%%%%%%%%%%%%%%%%%%%%%%%%%%%%%%%%%%%%%%%%%%%%%%%%%%%%%%%%%%%%%%%%%%%%%%%%%%%%%%%%%%%%%%%%%%%%%%%%%%%%%%%%%%%%%%%%%%%%%%%%%%%%%%%%%%%%%%%%%%%%%%%%%%%%%%%%

\begin{proof}[Proof of Theorem \ref{thm:wc.non.vanishing}] Let the underlying indexing metric space be the set of integers $\Z$ and consider the following two cases.
	
	\vspace{1.5mm}	
	\noi{\bf Case I \big($\theta_1^0$ and $\theta_2^0$ known\big):} The requirements to be verified here are as follows.
	\benr
	&(i)& {\rm The\, sequence}\,\, (\tilde\tau^*-\tau^0)\,\, {\rm is\, uniformly\, tight}.\nn\\
	&(ii)&	\cC(\z)\,\, {\rm satisfies\, suitable\, regularity\, conditions}.\nn\\
	&(iii)& {\rm For\, any}\,\, \z\in \{-c_u,-c_u+1,...,-1,0,1,...c_u\}\,\, {\rm we\, have,}\,\, \cC(\tau^0+\z,\theta_1^0,\theta_2^0)\Rightarrow \cC_{\iny}(\z).\nn
	\eenr
	As in the proof of Theorem \ref{thm:wc.vanishing}, requirement (i) follows directly from the result of Theorem \ref{thm:cpoptimal}. Lemma \ref{lem:regularity.argmax} below provides the regularity requirements of Part (ii). The requirement (iii) is verified in the following. For any $\z\in\{1,2,...,c_u\},$ consider,
	\benr
	\cC(\tau^0+\z,\theta_1^0,\theta_2^0)&=&-\sum_{t=(\tau^0+1)}^{\tau^0+\z}\big\{\|x_t-\theta_1\|_2^2-\|x_t-\theta_2\|_2^2\big\}\nn\\
	&=&\sum_{t=(\tau^0+1)}^{\tau^0+\z}\Big(2\vep_t^T\eta^0-\xi^2\Big)=^d \sum_{t=1}^{\z}\cL(-\xi^2,4\eta^{0T}\Si\eta^0)\nn\\
	&\Rightarrow& \sum_{t=1}^{\z}\cL(-\xi_{\iny}^2,4\xi_{\iny}^2\si^2_{\iny}),\nn
	\eenr
	The third equality follows from Condition A$'$ and the convergence in distribution follows from Condition E(ii). Repeating the same argument with $\z\in\{-c_u,-c_u+1,...,-1\},$ yields  	$\cC(\tau^0+\z,\theta_1^0,\theta_2^0)\Rightarrow \sum_{t=1}^{-\z}\cL(-\xi_{\iny}^2,4\xi_{\iny}^2\si^2_{\iny}).$ An application the Argmax theorem now yields $(\tilde\tau^*-\tau^0)\Rightarrow \argmax_{\z\in\Z}\cC_{\iny}(\z),$ which completes the proof of this case.
	
	\vspace{1.5mm}		
	\noi{\bf Case II \big($\theta_1^0$ and $\theta_2^0$ unknown\big):} In this case, the applicability of the argmax theorem requires verification of the following.
	\benr
	&(i)& {\rm The\, sequence}\,\, (\tilde\tau-\tau^0)\,\, {\rm is\, uniformly\, tight}.\nn\\
	&(ii)&	\cC(\z)\,\, {\rm satisfies\, suitable\, regularity\, conditions}.\nn\\
	&(iii)& {\rm For\, any}\,\, \z\in \{-c_u,-c_u+1,...,-1,0,1,...c_u\}\,\, {\rm we\, have,}\,\, \cC(\tau^0+\z,\h\theta_1,\h\theta_2)\Rightarrow \cC_{\iny}(\z).\nn	\eenr
	Part (i) follows from Theorem \ref{thm:cpoptimal} under the assumed Condition C.2 on the nuisance estimates $\h\theta_1$ and $\h\theta_2.$  Part (ii) is identical to the corresponding requirement of Case I. Finally to prove part (iii) note that from Lemma \ref{lem:Capprox} we have that,
	\benr\label{eq:38}
	\sup_{\tau\in\cG(c_uT^{-1}\xi^{-2},0)}|\cC(\tau,\h\theta_1,\h\theta_2)-\cC(\tau,\theta_1^0,\theta_2^0)|=o_p(1).
	\eenr
	The approximation (\ref{eq:38}) and Part (iii) of Case I together imply Part (iii) for this case. This completes the verification of all requirements for this case. The statement of the theorem now follows by an application of the Argmax theorem.
\end{proof}
%%%%%%%%%%%%%%%%%%%%%%%%%%%%%%%%%%%%%%%%%%%%%%%%%%%%%%%%%%%%%%%%%%%%%%%%%%%%%%%%%%%%%%%%%%%%%%%%%%%%%%%%%%%%%%%%%%%%%%%%%%%%%%%%%%%%%%%%%%%%%%%%%%%%%%%%%%%%%%%%%%%%%%%%%%%%

\bc$\rule{4in}{0.1mm}$\ec

%%%%%%%%%%%%%%%%%%%%%%%%%%%%%%%%%%%%%%%%%%%%%%%%%%%%%%%%%%%%%%%%%%%%%%%%%%%%%%%%%%%%%%%%%%%%%%%%%%%%%%%%%%%%%%%%%%%%%%%%%%%%%%%%%%%%%%%%%%%%%%%%%%%%%%%%%%%%%%%%%%%%%%%%%%%%
\begin{lem}[Regularity conditions of $\argmax\cC(\z)$]\label{lem:regularity.argmax} Let $\cC(\z)$ be as defined in (\ref{def:cCz}) and suppose Condition A$'$ holds. Then the map $\z\to \cC(\z)$ is continuous with respect to the domain space $\Z.$ Additionally suppose Condition B and Condition E(ii) holds, then $\argmax_{\z\in\Z}\cC(\z)$ possesses an almost sure unique maximum at $\omega_{\iny,}$ which as a random map in $\Z$ is tight.
\end{lem}
%%%%%%%%%%%%%%%%%%%%%%%%%%%%%%%%%%%%%%%%%%%%%%%%%%%%%%%%%%%%%%%%%%%%%%%%%%%%%%%%%%%%%%%%%%%%%%%%%%%%%%%%%%%%%%%%%%%%%%%%%%%%%%%%%%%%%%%%%%%%%%%%%%%%%%%%%%%%%%%%%%%%%%%%%%%%

%%%%%%%%%%%%%%%%%%%%%%%%%%%%%%%%%%%%%%%%%%%%%%%%%%%%%%%%%%%%%%%%%%%%%%%%%%%%%%%%%%%%%%%%%%%%%%%%%%%%%%%%%%%%%%%%%%%%%%%%%%%%%%%%%%%%%%%%%%%%%%%%%%%%%%%%%%%%%%%%%%%%%%%%%%%%
\begin{proof}[Proof of Lemma \ref{lem:regularity.argmax}] From Condition A$',$ each side of the random walk $\cC(\z)$ has increments supported on $\R,$ thus the first assertion on the continuity of the map $\z\to \cC(z)$ follows trivially since the domain space of this map is restricted to only the integers $\Z$ ($\ep-\delta$ definition of continuity). To prove the remaining assertions note that from Condition A$'$, Condition B and Condition E(ii), together with the definition (\ref{def:cCz}), we have that $\cC_{\iny}(\z)$ has i.i.d increments with a negative drift of $-\xi_{\iny}^2.$ Consequently, we have $\cC_{\iny}(\z)\to-\iny,$ as $\z\to \iny$ almost surely (strong law of large numbers). Using elementary properties of random walks, this implies that $\max_{\z}\cC_{\iny}(\z)<\iny,$ a.s. (follows from the Hewitt-Savage $0$-$1$ law, see, e.g. (1.1) and (1.2) on Page 172, 173 of \cite{durrett2010probability}). Additionally $\om_{\iny}\ge 0,$ from the construction of $\cC_{\iny}(\z).$ Thus, we have $0\le \om_{\iny}<\iny,$ a.s.  which directly implies that when $\om_{\iny}$ is well defined (unique) then it must be tight. To show that $\om_{\iny}$ is unique, note that since by assumption (Condition A$'$) the increments are continuously distributed and supported on $\R$, therefore $\max\cC_{\iny}(\z)$ is continuously distributed on $(0,\iny),$ with some additional probability mass at the singleton zero. Hence, the probability of $\max\cC_{\iny}(\z)$ attaining any two identical values is zero. Consequently $\om_{\iny}$ is unique a.s.
\end{proof}
%%%%%%%%%%%%%%%%%%%%%%%%%%%%%%%%%%%%%%%%%%%%%%%%%%%%%%%%%%%%%%%%%%%%%%%%%%%%%%%%%%%%%%%%%%%%%%%%%%%%%%%%%%%%%%%%%%%%%%%%%%%%%%%%%%%%%%%%%%%%%%%%%%%%%%%%%%%%%%%%%%%%%%%%%%%%

\bc$\rule{4in}{0.1mm}$\ec

%%%%%%%%%%%%%%%%%%%%%%%%%%%%%%%%%%%%%%%%%%%%%%%%%%%%%%%%%%%%%%%%%%%%%%%%%%%%%%%%%%%%%%%%%%%%%%%%%%%%%%%%%%%%%%%%%%%%%%%%%%%%%%%%%%%%%%%%%%%%%%%%%%%%%%%%%%%%%%%%%%%%%%%%%%%%
\begin{lem}\label{lem:Capprox} Let $\cC(\tau,\theta_1,\theta_2)$ be as defined in (\ref{def:cC}) and suppose Condition A, B and C.2 hold. Additionally assume sequence $r_T$ of Condition C.2 satisfies (\ref{eq:30}). Then, for any $c_u>0,$ we have,
	\benr
	\sup_{\tau\in\cG(c_uT^{-1}\xi^{-2},0)}\big|\cC(\tau,\h\theta_1,\h\theta_2)-\cC(\tau,\theta_1^0,\theta_2^0)\big|=o_p(1).\nn
	\eenr
\end{lem}
%%%%%%%%%%%%%%%%%%%%%%%%%%%%%%%%%%%%%%%%%%%%%%%%%%%%%%%%%%%%%%%%%%%%%%%%%%%%%%%%%%%%%%%%%%%%%%%%%%%%%%%%%%%%%%%%%%%%%%%%%%%%%%%%%%%%%%%%%%%%%%%%%%%%%%%%%%%%%%%%%%%%%%%%%%%%

%%%%%%%%%%%%%%%%%%%%%%%%%%%%%%%%%%%%%%%%%%%%%%%%%%%%%%%%%%%%%%%%%%%%%%%%%%%%%%%%%%%%%%%%%%%%%%%%%%%%%%%%%%%%%%%%%%%%%%%%%%%%%%%%%%%%%%%%%%%%%%%%%%%%%%%%%%%%%%%%%%%%%%%%%%%%
\begin{proof}[Proof of Lemma \ref{lem:Capprox}] By proceeding similar to (\ref{eq:11}), we have under Condition C.2 and (\ref{eq:30}) that,
	\benr\label{eq:33}
	\|\h\eta-\eta^0\|_1&\le& c_{u1}\surd{s}\|\h\eta-\eta^0\|_2\nn\\
	&=&\begin{cases} \frac{o(1)\xi}{\log^{1/2} (p\vee T)}, & {\rm for\,\, subgaussian\,\,case,}\\ \frac{o(1)\xi}{\log (p\vee T)}, & {\rm for\,\, subexpoential\,\,case}, \end{cases}
	\eenr
	with probability at least $1-\pi_T.$ Consider any $\tau\ge \tau^0$ and define the following,
	\benr
	R_1=2\sum_{t=\tau^0+1}^{\tau}\vep_t^T(\h\eta-\eta^0),\quad{\rm and}\quad R_2=(\tau-\tau^0)\big(\|\h\eta\|_2^2-\|\eta^0\|_2^2).\nn
	\eenr	
	Then we have the following algebraic expansion,	
	\benr\label{eq:34}
	\cC(\tau,\h\theta_1,\h\theta_2)-\cC(\tau,\theta_1^0,\theta_2^0)&=&
	T\big\{Q(\tau^0,\h\theta_1,\h\theta_2)\nn\\
	&&-Q(\tau,\h\theta_1,\h\theta_2)\big\}-T\big\{Q(\tau^0,\theta_1^0,\theta_2^0)-Q(\tau,\theta_1^0,\theta_2^0)\big\}\nn\\
	&=& 2\sum_{t=\tau^0+1}^{\tau}\vep_t^T(\h\eta-\eta^0)- (\tau-\tau^0)\big(\|\h\eta\|_2^2-\|\eta^0\|_2^2)\nn\\
	&=&R_1-R_2
	\eenr
	Next we provide uniform bounds on the terms $R_1$ and $R_2$ of (\ref{eq:34}). Consider,
	\benr\label{eq:36}
	\sup_{\substack{\tau\in \cG(c_u T^{-1}\xi^{-2},0);\\\tau\ge\tau^0}}|R_1|&\le& 2\sup_{\substack{\tau\in \cG(c_uT^{-1}\xi^{-2},0);\\\tau\ge\tau^0}}\big\|\sum_{t=\tau^0+1}^{\tau}\vep_t\big\|_{\iny}\|\h\eta-\eta^0\|_1\nn\\
	&\le& \begin{cases}\surd c_u\si\xi^{-1}\log^{1/2}(p\vee T)\|\h\eta-\eta^0\|_1 & {\rm for\,\, subG\,\,case} \\ \surd c_u\si\xi^{-1}\log(p\vee T)\|\h\eta-\eta^0\|_1 & {\rm for\,\, subE\,\,case} \end{cases}\nn\\
	&=&o(1),
	\eenr
	with probability at least $1-o(1).$ Here the second inequality follows from Lemma \ref{lem:nearoptimalcross} and Lemma \ref{lem:nearoptimalcross.subE.special}. The final equality follows from an application of (\ref{eq:33}). Next consider term $R_2$ of (\ref{eq:34}).
	\benr\label{eq:37}
	\sup_{\substack{\tau\in \cG(c_u T^{-1}\xi^{-2},0);\\\tau\ge\tau^0}}|R_2|&\le& c_u\xi^{-2}\big|\|\h\eta\|_2^2-\|\eta^0\|_2^2\big|= c_u\xi^{-2}\big|\|\h\eta-\eta^0\|_2^2+2(\h\eta-\eta^0)^T\eta^0\big|\nn\\
	&\le& c_u\xi^{-2}\|\h\eta-\eta^0\|_2^2+2c_u\xi^{-1}\|\h\eta-\eta^0\|_2=o_p(1),
	\eenr
	where the final equality follows from (\ref{eq:33}). Applying (\ref{eq:36}) and (\ref{eq:37}) in the expression (\ref{eq:34}) yields,
	\benr
	\sup_{\substack{\tau\in \cG(c_u T^{-1}\xi^{-2},0);\\\tau\ge\tau^0}}\big|\cC(\tau,\h\theta_1,\h\theta_2)-\cC(\tau,\theta_1^0,\theta_2^0)\big|&\le& \sup_{\substack{\tau\in \cG(c_u T^{-1}\xi^{-2},0);\\\tau\ge\tau^0}}|R_1|\nn\\
	&&+ \sup_{\substack{\tau\in \cG(c_u T^{-1}\xi^{-2},0);\\\tau\ge\tau^0}}|R_2|\nn\\
	&=&o_p(1)\nn
	\eenr
	Corresponding bound for the mirroring case of $\{\tau\le\tau^0\}$ can be obtained via symmetrical arguments.  This completes the proof of this lemma.
\end{proof}
%%%%%%%%%%%%%%%%%%%%%%%%%%%%%%%%%%%%%%%%%%%%%%%%%%%%%%%%%%%%%%%%%%%%%%%%%%%%%%%%%%%%%%%%%%%%%%%%%%%%%%%%%%%%%%%%%%%%%%%%%%%%%%%%%%%%%%%%%%%%%%%%%%%%%%%%%%%%%%%%%%%%%%%%%%%%

\bc$\rule{4in}{0.1mm}$\ec

%%%%%%%%%%%%%%%%%%%%%%%%%%%%%%%%%%%%%%%%%%%%%%%%%%%%%%%%%%%%%%%%%%%%%%%%%%%%%%%%%%%%%%%%%%%%%%%%%%%%%%%%%%%%%%%%%%%%%%%%%%%%%%%%%%%%%%%%%%%%%%%%%%%%%%%%%%%%%%%%%%%%%%%%%%%%
\subsection{Proofs of Section \ref{sec:algorithm}}\label{app:sec.alg}
%%%%%%%%%%%%%%%%%%%%%%%%%%%%%%%%%%%%%%%%%%%%%%%%%%%%%%%%%%%%%%%%%%%%%%%%%%%%%%%%%%%%%%%%%%%%%%%%%%%%%%%%%%%%%%%%%%%%%%%%%%%%%%%%%%%%%%%%%%%%%%%%%%%%%%%%%%%%%%%%%%%%%%%%%%%%

The proof of Theorem \ref{thm:alg1.optimal} first requires some preliminary work, in particular we first need to examine the behavior of the estimates $\tilde\theta_1(\tau),$ and $\tilde\theta_2(\tau),$ uniformly over a collection of values of $\tau.$ This is provided in the following theorem.

%%%%%%%%%%%%%%%%%%%%%%%%%%%%%%%%%%%%%%%%%%%%%%%%%%%%%%%%%%%%%%%%%%%%%%%%%%%%%%%%%%%%%%%%%%%%%%%%%%%%%%%%%%%%%%%%%%%%%%%%%%%%%%%%%%%%%%%%%%%%%%%%%%%%%%%%%%%%%%%%%%%%%%%%%%%
\begin{thm}\label{thm:unifmean} Suppose model (\ref{model:rvmcp}), let $0\le u_T\le 1$ be any non-negative sequence, $\psi=\|\eta^0\|_{\iny},$ and for any constants $c_u,c_{u1}>0,$ let,
	\benr\label{eq:la}
	\la:=\la_1=\la_2=8\max\Big[\si\Big\{\frac{2c_{u1}\log(p\vee T)}{c_uTl_T}\Big\}^{\frac{1}{2}},\,\,\frac{u_T\psi}{c_ul_T}\Big].
	\eenr
	Additionally let either one of the following two conditions hold.\\~
	(a) Condition A(I) (subgaussian), B, and D holds. \\~
	(b) Condition A(II) (subexponential), B, and D holds and $c_uTl_T\ge \log (p\vee T).$\\~
	Then, the estimates  $\tilde\theta_1(\tau)$ and $\tilde\theta_2(\tau)$ satisfy the following two results with probability at least $1-\pi_T.$\\~
	(i) For any $\tau\in \cG(u_T,0),$ with $\tau\wedge(T-\tau)\ge c_uTl_T,$ we have $\big\|\big(\tilde\theta_1(\tau)\big)_{S_1^c}\big\|_1\le 3\big\|\big(\tilde\theta_1(\tau)-\theta_1^0\big)_{S_1}\big\|_1,$ and $\big\|\big(\tilde\theta_2(\tau)\big)_{S_2^c}\big\|_1\le 3\big\|\big(\tilde\theta_2(\tau)-\theta_2^0\big)_{S_2}\big\|_1.$\\~
	(ii) The following bound is satisfied,
	\benr
	\sup_{\substack{\tau\in\cG(u_T,0)\\ \tau\wedge(T-\tau)\ge c_uTl_T}} \|\tilde\theta_1(\tau)-\theta_1^0\|_2\vee \|\tilde\theta_2(\tau)-\theta_2^0\|_2\le 6\surd{s}\la.\nn
	\eenr
	Here $\pi_T=2\exp\{-(c_{u1}-2)\log(p\vee T)\}$ under conditions (a) and $\pi_T=2\exp\big\{-\big(c_{u2}-2\big)\log (p\vee T)\big\},$ under conditions (b), where $c_{u2}=c_{u1}\wedge \surd(c_uc_{u1}/2).$	
\end{thm}
%%%%%%%%%%%%%%%%%%%%%%%%%%%%%%%%%%%%%%%%%%%%%%%%%%%%%%%%%%%%%%%%%%%%%%%%%%%%%%%%%%%%%%%%%%%%%%%%%%%%%%%%%%%%%%%%%%%%%%%%%%%%%%%%%%%%%%%%%%%%%%%%%%%%%%%%%%%%%%%%%%%%%%%%%%%

%%%%%%%%%%%%%%%%%%%%%%%%%%%%%%%%%%%%%%%%%%%%%%%%%%%%%%%%%%%%%%%%%%%%%%%%%%%%%%%%%%%%%%%%%%%%%%%%%%%%%%%%%%%%%%%%%%%%%%%%%%%%%%%%%%%%%%%%%%%%%%%%%%%%%%%%%%%%%%%%%%%%%%%%%%%
\begin{proof}[Proof of Theorem \ref{thm:unifmean}]
	Consider any $\tau\in\cG(u_T,0),$ satisfying $\tau\ge c_uTl_T,$ and without loss of generality assume $\tau\ge\tau^0.$ The mirroring case of $\tau\le\tau^0$ can be proved using symmetrical arguments. An algebraic rearrangement of the elementary inequality $\big\|\bar x_{(0:\tau]}-\tilde\theta_1(\tau)\big\|^2_2+\la_1\|\tilde\theta_1(\tau)\|_1\le \big\|\bar x_{(0:\tau]}-\theta_1^0\big\|^2_2+\la_1\|\theta_1^0\|_1$ yields,
	\benr\label{eq:20}
	\big\|\tilde\theta_1(\tau)-\theta_1^0\big\|_2^2+\la_1\big\|\tilde\theta_1(\tau)\big\|_1&\le& \la_1\big\|\theta^0_1\big\|_1+ 2\sum_{t=1}^{\tau}\tilde\vep_t^T(\tilde\theta_1(\tau)-\theta_1^0),\nn\\
	&=& \la_1\big\|\theta^0_1\big\|_1+2\sum_{t=1}^{\tau}\vep_t^T(\tilde\theta_1(\tau)-\theta_1^0)\nn\\
	&&-2\frac{(\tau-\tau^0)}{\tau}(\theta_1^0-\theta_2^0)^T(\tilde\theta_1(\tau)-\theta_1^0)\nn\\
	&\le &\la_1\big\|\theta^0_1\big\|_1+2\Big\|\sum_{t=1}^{\tau}\vep_t\Big\|_{\iny}\big\|\tilde\theta_1(\tau)-\theta_1^0\big\|_1\nn\\
	&&-\frac{2u_T}{c_ul_T}\psi\big\|\tilde\theta_1(\tau)-\theta_1^0\big\|_1
	\eenr
	where $\tilde\vep_t=\vep_t-\theta^0_1,$ for $t=1,...,\tau^0,$ and $\tilde\vep_t=\vep_t-(\theta_1^0-\theta^0_2),$ $t=\tau^0+1,...,\tau.$ Now using the bound of Lemma \ref{lem:1.to.tau.bound} we have that,
	\benr
	\frac{2}{\tau}\Big\|\sum_{t=1}^{\tau}\vep_t\Big\|_{\iny}\le 2\surd(2c_{u1}/c_{u})\si\Big\{\frac{\log(p\vee T)}{Tl_T}\Big\}^{\frac{1}{2}}\nn
	\eenr	
	with probability at least  $1-2\exp\{-(c_{u1}-2)\log(p\vee T)\},$ or $1-\exp\big\{-(c_{u2}-2)\log (p\vee T)\big\},$ under the subgaussian or subexponential setting, respectively. Consequently, upon choosing,
	\benr
	\la=\max\Big\{4\surd(2c_{u1}/c_{u})\si\Big\{\frac{\log(p\vee T)}{Tl_T}\Big\}^{\frac{1}{2}},\,\,\frac{4u_T\psi}{c_ul_T}\Big\},\nn
	\eenr
	and substituting these bounds in (\ref{eq:20}) we obtain,
	\benr\label{eq:21}
	\big\|\tilde\theta_1(\tau)-\theta_1^0\big\|_2^2+\la_1\big\|\tilde\theta_1(\tau)\big\|_1\le\la_1\big\|\theta^0_1\big\|_1+\la\big\|\tilde\theta_1(\tau)-\theta_1^0\big\|_1,
	\eenr
	with probability at least  $1-2\exp\{-(c_{u1}-2)\log(p\vee T)\},$ or $1-\exp\big\{-(c_{u2}-2)\log (p\vee T)\big\},$ under the subgaussian or subexponential setting, respectively. Now choosing $\la_1= 2\la,$ leads to $\|\big(\tilde\theta_1(\tau)\big)_{S_1^c}\|_1\le 3\|\big(\tilde\theta_1(\tau)-\theta_1^0\big)_{S_1}\|_1,$ which proves part (i) of this theorem. From inequality (\ref{eq:21}) we also have that,
	\benr
	\|\tilde\theta_1(\tau)-\theta_1^0\|_2^2\le \frac{3}{2}\la_1\|\tilde\theta_1(\tau)-\theta_1^0\|_1\le 6\la_1\surd{s}\|\tilde\theta_1(\tau)-\theta_1^0\|_2
	\eenr
	This directly implies that  $\|\tilde\theta_1(\tau)-\theta_1^0\|_2\le 6\la_1\surd{s},$ where we have used $\|\tilde\theta_1(\tau)-\theta_1^0\|_1\le 4\sqrt{s}\|\tilde\theta_1(\tau)-\theta_1^0\|_2,$ which follows in turn from $\|\big(\tilde\theta_1(\tau)\big)_{S_1^c}\|_1\le 3\|\big(\tilde\theta_1(\tau)-\theta_1^0\big)_{S_1}\|_1.$ To finish the proof of this part recall that the only stochastic bound used here is the uniform bound over $\cG(u_T,0)$ of Lemma \ref{lem:1.to.tau.bound}, consequently the final bound also holds uniformly over the same collection. Similar results for $\tilde\theta_2(\tau)$ follow analogously. This result can alternatively be proved using the properties of the soft-thresholding operator $k_{\la}(\cdotp),$ by building uniform versions of arguments such as those in \cite{rothman2009generalized}, or \cite{kaul2017structural}.
\end{proof}
%%%%%%%%%%%%%%%%%%%%%%%%%%%%%%%%%%%%%%%%%%%%%%%%%%%%%%%%%%%%%%%%%%%%%%%%%%%%%%%%%%%%%%%%%%%%%%%%%%%%%%%%%%%%%%%%%%%%%%%%%%%%%%%%%%%%%%%%%%%%%%%%%%%%%%%%%%%%%%%%%%%%%%%%%%%

Following is another preliminary result required to prove Theorem \ref{thm:alg1.optimal}. This results uses Theorem \ref{thm:unifmean} to provide the rate of convergence of $\check\theta_1=\tilde\theta_1(\check\tau),$ and $\check\theta_2=\tilde\theta_2(\check\tau),$ of Step 1 of Algorithm 1.

%%%%%%%%%%%%%%%%%%%%%%%%%%%%%%%%%%%%%%%%%%%%%%%%%%%%%%%%%%%%%%%%%%%%%%%%%%%%%%%%%%%%%%%%%%%%%%%%%%%%%%%%%%%%%%%%%%%%%%%%%%%%%%%%%%%%%%%%%%%%%%%%%%%%%%%%%%%%%%%%%%%%%%%%%%%
\begin{lem}\label{lem:step1mean} Suppose model (\ref{model:rvmcp}) and assume the following,
	\benr
	\frac{c_{u}\si}{\xi}\Big\{\frac{s\log (p\vee T)}{Tl_T}\Big\}^{\frac{1}{2}}\le c_{u1}\nn
	\eenr	
	for an appropriately chosen small enough constant $c_{u1}>0.$ Additionally, let either one of the following two conditions hold.\\~
	(a) Condition A(I) (subgaussian), B, and D holds and $T\ge 2.$ \\~
	(b) Condition A(II) (subexponential), B, and D holds, $T\ge 2$ and $c_uTl_T\ge \log (p\vee T).$\\~
	Then, the mean estimates $\check\theta_1=\tilde\theta_1(\check\tau),$ and  $\check\theta_2=\tilde\theta_2(\check\tau),$ satisfy the following, with probability $1-o(1).$\\~
	(i)  $\big\|\big(\check\theta_1\big)_{S_1^c}\big\|_1\le 3\big\|\big(\check\theta_1-\theta_1^0\big)_{S_1}\big\|_1,$ and $\big\|\big(\check\theta_2\big)_{S_2^c}\big\|_1\le 3\big\|\big(\check\theta_2-\theta_2^0\big)_{S_2}\big\|_1.$\\~
	(ii) The following bound is satisfied,
	\benr
	\|\check\theta_1-\theta_1^0\|_2\vee \|\check\theta_2-\theta_2^0\|_2\le c_{u1}\xi.\nn
	\eenr
	Consequently, the mean estimates $\check\theta_1$ and $\check\theta_2$ satisfy Condition C.1.
\end{lem}
%%%%%%%%%%%%%%%%%%%%%%%%%%%%%%%%%%%%%%%%%%%%%%%%%%%%%%%%%%%%%%%%%%%%%%%%%%%%%%%%%%%%%%%%%%%%%%%%%%%%%%%%%%%%%%%%%%%%%%%%%%%%%%%%%%%%%%%%%%%%%%%%%%%%%%%%%%%%%%%%%%%%%%%%%%%

%%%%%%%%%%%%%%%%%%%%%%%%%%%%%%%%%%%%%%%%%%%%%%%%%%%%%%%%%%%%%%%%%%%%%%%%%%%%%%%%%%%%%%%%%%%%%%%%%%%%%%%%%%%%%%%%%%%%%%%%%%%%%%%%%%%%%%%%%%%%%%%%%%%%%%%%%%%%%%%%%%%%%%%%%%%
\begin{proof}[Proof of Lemma \ref{lem:step1mean}] From Condition F we have that $\check\tau$ satisfies,
	\benr\label{eq:28}
	(i)\,\,\check\tau\wedge(T-\check\tau)\ge c_{u}Tl_T,\quad{\rm and}\quad (ii)\,\,|\check\tau-\tau^0|\le Tu_T,
	\eenr	
	where $u_T=c_{u1}l_T\xi\big/(\surd(2s)\psi),$ i.e., $\check\tau\in\cG(u_T,0).$ Now upon choosing,
	\benr\label{eq:la.step1.choice}
	\la\,\,{\rm as\, prescribed\,\, in\,\, (\ref{eq:la})\,\, with}\,\,u_T=c_{u1}l_T\xi\big/(\surd(2s)\psi),
	\eenr
	and applying Theorem \ref{thm:unifmean} we obtain the following two results that hold with probability $1-o(1),$ for both subgaussian and subexponential cases. First,  $\big\|\big(\check\theta_1\big)_{S_1^c}\big\|_1\le 3\big\|\big(\check\theta_1-\theta_1^0\big)_{S_1}\big\|_1,$ and $\big\|\big(\check\theta_2\big)_{S_2^c}\big\|_1\le 3\big\|\big(\check\theta_2-\theta_2^0\big)_{S_2}\big\|_1.$ Second,
	\benr\label{eq:22}
	\|\check\theta_1-\theta_1^0\|_2\vee \|\check\theta_2-\theta_2^0\|_2&\le& \max\Big[c_u\si\Big\{\frac{s\log(p\vee T)}{Tl_T}\Big\}^{\frac{1}{2}},\,\,c_u\frac{u_T\surd{(2s)}\psi}{l_T}\Big],\nn\\
	&=& \xi\Big[c_u\frac{\si}{\xi}\Big\{\frac{s\log(p\vee T)}{Tl_T}\Big\}^{\frac{1}{2}},\,\,c_u\frac{u_T\surd{(2s)}\psi}{l_T\xi}\Big]\nn\\
	&=& \xi\Big[R_1,R_2\Big]
	\eenr	
	Here the first equality is simply an algebraic manipulation. Now for a suitable chosen $c_{u1}>0,$ we have from assumption (\ref{eq:23}) that,
	\benr
	\frac{c_u\si}{\xi}\Big\{\frac{s\log (p\vee T)}{Tl_T}\Big\}^{\frac{1}{2}}\le c_{u1},\nn
	\eenr
	which provides a bound for term $R_1$ on the RHS of (\ref{eq:22}). Next we bound term $R_2$ of the same expression. Substituting the choice of $u_T$ from (\ref{eq:28}) in term $R_2$, together with the earlier bound for $R_1,$ we obtain,
	\benr\|\check\theta_1-\theta_1^0\|_2\vee \|\check\theta_2-\theta_2^0\|_2\le
	\xi\Big[R_1,R_2\Big]\le c_{u1}\xi,\nn
	\eenr
	with probability $1-o(1).$ Thereby all requirement of Condition C.1 are met and this completes the proof of the lemma.
\end{proof}
%%%%%%%%%%%%%%%%%%%%%%%%%%%%%%%%%%%%%%%%%%%%%%%%%%%%%%%%%%%%%%%%%%%%%%%%%%%%%%%%%%%%%%%%%%%%%%%%%%%%%%%%%%%%%%%%%%%%%%%%%%%%%%%%%%%%%%%%%%%%%%%%%%%%%%%%%%%%%%%%%%%%%%%%%%%

\bc$\rule{4in}{0.1mm}$\ec

%%%%%%%%%%%%%%%%%%%%%%%%%%%%%%%%%%%%%%%%%%%%%%%%%%%%%%%%%%%%%%%%%%%%%%%%%%%%%%%%%%%%%%%%%%%%%%%%%%%%%%%%%%%%%%%%%%%%%%%%%%%%%%%%%%%%%%%%%%%%%%%%%%%%%%%%%%%%%%%%%%%%%%%%%%%
\begin{proof}[Proof of Theorem \ref{thm:al1.near.optimal}]
	The proof of this theorem is a fairly direct consequence of Lemma \ref{lem:step1mean}, Theorem \ref{thm:nearoptimalcp} and Theorem \ref{thm:subE.nearoptimal.special}. First note that Lemma \ref{lem:step1mean} establishes that   $\check\theta_1=\tilde\theta_1(\check\tau),$ and $\check\theta_2=\tilde\theta_2(\check\tau)$ of Step 1 of Algorithm 1 will satisfy Condition C.1. The availability of Condition C.1 on these mean estimates, now allows us to directly obtain the rate of convergence of $\h\tau$ of Step 1 of Algorithm 1 using Part (i) of Theorem \ref{thm:nearoptimalcp} for the subgaussian case and Theorem \ref{thm:subE.nearoptimal.special} for the subexponential case. Specifically, under the assumed conditions of the corresponding results, we have,
	\benr
	|\h\tau-\tau^0|\le \begin{cases}c_{u}\si^2\xi^{-2}s\log(p\vee T) & {\rm for\,\, subgaussian\,\, case }\\ c_{u}\si^2\xi^{-2}s\log^2(p\vee T) & {\rm for\,\, subexponential\,\, case}\end{cases}\nn
	\eenr	
	with probability at least $1-o(1).$ This completes the proof of this theorem.
\end{proof}
%%%%%%%%%%%%%%%%%%%%%%%%%%%%%%%%%%%%%%%%%%%%%%%%%%%%%%%%%%%%%%%%%%%%%%%%%%%%%%%%%%%%%%%%%%%%%%%%%%%%%%%%%%%%%%%%%%%%%%%%%%%%%%%%%%%%%%%%%%%%%%%%%%%%%%%%%%%%%%%%%%%%%%%%%%%

\bc$\rule{4in}{0.1mm}$\ec

%%%%%%%%%%%%%%%%%%%%%%%%%%%%%%%%%%%%%%%%%%%%%%%%%%%%%%%%%%%%%%%%%%%%%%%%%%%%%%%%%%%%%%%%%%%%%%%%%%%%%%%%%%%%%%%%%%%%%%%%%%%%%%%%%%%%%%%%%%%%%%%%%%%%%%%%%%%%%%%%%%%%%%%%%%%
\begin{proof}[Proof of Theorem \ref{thm:alg1.optimal}] The starting point for this proof is the near optimality of $\h\tau$ of Step 1 of Algorithm 1, which is provided by Theorem \ref{thm:al1.near.optimal} or alternatively assumed in Step 1 of Algorithm 2. From (\ref{eq:24}) of Theorem \ref{thm:al1.near.optimal}, we have that $\h\tau\in\cG(u_T,0),$ with probability $1-o(1),$ where,
	\benr\label{eq:ut}
	Tu_T=\begin{cases}c_{u}\si^2\xi^{-2}s\log(p\vee T) & {\rm under\,\, subgaussian\,\,case}\\ c_{u}\si^2\xi^{-2}s\log^2(p\vee T) &  {\rm under\,\, subexponential\,\,case.}\end{cases}
	\eenr
	Moreover, by assumption we have $c_uTl_T\ge s\log (p\vee T),$ and $c_uTl_T\ge s\log^2 (p\vee T),$ in the subgaussian and subexponential case respectively. Thus, with the same probability at above we also have $\h\tau\wedge(T-\h\tau)\ge c_uTl_T.$ Now upon choosing,
	\benr\label{eq:la.step2.choice}
	\la\,\,{\rm as\, prescribed\,\, in\,\, (\ref{eq:la})\,\, with}\,\,u_T\,\,{\rm as\,\,in\,\,(\ref{eq:ut})},
	\eenr
	we obtain from Theorem \ref{thm:unifmean} that  $\h\theta_1=\tilde\theta_1(\h\tau),$ and $\h\theta_2=\tilde\theta_2(\h\tau)$ of Step 2 of Algorithm 1 satisfies,
	\benr\label{eq:25}
	\big\|\h\theta_1-\theta_1^0\big\|_2\vee\big\|\h\theta_2-\theta_2^0\big\|_2&\le& c_u\surd{s}\max\Big[\si\Big\{\frac{\log(p\vee T)}{Tl_T}\Big\}^{\frac{1}{2}},\,\,\frac{u_T\psi}{l_T}\Big]\nn\\
	&=&\frac{\xi}{\surd{(s\log (p\vee T))}} \max\Big[c_u\si\Big\{\frac{s\log(p\vee T)}{\xi\surd{\big(Tl_T\big)}}\Big\},\,\,c_u\surd\{s\log (p\vee T)\}\frac{\surd{s}\psi}{l_T\xi}u_T\Big]\nn\\
	&=&\frac{\xi}{\surd{(s\log (p\vee T))}}\max\big[R_1,R_2\big]
	\eenr
	with probability at $1-o(1).$ Here the first equality is simply an algebraic manipulation. From assumption (\ref{eq:16}) and (\ref{eq:17}), we have that $R_1\le c_{u1},$ where $c_{u1}>0,$ is an appropriately chosen small enough constant. Next consider term $R_2$ of (\ref{eq:25}) under the subgaussian and subexponential cases separately.\\~
	
	\vspace{-3mm}
	{\bf Case I (subgaussian setting):} Substituting $u_T$ from (\ref{eq:ut}) in term $R_2$ we obtain,
	\benr
	c_u\surd\{s\log (p\vee T)\}\frac{\surd{s}\psi}{l_T\xi}u_T&=&c_u\surd\{s\log (p\vee T)\}\frac{\surd{s}\psi}{\xi}\Big\{\frac{\si^2}{\xi^2}\frac{s\log (p\vee T)}{Tl_T}\Big\}\nn\\
	&\le& c_u\Big\{\frac{\si}{\xi}\frac{s\log (p\vee T)}{\surd(Tl_T)}\Big\}^2
	\le c_uc_{u1}^2\le c_{u1}.\nn
	\eenr
	Here the first inequality follows from the assumption $(\psi\big/\xi)\le \surd\{\log(p\vee T)\}.$ The second inequality follows from assumption (\ref{eq:16}). Substituting the bounds for terms $R_1$ and $R_2$ back in (\ref{eq:25}) yields,
	\benr\label{eq:27}
	\big\|\h\theta_1-\theta_1^0\big\|_2\vee\big\|\h\theta_2-\theta_2^0\big\|_2
	\le \frac{c_{u1}\xi}{\surd{(s\log (p\vee T))}},
	\eenr 		
	for a suitably chosen small enough $c_{u1}>0,$ with probability $1-o(1).$ This provides the bound required for the validity of Condition C.2. The first requirement of Condition C.2 is directly satisfied by using Theorem \ref{thm:unifmean}. Thus, the estimates $\h\theta_1$ and $\h\theta_2$ of Step 2 of Algorithm 1 satisfy all requirement of  Condition C.2 and now the statement of this theorem follows from the result of Theorem \ref{thm:cpoptimal}. \\~
	
	\vspace{-3mm}
	{\bf Case 2: subexponential setting}, Substituting $u_T$ from (\ref{eq:ut}) in term $R_2$ we obtain,
	\benr
	c_u\surd\{s\log (p\vee T)\}\frac{\surd{s}\psi}{l_T\xi}u_T&=&c_u\surd\{s\log (p\vee T)\}\frac{\surd{s}\psi}{\xi}\Big\{\frac{\si^2}{\xi^2}\frac{s\log^2 (p\vee T)}{Tl_T}\Big\}\nn\\
	&\le& c_u\Big\{\frac{\si}{\xi}\frac{s\log^{3/2} (p\vee T)}{\surd(Tl_T)}\Big\}^2
	\le c_uc_{u1}^2\le c_{u1}.\nn
	\eenr
	Here the first inequality follows from the assumption $(\psi\big/\xi)\le \surd\{\log(p\vee T)\}.$ The second inequality follows from assumption (\ref{eq:17}). Substituting the bounds for terms $R_1$ and $R_2$ back in (\ref{eq:25}) yields the same bound as (\ref{eq:27}). Thus, the estimates $\h\theta_1$ and $\h\theta_2$ of Step 2 of Algorithm 1 satisfy all requirement of  Condition C.2 and now the statement of this theorem follows from the result of Theorem \ref{thm:cpoptimal}. This completes the proof of the theorem.
\end{proof}
%%%%%%%%%%%%%%%%%%%%%%%%%%%%%%%%%%%%%%%%%%%%%%%%%%%%%%%%%%%%%%%%%%%%%%%%%%%%%%%%%%%%%%%%%%%%%%%%%%%%%%%%%%%%%%%%%%%%%%%%%%%%%%%%%%%%%%%%%%%%%%%%%%%%%%%%%%%%%%%%%%%%%%%%%%%

\bc$\rule{4in}{0.1mm}$\ec

\subsection{Deviation bounds used in the proofs of Section \ref{sec:mainresults}}\label{app:sec.main}

%%%%%%%%%%%%%%%%%%%%%%%%%%%%%%%%%%%%%%%%%%%%%%%%%%%%%%%%%%%%%%%%%%%%%%%%%%%%%%%%%%%%%%%%%%%%%%%%%%%%%%%%%%%%%%%%%%%%%%%%%%%%%%%%%%%%%%%%%%%%%%%%%%%%%%%%%%%%%%%%%%%%%%%%%%%%
\begin{lem}\label{lem:nearoptimalcross} Suppose Condition A(I) (subgaussian setting) and B holds and let $0\le v_T\le u_T\le 1,$ be any non-negative sequences. Then for any $c_u\ge 1,$ we have,
	\benr\label{eq:6}
	\sup_{\substack{\tau\in\cG(u_T,v_T);\\\tau\ge\tau^0}}\Big\|\sum_{t=\tau^0+1}^{\tau}\vep_{t}\Big\|_{\iny}\le  \surd(2c_u)\si\{Tu_T\log(p\vee T)\}^{\frac{1}{2}}
	\eenr
	with probability at least $1-2\exp\{-(c_u-2)\log(p\vee T)\}.$ Alternatively, suppose Condition A(II) (subexponential setting) and B hold. Additionally assume that $T\ge\log (p\vee T),$ and $\log (p\vee T)\le Tv_T\le Tu_T.$ Then, for any constant $c_u\ge 1,$ the same bound (\ref{eq:6}) holds with probability at least $1-\exp\big\{-\big(\surd(c_u/2)-2\big)\log (p\vee T)\big\}.$
\end{lem}
%%%%%%%%%%%%%%%%%%%%%%%%%%%%%%%%%%%%%%%%%%%%%%%%%%%%%%%%%%%%%%%%%%%%%%%%%%%%%%%%%%%%%%%%%%%%%%%%%%%%%%%%%%%%%%%%%%%%%%%%%%%%%%%%%%%%%%%%%%%%%%%%%%%%%%%%%%%%%%%%%%%%%%%%%%%%

%%%%%%%%%%%%%%%%%%%%%%%%%%%%%%%%%%%%%%%%%%%%%%%%%%%%%%%%%%%%%%%%%%%%%%%%%%%%%%%%%%%%%%%%%%%%%%%%%%%%%%%%%%%%%%%%%%%%%%%%%%%%%%%%%%%%%%%%%%%%%%%%%%%%%%%%%%%%%%%%%%%%%%%%%%%%
\begin{proof}[Proof of Lemma \ref{lem:nearoptimalcross}] First note that without loss of generality we can assume $v_T\ge (1/T).$ This follows since the only additional element in the set $\cG(u_T,0)$ in comparison to $\cG(u_T,(1/T))$ is $\tau^0,$ and at this value, the sum of interest is over indices in an empty set and is thus trivially zero.
	
	We begin with subgaussian case. Consider any $j\in\{1,2,...,p\}$ and any $\tau>\tau^0,$ and note that  $\sum_{t=\tau^0+1}^{\tau}\vep_{tj}\sim {\rm subG}(\la),$ with $\la=\si\surd(\tau-\tau^0).$ This follows since $\vep_{tj}$ are independent over $t=1,...,T$ (see, Part (iii) of Lemma \ref{lem:lcsubG}). Now using Lemma \ref{lem:tailb}, for any $d>0,$ we have,
	\benr
	pr\Big(\big|\sum_{t=\tau^0+1}^{\tau}\vep_{tj}\big|>d\Big)\le 2\exp\Big(-\frac{d^2}{2(\tau-\tau^0)\si^2}\Big).\nn
	\eenr
	Choosing $d=\si\{2c_u(\tau-\tau^0)\log (p\vee T)\}^{1/2},$ yields,
	\benr
	\big|\sum_{t=\tau^0+1}^{\tau}\vep_{tj}\big|&\le& \si\{2c_u(\tau-\tau^0)\log (p\vee T)\}^{1/2}\nn\\
	&\le& \si \surd(2c_u)\{Tu_T\log(p\vee T)\}^{1/2},\nn
	\eenr
	with probability at least $1-2\exp\{-c_u\log(p\vee T)\}.$ Now applying a union bound over $j=1,...,p,$ and $\tau=1,...,T$ yields the statement for subgaussian part of this lemma.
	
	Next, we consider the subexponential case. Apply the Bernstein's inequality (Theorem \ref{lem:bernstein}) for any $d>0$ to obtain,
	\benr\label{eq:2a}
	pr\Big(\big|\sum_{t=\tau^0+1}^{\tau}\vep_{tj}\big|>d(\tau-\tau^0)\Big)\le 2\exp\Big\{-\frac{(\tau-\tau^0)}{2}\Big(\frac{d^2}{\si^2}\wedge\frac{d}{\si}\Big)\Big\}.
	\eenr
	Choose $d=\si\{2c_u\log (p\vee T)\big/(\tau-\tau^0)\}^{1/2},$ and due to the assumption $Tv_T\ge \log (p\vee T),$ we have,
	\benr
	(\tau-\tau^0)\frac{d^2}{2\si^2}&=&c_u\log(p\vee T),\quad {\rm and},\nn\\
	(\tau-\tau^0)\frac{d}{2\si}&\ge&\surd(c_u/2)(Tv_T)^{1/2}\{\log(p\vee T)\}^{1/2}\ge \surd(c_u/2)\log(p\vee T).\nn
	\eenr
	Thus, substituting this choice of $d$ in (\ref{eq:2a}) and recalling that by choice $c_u>1$, we obtain,
	\benr
	\big|\sum_{t=\tau^0+1}^{\tau}\vep_{tj}\big|\le \surd(2c_u)\si\{Tu_T\log(p\vee T)\}^{1/2},\nn
	\eenr
	with probability at least $1-\exp\big\{-\surd(c_u/2)\log(p\vee T)\big\}.$ The statement of the subexponential part of this result now follows by applying a union bound over $j=1,...,p,$ and $\tau=1,...,T.$
\end{proof}
%%%%%%%%%%%%%%%%%%%%%%%%%%%%%%%%%%%%%%%%%%%%%%%%%%%%%%%%%%%%%%%%%%%%%%%%%%%%%%%%%%%%%%%%%%%%%%%%%%%%%%%%%%%%%%%%%%%%%%%%%%%%%%%%%%%%%%%%%%%%%%%%%%%%%%%%%%%%%%%%%%%%%%%%%%%%

\bc$\rule{4in}{0.1mm}$\ec

%%%%%%%%%%%%%%%%%%%%%%%%%%%%%%%%%%%%%%%%%%%%%%%%%%%%%%%%%%%%%%%%%%%%%%%%%%%%%%%%%%%%%%%%%%%%%%%%%%%%%%%%%%%%%%%%%%%%%%%%%%%%%%%%%%%%%%%%%%%%%%%%%%%%%%%%%%%%%%%%%%%%%%%%%%%%
\begin{lem}\label{lem:nearoptimalcross.subE.special} Suppose Condition A(II) (subexponential setting) and B holds and let $0\le v_T\le u_T\le 1,$ be any non-negative sequences. Then for any $T\ge 2,$ and any $c_u\ge 1,$ we have,
	\benr
	\sup_{\substack{\tau\in\cG(u_T,v_T);\\\tau\ge\tau^0}}\Big\|\sum_{t=\tau^0+1}^{\tau}\vep_{t}\Big\|_{\iny}\le 2c_u\si\log (p\vee T)\surd \big(Tu_T\big)\nn
	\eenr
	with probability at least $1-2\exp\{-(c_u-2)\log(p\vee T)\}.$ 	
\end{lem}

\begin{proof}[Proof of Lemma \ref{lem:nearoptimalcross.subE.special}]
	Without loss of generality assume $v_T\ge (1/T)$ (see, first paragraph in proof of Lemma \ref{lem:nearoptimalcross}). Consider any $j\in\{1,2,...,p\}$ and any $\tau>\tau^0,$ and apply the Bernstein's inequality (Theorem \ref{lem:bernstein}) for any $d>0$ to obtain,
	\benr\label{eq:2b}
	pr\Big(\big|\sum_{t=\tau^0+1}^{\tau}\vep_{tj}\big|>d(\tau-\tau^0)\Big)\le 2\exp\Big\{-\frac{(\tau-\tau^0)}{2}\Big(\frac{d^2}{\si^2}\wedge\frac{d}{\si}\Big)\Big\}.
	\eenr
	Choose $d=2c_u\si\{\log^2 (p\vee T)/(\tau-\tau^0)\}^{1/2},$ and note that,
	\benr\label{eq:1}
	(\tau-\tau^0)\frac{d^2}{2\si^2}&=&2c_u^2\log^2(p\vee T),\quad {\rm and},\nn\\
	(\tau-\tau^0)\frac{d}{2\si}&\ge& c_u\log(p\vee T),
	\eenr
	where we have used $(\tau-\tau^0)\ge Tv_T\ge 1,$ to obtain the first inequality. Thus, substituting this choice of $d$ in (\ref{eq:2b}) and recalling that by choice $c_u\ge 1$, we obtain,
	\benr
	\big|\sum_{t=\tau^0+1}^{\tau}\vep_{tj}\big|\le 2c_u\si(\tau-\tau^0)^{1/2}\{\log^2(p\vee T)\}^{1/2}\le 2c_u\si\{Tu_T\log^2 (p\vee T)\}^{1/2},\nn
	\eenr
	with probability at least $1-2\exp\{-c_u\log (p\vee T)\}.$ The statement of this lemma follows by applying a union bound over $j=1,...,p,$ and $\tau=1,...,T.$
\end{proof}
%%%%%%%%%%%%%%%%%%%%%%%%%%%%%%%%%%%%%%%%%%%%%%%%%%%%%%%%%%%%%%%%%%%%%%%%%%%%%%%%%%%%%%%%%%%%%%%%%%%%%%%%%%%%%%%%%%%%%%%%%%%%%%%%%%%%%%%%%%%%%%%%%%%%%%%%%%%%%%%%%%%%%%%%%%%%

\bc$\rule{4in}{0.1mm}$\ec

%%%%%%%%%%%%%%%%%%%%%%%%%%%%%%%%%%%%%%%%%%%%%%%%%%%%%%%%%%%%%%%%%%%%%%%%%%%%%%%%%%%%%%%%%%%%%%%%%%%%%%%%%%%%%%%%%%%%%%%%%%%%%%%%%%%%%%%%%%%%%%%%%%%%%%%%%%%%%%%%%%%%%%%%%%%%
\begin{lem}\label{lem:optimalcross} Suppose Condition A and B hold and let $u_T,v_T$ be any non-negative sequences satisfying $0\le v_T\le u_T\le 1.$ Then for any $0<a<1,$ choosing $c_a\ge \surd(1/a),$ we have,
	\benr
	\sup_{\substack{\tau\in\cG(u_T,v_T);\\\tau\ge\tau^0}}\Big|\sum_{t=\tau^0+1}^{\tau}\vep_{t}^T\eta^0\Big|\le c_{a}\phi\|\eta^0\|_2\surd (T u_T),\nn
	\eenr
	with probability at least $1-a.$	
\end{lem}
%%%%%%%%%%%%%%%%%%%%%%%%%%%%%%%%%%%%%%%%%%%%%%%%%%%%%%%%%%%%%%%%%%%%%%%%%%%%%%%%%%%%%%%%%%%%%%%%%%%%%%%%%%%%%%%%%%%%%%%%%%%%%%%%%%%%%%%%%%%%%%%%%%%%%%%%%%%%%%%%%%%%%%%%%%%%

%%%%%%%%%%%%%%%%%%%%%%%%%%%%%%%%%%%%%%%%%%%%%%%%%%%%%%%%%%%%%%%%%%%%%%%%%%%%%%%%%%%%%%%%%%%%%%%%%%%%%%%%%%%%%%%%%%%%%%%%%%%%%%%%%%%%%%%%%%%%%%%%%%%%%%%%%%%%%%%%%%%%%%%%%%%%
\begin{proof}[Proof of Lemma \ref{lem:optimalcross}]
	This result is a direct application of the Kolmogorov's inequality (Theorem \ref{thm:kolmogorov}). First consider,
	\benr
	{\rm var}(\vep_t^T\eta^0)=\eta^{0T}\Si\eta^0\le \phi^2\|\eta^0\|_2^2\nn
	\eenr
	where the inequality follows from Condition B. Note that there are at most $Tu_T$ values of $\tau$ in the set $\cG(u_T,v_T),$ and now apply the Kolmogorov's inequality (Theorem \ref{thm:kolmogorov}) for any $d>0$ to obtain,
	\benr
	pr\Big(\sup_{\substack{\tau\in\cG(u_T,v_T);\\\tau\ge\tau^0}}\Big|\sum_{t=\tau^0+1}^{\tau}\vep_{t}^T\eta^0\Big|>d\Big)\le \frac{Tu_T}{d^2}\phi\|\eta^0\|_2^2.\nn
	\eenr
	Choosing $d=c_{a}\phi\|\eta^0\|_2\surd (Tu_T),$ with $c_{a}\ge \surd (1/a)$ yields the statement of the lemma.
\end{proof}
%%%%%%%%%%%%%%%%%%%%%%%%%%%%%%%%%%%%%%%%%%%%%%%%%%%%%%%%%%%%%%%%%%%%%%%%%%%%%%%%%%%%%%%%%%%%%%%%%%%%%%%%%%%%%%%%%%%%%%%%%%%%%%%%%%%%%%%%%%%%%%%%%%%%%%%%%%%%%%%%%%%%%%%%%%%%

\bc$\rule{4in}{0.1mm}$\ec

\subsection{Deviation bounds used in the proofs of Section \ref{sec:algorithm}}

%%%%%%%%%%%%%%%%%%%%%%%%%%%%%%%%%%%%%%%%%%%%%%%%%%%%%%%%%%%%%%%%%%%%%%%%%%%%%%%%%%%%%%%%%%%%%%%%%%%%%%%%%%%%%%%%%%%%%%%%%%%%%%%%%%%%%%%%%%%%%%%%%%%%%%%%%%%%%%%%%%%%%%%%%%%%
\begin{lem}\label{lem:1.to.tau.bound} Assume Condition A(I) (subgaussian) and B holds. Then, for any $c_u,$ $c_{u1}> 0,$ we have the following bound.
	\benr\label{eq:18}
	\sup_{\substack{\tau\in\{1,.....,T\};\\\tau\ge c_uTl_T}}\frac{1}{\tau}\Big\|\sum_{t=1}^{\tau}\vep_{t}\Big\|_{\iny}\le  \si\Big\{\frac{2c_{u1}\log(p\vee T)}{c_uTl_T}\Big\}^{\frac{1}{2}}
	\eenr
	with probability at least $1-2\exp\{-(c_{u1}-2)\log(p\vee T)\}.$ Alternatively, suppose Condition A(II) (subexponential) and B hold. Additionally assume that $c_uTl_T\ge\log (p\vee T).$ Then, the same bound (\ref{eq:18}) holds, with probability at least $1-\exp\big\{-(c_{u2}-2)\log (p\vee T)\big\},$ where $c_{u2}=c_{u1}\wedge \surd(c_uc_{u1}/2).$
\end{lem}
%%%%%%%%%%%%%%%%%%%%%%%%%%%%%%%%%%%%%%%%%%%%%%%%%%%%%%%%%%%%%%%%%%%%%%%%%%%%%%%%%%%%%%%%%%%%%%%%%%%%%%%%%%%%%%%%%%%%%%%%%%%%%%%%%%%%%%%%%%%%%%%%%%%%%%%%%%%%%%%%%%%%%%%%%%%%

%%%%%%%%%%%%%%%%%%%%%%%%%%%%%%%%%%%%%%%%%%%%%%%%%%%%%%%%%%%%%%%%%%%%%%%%%%%%%%%%%%%%%%%%%%%%%%%%%%%%%%%%%%%%%%%%%%%%%%%%%%%%%%%%%%%%%%%%%%%%%%%%%%%%%%%%%%%%%%%%%%%%%%%%%%%%
\begin{proof}[Proof of Lemma \ref{lem:1.to.tau.bound}]
	First consider the subgaussian case. For any $\tau\in\{1,...,T\},$ and any $j\in\{1,...,p\}$ we have $\sum_{1}^{\tau}\vep_{tj}\sim{\rm subG}(\surd{\tau}\si).$ Consequently, for any $d>0,$ we have,
	\benr
	pr\Big(\Big|\sum_{t=1}^{\tau}\vep_{tj}\Big|>d\Big)\le 2\exp\Big(-\frac{d^2}{2\tau\si^2}\Big)\nn
	\eenr	
	Choose $d=\si\{2c_{u1}\tau\log(p\vee T)\}^{1/2},$ yields,
	\benr
	\frac{1}{\tau}\Big|\sum_{t=1}^{\tau}\vep_{tj}\Big|\le \si\{\frac{2c_{u1}\log(p\vee T)}{\tau}\}^{1/2}\le\si\{\frac{2c_{u1}\log(p\vee T)}{c_uTl_T}\}^{1/2},\nn
	\eenr 	
	with probability at least $1-2\exp\{-c_{u1}\log(p\vee T)\}.$ Here the final inequality follows since by assumption $\tau\ge c_{u}Tl_T.$ Applying a union bound over all possible values of $\tau$ and $j$ yields the statement of the lemma for this case.
	
	Next, consider the subexponential case. For any $\tau\in\{1,...,T\},$ and any $j\in\{1,...,p\},$ applying the Bernstein's inequality (Lemma \ref{lem:bernstein}) for any $d>0,$ we obtain,
	\benr\label{eq:19}
	pr\Big(\Big|\sum_{t=1}^{\tau}\vep_{tj}\Big|>d\tau\Big)\le 2\exp\Big\{-\frac{\tau}{2}\Big(\frac{d^2}{\si^2}\wedge\frac{d}{\si}\Big)\Big\}.
	\eenr	
	Choose $d=\si\{2c_{u1}\log (p\vee T)\big/\tau\}^{1/2},$ and due to the assumption $\tau\ge c_uTl_T\ge \log (p\vee T),$ we have,
	\benr
	\tau\frac{d^2}{2\si^2}&=&c_{u1}\log(p\vee T),\quad {\rm and},\nn\\
	\tau\frac{d}{2\si}&\ge&\surd(c_{u1}/2)(c_uTl_T)^{1/2}\{\log(p\vee T)\}^{1/2}\ge \surd(c_uc_{u1}/2)\log(p\vee T).\nn
	\eenr	
	Now substituting this choice of $d$ in (\ref{eq:19}),  we obtain,	
	\benr
	\frac{1}{\tau}\Big|\sum_{t=1}^{\tau}\vep_{tj}\Big|\le \si\{2c_{u1}\log (p\vee T)\big/\tau\}^{1/2}\le\si\Big\{\frac{2c_{u1}\log(p\vee T)}{c_uTl_T}\Big\}^{1/2}\nn
	\eenr
	with probability at least $1-2\exp\{-c_{u2}\log (p\vee T)\},$ where $c_{u2}=c_{u1}\wedge\surd(c_uc_{u1}/2).$ The statement of the lemma now follows by applying a union bound over all values of $\tau$ and $j.$
\end{proof}
%%%%%%%%%%%%%%%%%%%%%%%%%%%%%%%%%%%%%%%%%%%%%%%%%%%%%%%%%%%%%%%%%%%%%%%%%%%%%%%%%%%%%%%%%%%%%%%%%%%%%%%%%%%%%%%%%%%%%%%%%%%%%%%%%%%%%%%%%%%%%%%%%%%%%%%%%%%%%%%%%%%%%%%%%%%%
\bc$\rule{4in}{0.1mm}$\ec

\section{Definitions and auxiliary results}\label{app:auxiliary}

In the following Definition's \ref{def:subg}-\ref{def:submult} and Lemma's \ref{lem:tailb}-\ref{lem:bernstein} we provide basic properties of subgaussian and subexponential distributions. These are largely reproduced from \cite{vershynin2019high} and \cite{rigollet201518}. Theorem \ref{thm:kolmogorov} and \ref{thm:argmax} below reproduce the Kolmogorov's inequality and the argmax theorem.

%%%%%%%%%%%%%%%%%%%%%%%%%%%%%%%%%%%%%%%%%%%%%%%%%%%%%%%%%%%%%%%%%%%%%%%%%%%%%%%%%%%%%%%%%%%%%%%%%%%%%%%%%%%%%%%%%%%%%%%%%%%%%%%%%%%%%%%%%%%%%%%%%%%%%%%%%%%%%%%%%%%%%%%%%%%%
\begin{Def}\label{def:subg}[Subgaussian r.v.] A random variable $X\in\R$ is said to be sub-gaussian with parameter $\si>0$ \big(denoted by $X\sim{\rm subG(\si)}$\big) if $E(X)=0$ and its moment generating function
	\benr
	E(\e^{tX})\le \e^{t^2\si^2/2},\qquad \forall\,\, t\in\R\nn
	\eenr
\end{Def}
%%%%%%%%%%%%%%%%%%%%%%%%%%%%%%%%%%%%%%%%%%%%%%%%%%%%%%%%%%%%%%%%%%%%%%%%%%%%%%%%%%%%%%%%%%%%%%%%%%%%%%%%%%%%%%%%%%%%%%%%%%%%%%%%%%%%%%%%%%%%%%%%%%%%%%%%%%%%%%%%%%%%%%%%%%%%

%%%%%%%%%%%%%%%%%%%%%%%%%%%%%%%%%%%%%%%%%%%%%%%%%%%%%%%%%%%%%%%%%%%%%%%%%%%%%%%%%%%%%%%%%%%%%%%%%%%%%%%%%%%%%%%%%%%%%%%%%%%%%%%%%%%%%%%%%%%%%%%%%%%%%%%%%%%%%%%%%%%%%%%%%%%%
\begin{Def}\label{def:sube}[Subexponential r.v.] A random variable $X\in\R$ is said to be sub-exponential with parameter $\si>0$ \big(denoted by $X\sim{\rm subE(\si)}$\big) if $E(X)=0$ and its moment generating function
	\benr
	E(\e^{tX})\le \e^{t^2\si^2/2},\qquad \forall\,\, |t|\le \frac{1}{\si}\nn
	\eenr
\end{Def}
%%%%%%%%%%%%%%%%%%%%%%%%%%%%%%%%%%%%%%%%%%%%%%%%%%%%%%%%%%%%%%%%%%%%%%%%%%%%%%%%%%%%%%%%%%%%%%%%%%%%%%%%%%%%%%%%%%%%%%%%%%%%%%%%%%%%%%%%%%%%%%%%%%%%%%%%%%%%%%%%%%%%%%%%%%%%

%%%%%%%%%%%%%%%%%%%%%%%%%%%%%%%%%%%%%%%%%%%%%%%%%%%%%%%%%%%%%%%%%%%%%%%%%%%%%%%%%%%%%%%%%%%%%%%%%%%%%%%%%%%%%%%%%%%%%%%%%%%%%%%%%%%%%%%%%%%%%%%%%%%%%%%%%%%%%%%%%%%%%%%%%%%%
\begin{Def}\label{def:submult} A random vector $X\in\R^p$ shall said to be subgaussian or subexponential with parameter $\si,$ if the inner products $\langle X, v\rangle\sim {\rm subG}(\si)$ or $\langle X, v\rangle\sim {\rm subE}(\si),$ respectively, for any $v\in\R^p$ with $\|v\|_2 = 1.$
\end{Def}
%%%%%%%%%%%%%%%%%%%%%%%%%%%%%%%%%%%%%%%%%%%%%%%%%%%%%%%%%%%%%%%%%%%%%%%%%%%%%%%%%%%%%%%%%%%%%%%%%%%%%%%%%%%%%%%%%%%%%%%%%%%%%%%%%%%%%%%%%%%%%%%%%%%%%%%%%%%%%%%%%%%%%%%%%%%%

%%%%%%%%%%%%%%%%%%%%%%%%%%%%%%%%%%%%%%%%%%%%%%%%%%%%%%%%%%%%%%%%%%%%%%%%%%%%%%%%%%%%%%%%%%%%%%%%%%%%%%%%%%%%%%%%%%%%%%%%%%%%%%%%%%%%%%%%%%%%%%%%%%%%%%%%%%%%%%%%%%%%%%%%%%%%
\begin{Def}\label{def:utight} A sequence of random variables $X_n$ is said to be uniformly tight if for every $\ep>0,$ there is a compact set $K$ such that $pr(X_n\in K)>1-\ep.$
\end{Def}
%%%%%%%%%%%%%%%%%%%%%%%%%%%%%%%%%%%%%%%%%%%%%%%%%%%%%%%%%%%%%%%%%%%%%%%%%%%%%%%%%%%%%%%%%%%%%%%%%%%%%%%%%%%%%%%%%%%%%%%%%%%%%%%%%%%%%%%%%%%%%%%%%%%%%%%%%%%%%%%%%%%%%%%%%%%%

%%%%%%%%%%%%%%%%%%%%%%%%%%%%%%%%%%%%%%%%%%%%%%%%%%%%%%%%%%%%%%%%%%%%%%%%%%%%%%%%%%%%%%%%%%%%%%%%%%%%%%%%%%%%%%%%%%%%%%%%%%%%%%%%%%%%%%%%%%%%%%%%%%%%%%%%%%%%%%%%%%%%%%%%%%%%
\begin{lem}\label{lem:tailb}[Tail bounds] (i) If $X\sim {\rm subG}(\si),$ then,
	\benr
	pr(|X|\ge \la)\le 2\exp(-\la^2/2\si^2).\nn
	\eenr
	(ii) If $X\sim {\rm subE}(\si),$ then
	\benr
	pr(|X|\ge \la)\le 2\exp\Big\{-\frac{1}{2}\Big(\frac{\la^2}{\si^2}\wedge\frac{\la}{\si}\Big) \Big\}.\nn
	\eenr
\end{lem}
%%%%%%%%%%%%%%%%%%%%%%%%%%%%%%%%%%%%%%%%%%%%%%%%%%%%%%%%%%%%%%%%%%%%%%%%%%%%%%%%%%%%%%%%%%%%%%%%%%%%%%%%%%%%%%%%%%%%%%%%%%%%%%%%%%%%%%%%%%%%%%%%%%%%%%%%%%%%%%%%%%%%%%%%%%%%

%%%%%%%%%%%%%%%%%%%%%%%%%%%%%%%%%%%%%%%%%%%%%%%%%%%%%%%%%%%%%%%%%%%%%%%%%%%%%%%%%%%%%%%%%%%%%%%%%%%%%%%%%%%%%%%%%%%%%%%%%%%%%%%%%%%%%%%%%%%%%%%%%%%%%%%%%%%%%%%%%%%%%%%%%%%%
\begin{proof}[Proof of Lemma \ref{lem:tailb} ]This proof is a simple application of the Markov inequality. For any $t>0,$
	\benr
	pr(X\ge \la)=pr(tX\ge t\la)\le \frac{E\e^{tX}}{\e^{t\la}}=\e^{-t\la+t^2\si^2/2}.\nn
	\eenr
	Minimizing over $t>0,$ yields the choice $t^*=\la/\si^2,$ and substituting in the above bound ,we obtain,
	\benr
	pr(X\ge \la)\le \inf_{t>0}\e^{-t\la+t^2\si^2/2}=e^{-\la^2/2\si^2}.\nn
	\eenr
	Repeating the same for $P(X\le -\la)$ yields part (i) of the lemma. To prove Part (ii),   repeat the above argument with $t\in(0,1/\si],$ to obtain,
	\benr
	pr(X\ge \la)=pr(tX\ge t\la)\le\e^{-t\la+t^2\si^2/2}.
	\eenr
	As in the subgaussian case, to obtain the tightest bound one needs to find $t^*$ that minimizes $-t\la+t^2\si^2/2,$ with the additional constraint for this subexponential case that $t\in(0,1/\si].$ We know that the unconstrained minimum occurs at $t^*=\la/\si^2>0.$ Now consider two cases:
	\begin{enumerate}[itemindent=0mm]
		\item If $t^*<(0,1/\si] \Leftrightarrow\la\le \si$ then the unconstrained minimum is same as the constrained minimum, and substituting this value yields the same tail behavior as the subgaussian case.
		\item If $t^*>(1/\si)\Leftrightarrow\la>\si,$ then note that $-t\la+t^2\si^2/2$ is decreasing in $t,$ in the interval $(0,(1/\si)],$ thus the minimum occurs at the boundary $t=1/\si.$ Substituting in the tail bound we obtain for this case,
		\benr
		pr(X\ge \la)\le\e^{-t\la+t^2\si^2/2}= \exp\{-(\la/\si)+(1/2)\}\le\exp{(-\la/2\si)},\nn
		\eenr
		where the final inequality follows since $\la>\si.$
	\end{enumerate}
	Part (ii) of the lemma is obtained by combining the results of the above two cases.
\end{proof}
%%%%%%%%%%%%%%%%%%%%%%%%%%%%%%%%%%%%%%%%%%%%%%%%%%%%%%%%%%%%%%%%%%%%%%%%%%%%%%%%%%%%%%%%%%%%%%%%%%%%%%%%%%%%%%%%%%%%%%%%%%%%%%%%%%%%%%%%%%%%%%%%%%%%%%%%%%%%%%%%%%%%%%%%%%%%

%%%%%%%%%%%%%%%%%%%%%%%%%%%%%%%%%%%%%%%%%%%%%%%%%%%%%%%%%%%%%%%%%%%%%%%%%%%%%%%%%%%%%%%%%%%%%%%%%%%%%%%%%%%%%%%%%%%%%%%%%%%%%%%%%%%%%%%%%%%%%%%%%%%%%%%%%%%%%%%%%%%%%%%%%%%
\begin{lem} \label{lem:lcsubG} Assume that $X\sim{\rm subG}(\si),$ and that $\al\in\R,$ then, (i) $\alpha X\sim{\rm subG}(|\alpha|\si).$ If $X_1\sim{\rm subG}(\si_1)$ and $X_2\sim{\rm subG}(\si_2),$  then, (ii) $X_1+X_2\sim {\rm subG}(\si_1+\si_2).$ If $X_1\sim{\rm subG}(\si)$ and $X_2\sim{\rm subG}(\si)$ are independent, then, (iii) $X_1+X_2\sim{\rm subG}(\si\surd 2).$
\end{lem}
%%%%%%%%%%%%%%%%%%%%%%%%%%%%%%%%%%%%%%%%%%%%%%%%%%%%%%%%%%%%%%%%%%%%%%%%%%%%%%%%%%%%%%%%%%%%%%%%%%%%%%%%%%%%%%%%%%%%%%%%%%%%%%%%%%%%%%%%%%%%%%%%%%%%%%%%%%%%%%%%%%%%%%%%%%%

%%%%%%%%%%%%%%%%%%%%%%%%%%%%%%%%%%%%%%%%%%%%%%%%%%%%%%%%%%%%%%%%%%%%%%%%%%%%%%%%%%%%%%%%%%%%%%%%%%%%%%%%%%%%%%%%%%%%%%%%%%%%%%%%%%%%%%%%%%%%%%%%%%%%%%%%%%%%%%%%%%%%%%%%%%%
\begin{proof}[of Lemma \ref{lem:lcsubG}] The first part follows directly from the inequality $E(\e^{t\al X})\le \exp(t^2\al^2\si^2/2).$ To prove Part (ii) use the H\"older's inequality to obtain,
	\benr
	E(\e^{t(X_1+X_2)})&=&E(\e^{tX_1}\e^{tX_2})\le \{E(\e^{tX_1p})\}^{\frac{1}{p}}\{E(\e^{tX_2q})\}^{\frac{1}{q}}\nn\\
	&\le& \e^{\frac{t^2}{2}\si_1^2p^2}\e^{\frac{t^2}{2}\si_2^2q^2}=\e^{\frac{t^2}{2}(p\si_1^2+q\si_2^2)}\nn
	\eenr
	where $p,q\in[1,\iny],$ with $1/p+1/q=1.$ Choose $p^{*}=(\si_2/\si_1)+1,$ $q^*=(\si_1/\si_2)+1$ to obtain $E(\e^{t(X_1+X_2)})\le \exp\big\{\frac{t^2}{2}(\si_1+\si_2)^2\big\}.$ For Part (iii) note that,
	\benr
	E(\e^{t(X_1+X_2)})&=&E(\e^{tX_1}\e^{tX_2})=E(\e^{tX_1})E(\e^{tX_2})\nn\\
	&\le& \e^{\frac{t^2\si^2}{2}}\e^{\frac{t^2\si^2}{2}}=\e^{\frac{t^2(\si\surd 2)^2}{2}}\nn
	\eenr
	This completes the proof of this lemma.	
\end{proof}
%%%%%%%%%%%%%%%%%%%%%%%%%%%%%%%%%%%%%%%%%%%%%%%%%%%%%%%%%%%%%%%%%%%%%%%%%%%%%%%%%%%%%%%%%%%%%%%%%%%%%%%%%%%%%%%%%%%%%%%%%%%%%%%%%%%%%%%%%%%%%%%%%%%%%%%%%%%%%%%%%%%%%%%%%%%

%%%%%%%%%%%%%%%%%%%%%%%%%%%%%%%%%%%%%%%%%%%%%%%%%%%%%%%%%%%%%%%%%%%%%%%%%%%%%%%%%%%%%%%%%%%%%%%%%%%%%%%%%%%%%%%%%%%%%%%%%%%%%%%%%%%%%%%%%%%%%%%%%%%%%%%%%%%%%%%%%%%%%%%%%%
\begin{lem}\label{lem:lcsubE} Assume that $X\sim{\rm subE(\la)},$ and that $\al\in\R,$ then, (i) $\alpha X\sim{\rm subE}(|\alpha|\la).$ If $X_1\sim{\rm subE(\la_1)}$ and $X_2\sim{\rm subE(\la_2)},$ then, (ii) $X_1+X_2\sim{\rm subE(\la_1+\la_2)}.$ If $X_1\sim{\rm subE}(\la)$ and $X_2\sim{\rm subE}(\la)$ are independent, then, (iii) $X_1+X_2\sim{\rm subE}(\la\surd 2).$
\end{lem}
%%%%%%%%%%%%%%%%%%%%%%%%%%%%%%%%%%%%%%%%%%%%%%%%%%%%%%%%%%%%%%%%%%%%%%%%%%%%%%%%%%%%%%%%%%%%%%%%%%%%%%%%%%%%%%%%%%%%%%%%%%%%%%%%%%%%%%%%%%%%%%%%%%%%%%%%%%%%%%%%%%%%%%%%%%

The proof of Lemma \ref{lem:lcsubE} is analogous to that of Lemma \ref{lem:lcsubG} and is thus omitted.
The next result is the Bernstein's inequality, reproduced from Lemma 1.13 of \cite{rigollet201518}. This result is a direct consequence of Lemma \ref{lem:tailb} and Lemma \ref{lem:lcsubE}.

%%%%%%%%%%%%%%%%%%%%%%%%%%%%%%%%%%%%%%%%%%%%%%%%%%%%%%%%%%%%%%%%%%%%%%%%%%%%%%%%%%%%%%%%%%%%%%%%%%%%%%%%%%%%%%%%%%%%%%%%%%%%%%%%%%%%%%%%%%%%%%%%%%%%%%%%%%%%%%%%%%%%%%%%%%
\begin{lem}[Bernstein's inequality]\label{lem:bernstein} Let $X_1,X_2,...,X_T$ be independent random variables such that $X_t\sim {\rm subE}(\la).$ Then for any $d>0$ we have,
	\benr
	pr(|\bar X|>d)\le 2\exp\Big\{-\frac{T}{2}\Big(\frac{d^2}{\la^2}\wedge \frac{d}{\la}\Big)\Big\}\nn
	\eenr
\end{lem}
%%%%%%%%%%%%%%%%%%%%%%%%%%%%%%%%%%%%%%%%%%%%%%%%%%%%%%%%%%%%%%%%%%%%%%%%%%%%%%%%%%%%%%%%%%%%%%%%%%%%%%%%%%%%%%%%%%%%%%%%%%%%%%%%%%%%%%%%%%%%%%%%%%%%%%%%%%%%%%%%%%%%%%%%%%

%%%%%%%%%%%%%%%%%%%%%%%%%%%%%%%%%%%%%%%%%%%%%%%%%%%%%%%%%%%%%%%%%%%%%%%%%%%%%%%%%%%%%%%%%%%%%%%%%%%%%%%%%%%%%%%%%%%%%%%%%%%%%%%%%%%%%%%%%%%%%%%%%%%%%%%%%%%%%%%%%%%%%%%%%%
The next result is the Kolmogorov's inequality reproduced from \cite{hajek1955generalization}
\begin{thm}[Kolmogorov's inequality]\label{thm:kolmogorov} If $\xi_1,\xi_2,...$ is a sequence of mutually independent random variables with mean values $E(\xi_k)=0$ and finite variance ${\rm var}(\xi_k)=D_k^2$ $(k=1,2,...),$ we have, for any $\vep>0,$
	\benr
	pr\Big(\max_{1\le k\le m}\big|\xi_1+\xi_2+...+\xi_k\big|>\vep\Big)\le \frac{1}{\vep^2}\sum_{k=1}^mD_k^2\nn
	\eenr	
\end{thm}
%%%%%%%%%%%%%%%%%%%%%%%%%%%%%%%%%%%%%%%%%%%%%%%%%%%%%%%%%%%%%%%%%%%%%%%%%%%%%%%%%%%%%%%%%%%%%%%%%%%%%%%%%%%%%%%%%%%%%%%%%%%%%%%%%%%%%%%%%%%%%%%%%%%%%%%%%%%%%%%%%%%%%%%%%%

Following is the well known `Argmax' theorem reproduced from Theorem 3.2.2 of \cite{vaart1996weak} and an elementary definition of uniform tightness of a sequence of random variables reproduced from Page 166, Chapter 2 of \cite{durrett2010probability}.
%%%%%%%%%%%%%%%%%%%%%%%%%%%%%%%%%%%%%%%%%%%%%%%%%%%%%%%%%%%%%%%%%%%%%%%%%%%%%%%%%%%%%%%%%%%%%%%%%%%%%%%%%%%%%%%%%%%%%%%%%%%%%%%%%%%%%%%%%%%%%%%%%%%%%%%%%%%%%%%%%%%%%%%%%%%
\begin{thm}[Argmax Theorem]\label{thm:argmax} Let $\cM_n,\cM$ be stochastic processes indexed by a metric space $H$ such that $\cM_n\Rightarrow\cM$ in $\ell^{\iny}(K)$ for every compact set $K\subseteq H$. Suppose that almost all sample paths $h\to \cM(h)$ are upper semicontinuous and posses a unique maximum at a (random) point $\h h,$ which as a random map in $H$ is tight. If the sequence $\h h_n$ is uniformly tight and satisfies $\cM_n(\h h_n)\ge \sup_h \cM_n(h)-o_p(1),$ then $\h h_n\Rightarrow \h h$ in $H.$
\end{thm}

\section{Discussion on sparsity assumption}\label{app:equivalence}

The purpose of this section is to show that the sparsity assumption (\ref{def:setS}) on the mean vectors $\theta_1^0$ and $\theta_2^0$ holds without loss of generality with respect to assuming sparsity of jump vector $\eta^0=\theta_1^0-\theta_2^0,$ and in context of the estimator $\tilde\tau$ and the estimation and inference results presented in this manuscript.

Recall the model (\ref{model:rvmcp}) under consideration,
\benr
x_t=\begin{cases}\theta_1^0+\vep_t & t=1,...,\tau^0\\
	\theta_2^0+\vep_t & t=\tau^0+1,...,T.\end{cases}\nn
\eenr
where $\vep_t$ are i.i.d as per Condition A. Define,
\benr
&&\bar x=\frac{1}{T}\sum_{t=1}^T x_t= \frac{1}{T}\big[\tau^0\theta_1^0+(T-\tau^0)\theta_2^0\big]+\frac{1}{T}\sum_{t=1}^T\vep_t,\quad{\rm and}\nn\\
&&\theta_1^*=\frac{(T-\tau^0)\eta^0}{T},\quad \theta_2^*=-\frac{\tau^0\eta^0}{T},\quad\vep_t^*=\vep_t-\bar\vep,\,\,\,{\rm and}\,\,\,\bar\vep=\frac{1}{T}\sum_{t=1}^T\vep_t,\nn
\eenr
Then performing a mean centering operation $x_t^*=x_t-\bar x,$ $t=1,...,T,$ yields the transformed model,
\benr\label{mod:center.transform}
x_t^*=\begin{cases}\theta_1^{*}+\vep_t^* & t=1,...,\tau^0\\
	\theta_2^{*}+\vep_t^* & t=\tau^0+1,...,T.\end{cases}
\eenr
Note that assuming $\eta^0=\theta_1^0-\theta_2^0$ is s-sparse directly implies that the mean vectors $\theta_1^*$ and $\theta_2^*$ of model (\ref{mod:center.transform}) are now individually s-sparse, also note that $\eta^{0*}=\theta_1^*-\theta_2^*=\eta^0.$ Thus making assumption (\ref{def:setS}) feasible.

The consequence of this centering operation is an alteration to the unobserved noise term as $\vep_t^*=(\vep_t-\bar\vep),$ $t=1,...,T.$ Although this induces a dependence amongst $\vep_t^*,$ however its representation allows separability of this structure and in turn allows all results of the manuscript to remain valid. The only consequence of this alteration being in the universal constants of the localization bounds of Section \ref{sec:mainresults} and Section \ref{sec:algorithm}. There will be no consequence in context of limiting distributions of Section \ref{sec:inference}. This is illustrated in the following discussion.

\begin{lem}\label{lem:elementary} Suppose $\vep_t,$ $t=1,...,T$ satisfy Condition A and B, then
	\benr
	\|\bar\vep\|_{\iny}\le
	\begin{cases}
		c_u\si\surd\{\log (p\vee T)/T\}, & {\rm under\,\, subG},\\
		c_u\si\surd\{\log (p\vee T)/T\}, & {\rm under\,\, subE\,\, when}\,\, T\ge \log (p\vee T)\\
		c_u\si\log (p\vee T)/\surd T, & {\rm under\,\, subE\,\, when}\,\, T\ge 1\\ 	
	\end{cases}	
	\eenr
	with probability $1-o(1).$ Moreover, for any non-random $\delta\in\R^p,$ $\|\delta\|_2=1,$ we have, $\surd{T}\delta^T\bar\vep=O_p(1).$ More precisely, for any $0<a<1$ there exists $c_a'>0$ such that $pr\big(\big|\surd{T}\delta^T\bar\vep\big|>c_a')\le a.$
\end{lem}

This is a straightforward result. The bound on the sup-norm follows directly by applying subgaussian or subexponential tail bounds component-wise, followed with a union bound over components. The second result follows directly from the Markov inequality upon noting that $\surd{T}\delta^T\bar\vep$ is also subgaussian or subexponential when $\vep_t$ are subgaussian or subexponential, respectively (Definition \ref{def:submult}), with ${\rm var}(\surd{T}\delta^T\bar\vep)\le \phi^2\le c_u\si^2.$

Now recall the construction of the set $\cG(u_T,v_T)$ from (\ref{def:setG}) and note that the results of Theorem \ref{thm:nearoptimalcp}, Theorem \ref{thm:subE.nearoptimal.special} and Theorem \ref{thm:cpoptimal} of Section \ref{sec:mainresults} rely on the following uniform bounds (provided in Appendix \ref{app:sec.main}),
\benr
&&\sup_{\substack{\tau\in\cG(u_T,v_T);\\\tau\ge\tau^0}}\Big\|\sum_{t=\tau^0+1}^\tau\vep_t\Big\|_{\iny}\le\begin{cases}
	c_u\si\{Tu_T\log(p\vee T)\}^{\frac{1}{2}}, & {\rm under\,\,subG}\\
	c_u\si\{Tu_T\log(p\vee T)\}^{\frac{1}{2}}, & {\rm under\,\,subE}\,\big(Tv_T\ge \log(p\vee T)\big)\\
	c_u\si\log(p\vee T)\{Tu_T\}^{\frac{1}{2}}, & {\rm under\,\,subE}\,\big(v_T\ge 0
	\big)\end{cases}\hspace{1cm}\label{eq:46}\\[6pt]
&&\sup_{\substack{\tau\in\cG(u_T,v_T);\\\tau\ge\tau^0}}\Big|\sum_{t=\tau^0+1}^\tau\vep_t^T\eta^0\Big|\le c_a\phi \|\eta^0\|_2\surd(Tu_T)\label{eq:47},
\eenr
where (\ref{eq:46}) hold with probability $1-o(1)$ and (\ref{eq:47}) holds with probability $1-a,$ when $c_a\ge \surd(1/a).$ (see, Lemma \ref{lem:nearoptimalcross}, Lemma \ref{lem:nearoptimalcross.subE.special} and Lemma \ref{lem:optimalcross}).

The only additional requirement to replicate the proofs of Theorem \ref{thm:nearoptimalcp}, \ref{thm:subE.nearoptimal.special}, Theorem \ref{thm:cpoptimal} under the model (\ref{mod:center.transform}) are the bounds (\ref{eq:46}) and (\ref{eq:47}) with $\vep_t$ replaced with $\vep_t^*.$ This can be done using Lemma \ref{lem:elementary} as follows. Consider the subgaussian case and note that,
\benr
\sup_{\substack{\tau\in\cG(u_T,v_T);\\\tau\ge\tau^0}}\Big\|\sum_{t=\tau^0+1}^\tau\vep_t^*\Big\|_{\iny}&\le& \sup_{\substack{\tau\in\cG(u_T,v_T);\\\tau\ge\tau^0}}\Big\|\sum_{t=\tau^0+1}^\tau\vep_t\Big\|_{\iny} + Tu_T\|\bar\vep\|_{\iny}\nn\\
&\le& c_u\si\{Tu_T\log(p\vee T)\}^{\frac{1}{2}}+ c_u\si u_T\surd\{T\log(p\vee T)\}\nn\\
&\le& \si\{Tu_T\log(p\vee T)\}^{\frac{1}{2}}\big[c_u+\surd{u_T}\big]\nn\\
&\le& (c_u+1)\si\{Tu_T\log(p\vee T)\}^{\frac{1}{2}}\nn
\eenr
with probability $1-o(1).$ Here the second inequality follows from (\ref{eq:46}) and Lemma \ref{lem:elementary}. The final inequality follows since $u_T\le 1.$ Thus, the only impact of using the transformed model (\ref{mod:center.transform}) is on the associated universal constant. A similar argument also yields,
\benr
\sup_{\substack{\tau\in\cG(u_T,v_T);\\\tau\ge\tau^0}}\Big|\sum_{t=\tau^0+1}^\tau\vep_t^{*T}\eta^0\Big|&\le& \sup_{\substack{\tau\in\cG(u_T,v_T);\\\tau\ge\tau^0}}\Big|\sum_{t=\tau^0+1}^\tau\vep_t^{T}\eta^0\Big|+Tu_T\big|\bar\vep^T\eta^0\big|\nn\\
&\le& c_a\phi \|\eta^0\|_2\surd(Tu_T)+ u_Tc_a'\|\eta^0\|_2\si\surd{T}\nn\\
&\le& \si(c_a+c_a')\|\eta^0\|_2\surd{(Tu_T)}\nn
\eenr
with probability at least $1-2a.$ The subexponential cases can be handled similarly using the corresponding cases in (\ref{eq:46}) and Lemma \ref{lem:elementary}. These bounds allow reproducing identical arguments to obtain the bounds provided in Theorem \ref{thm:nearoptimalcp}, \ref{thm:subE.nearoptimal.special} and \ref{thm:cpoptimal} under the transformed model (\ref{mod:center.transform}) upto universal constants.

In context of results of Section \ref{sec:inference}, the only consequence of centering is that the stochastic term that stabilizes to form the limiting distribution will under the transformed model (\ref{mod:center.transform}) comprise of an additional $o_p(1)$ residual term. Clearly this will not affect the weak limit. For e.g. in the vanishing case of $\xi\to 0,$ the limiting distribution is governed by the term $\sum_{t=\tau^0}^{\tau^0+r\xi^{-2}}\vep_t^T\eta^0,$ where $0<r\le c_u.$ When stated w.r.t model (\ref{mod:center.transform}) this term changes to the following,
\benr
\sum_{t=\tau^0}^{\tau^0+r\xi^{-2}}\vep_t^{*T}\eta^0&=&\sum_{t=\tau^0}^{\tau^0+r\xi^{-2}}\vep_t^{T}\eta^0 +r\xi^{-1}\bar\vep^T\delta=\sum_{t=\tau^0}^{\tau^0+r\xi^{-2}}\vep_t^{T}\eta^0+r\xi^{-1}O_p\big(1/\surd{T}\big)\nn\\
&=&\sum_{t=\tau^0}^{\tau^0+r\xi^{-2}}\vep_t^{T}\eta^0+o_p(1).\nn
\eenr
Here the second equality follows from Lemma \ref{lem:elementary}. Consequently, centering does not alter the weak limit. A similar argument holds for the non-vanishing case.

Finally, the only additional result that enables the results of Section \ref{sec:algorithm} is Theorem \ref{thm:unifmean}. The proof of this theorem requires control on the stochastic term in Lemma \ref{lem:1.to.tau.bound}. Proceeding analogously as earlier, it may be observed that the same bound holds upto universal constants when $\vep_t$ is replaced with $\vep_t^*.$

In conclusion, the discussion of this section implies that assuming sparsity (\ref{def:setS}) on the mean vectors $\theta_1^0$ and $\theta_2^0$ holds without loss of generality with respect to assuming sparsity of the jump vector $\eta^0,$ in context of the change point estimator $\tilde\tau$ and the results presented in this article.

\section{Additional numerical results and further details}\label{app:numerical}

This section provides remaining results of Simulation A and Simulation B discussed in Section \ref{sec:numerical}, an additional Simulation C is also provided here that numerically examines the uniform validity of the estimation and inference results developed in the main manuscript where uniformity under consideration is that of the parametric space of the mean parameters $\theta_1^0$ and $\theta_2^0.$ Finally we also provide in this section the pertinent details regarding estimation of the jump size and asymptotic variance that was utilized to obtain confidence intervals computed in Section \ref{sec:numerical}.

\vspace{5mm}
\noi{\bf Simulation C:} (numerical evaluation of uniformity of results over mean parametric space) The design of this simulation is as follows. The unobserved noise is generated as in the Gaussian setting of Simulation A. We set model parameters $T=425,$ $p=250$ and $\tau^0=\lfloor0.4\cdotp T\rfloor.$ We set the mean parameters as $c\theta_1^0$ and $c\theta_2^0,$ where $\theta_1$ and $\theta_2^0$ are as described in Section \ref{sec:numerical} and the $c$ is a constant chosen from an equally separated grid of twenty five values, $c\in\{1,...,0.25\}.$ We evaluate bias, RMSE (over 100 replications) and coverage, average margin of error (over 500 replications).

\begin{figure}[]
	\centering
	\resizebox{\textwidth}{!}{
		\begin{minipage}[b]{0.5\textwidth}
			\includegraphics[width=\textwidth]{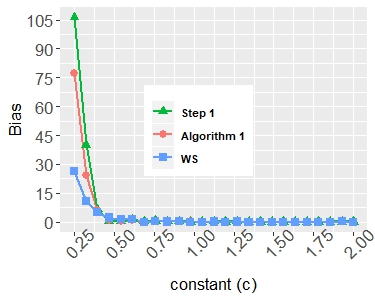}
		\end{minipage}
		\begin{minipage}[b]{0.5\textwidth}
			\includegraphics[width=\textwidth]{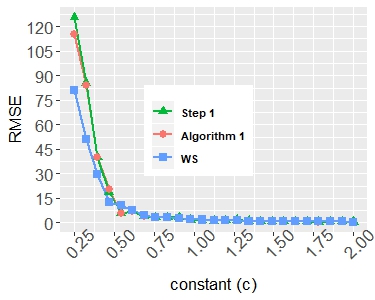}
	\end{minipage}}
	% note that files may not be rotated
	\caption{\footnotesize{Estimation results of Simulation C (100 monte carlo replications). x-axis: constant $c\in\{2,....,0.25\}.$ {\it Left panel:} y-axis: Bias ($|E(\h\tau-\tau^0)|$). {\it Right panel:} RMSE ($E^{1/2}(\h\tau-\tau^0)^2$).}}
	\label{fig:simCest}
\end{figure}

\begin{figure}[]
	\centering
	\resizebox{\textwidth}{!}{
		\begin{minipage}[b]{0.5\textwidth}
			\includegraphics[width=\textwidth]{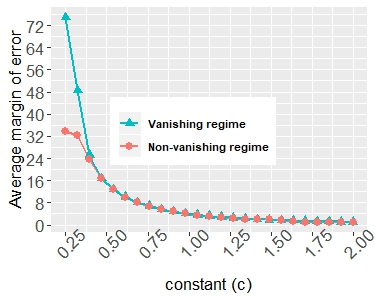}
		\end{minipage}
		\begin{minipage}[b]{0.5\textwidth}
			\includegraphics[width=\textwidth]{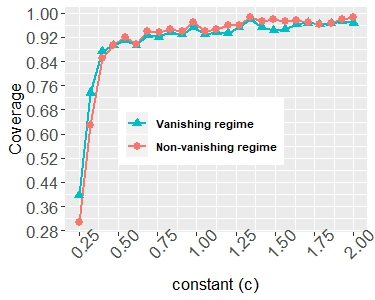}
	\end{minipage}}
	% note that files may not be rotated
	\caption{\footnotesize{Inference results of Simulation C (500 monte carlo replications). x-axis: constant $c\in\{2,....,0.25\}.$ {\it Left panel:} y-axis: Average over replications of margin of error of confidence intervals. {\it Right panel:} coverage over replications of confidence intervals.}}
	\label{fig:simCinf}
\end{figure}

The results of Simulation C presented in Figure \ref{fig:simCest} and Figure \ref{fig:simCinf} are as per expectation. The proposed inference methodology yields a coverage controlled at near the nominal level, uniformly in the sub-interval $c\in(0.5,2).$ As the constant $c$ diminishes the variance of estimators increases (as observed via rmse), this is in turn correctly captured in the asymptotic variance (as observed by the margin of error) which is then evidenced by the proper control on coverage in this sub-interval $c\in (0.5,2).$ The methodology appears to break down in the sub-interval $(0.25,0.5).$ This is not due to a lack of uniformity of the results over the mean parametric space but instead due to the jump size getting smaller and going beyond the region of detectability, i.e, at these smaller values of $c,$ the jump size is $\xi=c\|\theta_1^0-\theta_2^0\|_2$ and may no longer be able to preserve the relation (\ref{eq:16}), thus leading to a breakdown of the theoretical results supporting the methodology and thereby leading to the observations of Figure \ref{fig:simCest} and Figure \ref{fig:simCinf}.

\vspace{5mm}
\noi{\bf Computation of $\h\si^2$ and $\h\xi^2$}: Here we discuss the computation of $\h\si^2_{\iny}$ and $\h\xi$ utilized for the computation of confidence intervals for $\tau^0$ using the result of Theorem \ref{thm:wc.vanishing} and Theorem \ref{thm:wc.non.vanishing}.

In order to alleviate finite sample regularization biases from the mean parameter estimates we utilize refitted mean estimates computed as  $\breve\theta_1=\big[\bar x_{(1:\breve\tau]}\big]_{\h S1}$ and $\breve\mu_2=\big[\bar x_{(\breve\tau:T]}\big]_{\h S2},$ where $\breve\tau$ is the change point estimate of Algorithm 1 and $\h S1=\{j\,\,\h\theta_{1j}\ne 0\},$ $\h S2=\{j\,\,\h\theta_{2j}\ne 0\},$ where $\h\theta_1$ and $\h\theta_2$ are the Step 2 mean estimates of Algorithm 1. It is well known in the literature that refitted mean estimates preserve the rate of convergence of the regularized version while reducing finite sample biases, see, e.g. \cite{belloni2011square} and \cite{belloni2017pivotal}. The jump vector and jump size are then estimated as $\h\eta=\breve\theta_1-\breve\theta_2,$ and  $\h\xi=\|\h\eta\|_2,$ respectively.

Next, recall the asymptotic variance $\si_{\iny}^2$ from Condition E where it is defined as $\lim_{T\to\iny} \eta^{0T}\Si\eta^0\big/\xi^2.$ The direct way to estimate this quantity is to first estimate the high dimensional covariance $\Si$ and plug in previously estimated quantities $\h\eta$ and $\h\xi$ into the finite sample representation of $\si^2_{\iny},$ i.e., $\eta^{0T}\Si\eta^0\big/\xi^2.$ However, this direct approach shall also be expected to further induce regularization biases that will inevitably seep in due to the estimation of the high dimensional $\Si.$ We note that an explicit estimation of $\Si$ is itself not necessary for the proposed inference methodology, and instead only the jump size and asymptotic variance are necessary parameters. In view of this observation we use the following approach in order to avoid the above eventuality.

Consider a one-dimensional projection $z_t=\xi^{-1}\eta^{0T}x_t,$ $t=1,...,T$ of the model (\ref{model:rvmcp}) and note that it yields a transformed model of the form,

\benr\label{mod:projectedseries}
z_t=\xi^{-1}\eta^{0T}x_t=\begin{cases}\mu^0_1+\psi_t, &t=1,...,\tau^0\\
	\mu^0_2+\psi_t,& t=\tau^0+1,...,T,\end{cases}
\eenr
where $\mu_1^0=\xi^{-1}\eta^{0T}\theta_1^0\in\R,$ $\mu_2^0=\xi^{-1}\eta^{0T}\theta_2^0\in\R$ and more importantly $\psi_t=\xi^{-1}\eta^{0T}\vep_t,$ $t=1,...,T.$ Consequently the variance of the transformed unobserved noise term $\psi_t$ is $\eta^{0T}\Si\eta^0/\xi^{2},$ which is exactly a finite sample representation of the asymptotic variance $\si^{2}_{\iny}$ of interest. In view of this observation we estimate this quantity as the sample variance of the residual of the transformed model \ref{mod:projectedseries} implemented by utilizing the previously estimated jump size and jump vector, i.e., let,
\benr
z_t=\h\xi^{-1}\h\eta^{T}x_t,\quad \h\mu_1=\h\xi^{-1}\h\eta^{T}\breve\theta_1\quad{\rm and}\quad\h\mu_2=\h\xi^{-1}\h\eta^{T}\breve\theta_2.\nn
\eenr
Then we estimate $\si^2_{\iny}$ as the sample variance,
\benr
\h\si^2_{\iny}= \frac{1}{T}\Bigg\{\sum_{t=1}^{\breve\tau}(\h z_t-\h\mu_1)^2+\sum_{t=\breve\tau+1}^{T}(\h z_t-\h\mu_2)^2\Bigg\}.\nn
\eenr

\vspace{5mm}
\noi{\bf Additional results of Simulation A and Simulation B:}
%%%%%%%%%%%%%%%%%%%%GAUSSIAN ESTIMATION TABLES%%%%%%%%%%%%%%%%%%%%%%%%%%%%%%%%%%%%%%%%%%%%%%

\begin{table}[H]
	\caption{\footnotesize{Simulation A(i): estimation performance of Step 1 ($\h\tau$), AL1 $(\breve\tau)$ and WS methods under Gaussian setting with $\tau^0=\lfloor0.4\cdotp T\rfloor.$ Bias ($|E(\h\tau-\tau^0)|$), and RMSE ($E^{1/2}(\h\tau-\tau^0)^2$) and time (in seconds), approximated with $100$ monte carlo replications.}}
	\resizebox{1\textwidth}{!}{
		\begin{tabular}{cclllllllll}
			\hline
			\multicolumn{2}{c}{$\tau^0=\lfloor0.4\cdotp T\rfloor$} & \multicolumn{3}{c}{Step 1}                                                                                & \multicolumn{3}{c}{AL1}                                                                                   & \multicolumn{3}{c}{WS}                                                                                    \\ \hline
			$T$                        & $p$                       & \multicolumn{1}{c}{\textbf{bias}} & \multicolumn{1}{c}{\textbf{RMSE}} & \multicolumn{1}{c}{\textbf{time}} & \multicolumn{1}{c}{\textbf{bias}} & \multicolumn{1}{c}{\textbf{RMSE}} & \multicolumn{1}{c}{\textbf{time}} & \multicolumn{1}{c}{\textbf{bias}} & \multicolumn{1}{c}{\textbf{RMSE}} & \multicolumn{1}{c}{\textbf{time}} \\ \hline
			200                        & 50                        & 0.350                             & 2.696                             & 0.070                             & 0.290                             & 2.751                             & 0.114                             & 0.180                             & 2.665                             & 0.118                             \\
			200                        & 250                       & 0.520                             & 2.117                             & 0.155                             & 0.370                             & 2.189                             & 0.274                             & 0.100                             & 3.124                             & 1.513                             \\
			200                        & 500                       & 0.000                             & 2.005                             & 0.207                             & 0.090                             & 1.792                             & 0.388                             & 0.500                             & 3.481                             & 7.307                             \\
			200                        & 750                       & 0.230                             & 2.447                             & 0.279                             & 0.190                             & 2.472                             & 0.513                             & 0.110                             & 2.632                             & 23.222                            \\ \hline
			275                        & 50                        & 0.240                             & 1.637                             & 0.092                             & 0.120                             & 1.619                             & 0.168                             & 0.040                             & 2.069                             & 0.141                             \\
			275                        & 250                       & 0.200                             & 2.074                             & 0.215                             & 0.140                             & 1.990                             & 0.387                             & 0.490                             & 3.205                             & 1.835                             \\
			275                        & 500                       & 0.150                             & 1.803                             & 0.340                             & 0.220                             & 1.828                             & 0.620                             & 0.080                             & 2.638                             & 7.929                             \\
			275                        & 750                       & 0.340                             & 2.035                             & 0.400                             & 0.240                             & 1.980                             & 0.752                             & 0.290                             & 2.544                             & 23.124                            \\ \hline
			350                        & 50                        & 0.060                             & 1.667                             & 0.101                             & 0.230                             & 1.578                             & 0.191                             & 0.040                             & 1.649                             & 0.158                             \\
			350                        & 250                       & 0.370                             & 2.142                             & 0.295                             & 0.170                             & 1.797                             & 0.535                             & 0.050                             & 2.105                             & 2.069                             \\
			350                        & 500                       & 0.330                             & 2.830                             & 0.400                             & 0.040                             & 2.702                             & 0.720                             & 0.300                             & 2.821                             & 8.652                             \\
			350                        & 750                       & 0.110                             & 2.007                             & 0.532                             & 0.010                             & 1.952                             & 0.975                             & 0.180                             & 1.811                             & 24.239                            \\ \hline
			425                        & 50                        & 0.340                             & 2.482                             & 0.122                             & 0.120                             & 2.375                             & 0.231                             & 0.110                             & 3.002                             & 0.168                             \\
			425                        & 250                       & 0.140                             & 2.005                             & 0.353                             & 0.120                             & 2.245                             & 0.650                             & 0.170                             & 2.057                             & 2.330                             \\
			425                        & 500                       & 0.320                             & 2.371                             & 0.510                             & 0.020                             & 2.307                             & 0.988                             & 0.050                             & 2.476                             & 9.491                             \\
			425                        & 750                       & 0.160                             & 1.944                             & 0.649                             & 0.100                             & 1.766                             & 1.251                             & 0.350                             & 2.536                             & 26.271                            \\ \hline
	\end{tabular}}
	\label{tab:est.A(i).t04}
\end{table}

\begin{table}[H]
	\caption{\footnotesize{Simulation A(i): estimation performance of Step 1 ($\h\tau$), AL1 $(\breve\tau)$ and WS methods under Gaussian setting with $\tau^0=\lfloor0.6\cdotp T\rfloor.$ Bias ($|E(\h\tau-\tau^0)|$), and RMSE ($E^{1/2}(\h\tau-\tau^0)^2$) and time (in seconds), approximated with $100$ monte carlo replications.}}
	\resizebox{1\textwidth}{!}{	\begin{tabular}{cclllllllll}
			\hline
			\multicolumn{2}{c}{$\tau^0=\lfloor0.6\cdotp T\rfloor$} & \multicolumn{3}{c}{Step 1}                                                                                & \multicolumn{3}{c}{AL1}                                                                                   & \multicolumn{3}{c}{WS}                                                                                    \\ \hline
			$T$                        & $p$                       & \multicolumn{1}{c}{\textbf{bias}} & \multicolumn{1}{c}{\textbf{RMSE}} & \multicolumn{1}{c}{\textbf{time}} & \multicolumn{1}{c}{\textbf{bias}} & \multicolumn{1}{c}{\textbf{RMSE}} & \multicolumn{1}{c}{\textbf{time}} & \multicolumn{1}{c}{\textbf{bias}} & \multicolumn{1}{c}{\textbf{RMSE}} & \multicolumn{1}{c}{\textbf{time}} \\ \hline
			200                        & 50                        & 0.160                             & 1.568                             & 0.066                             & 0.070                             & 1.539                             & 0.115                             & 0.190                             & 2.095                             & 0.123                             \\
			200                        & 250                       & 0.060                             & 1.562                             & 0.162                             & 0.090                             & 1.603                             & 0.281                             & 0.080                             & 2.112                             & 1.602                             \\
			200                        & 500                       & 0.620                             & 3.444                             & 0.218                             & 0.510                             & 3.288                             & 0.399                             & 0.610                             & 3.557                             & 7.431                             \\
			200                        & 750                       & 0.040                             & 2.049                             & 0.269                             & 0.060                             & 1.772                             & 0.492                             & 0.040                             & 2.775                             & 23.284                            \\ \hline
			275                        & 50                        & 0.460                             & 2.272                             & 0.088                             & 0.190                             & 2.313                             & 0.162                             & 0.170                             & 3.220                             & 0.139                             \\
			275                        & 250                       & 0.230                             & 1.775                             & 0.220                             & 0.080                             & 1.691                             & 0.399                             & 0.210                             & 1.634                             & 1.868                             \\
			275                        & 500                       & 0.480                             & 2.040                             & 0.349                             & 0.350                             & 1.931                             & 0.653                             & 0.550                             & 3.041                             & 8.245                             \\
			275                        & 750                       & 0.100                             & 2.015                             & 0.404                             & 0.110                             & 2.335                             & 0.743                             & 0.210                             & 3.110                             & 23.378                            \\ \hline
			350                        & 50                        & 0.360                             & 1.794                             & 0.101                             & 0.320                             & 1.783                             & 0.191                             & 0.120                             & 2.093                             & 0.161                             \\
			350                        & 250                       & 0.280                             & 2.241                             & 0.290                             & 0.130                             & 2.133                             & 0.547                             & 0.010                             & 2.402                             & 2.167                             \\
			350                        & 500                       & 0.210                             & 1.609                             & 0.383                             & 0.040                             & 1.517                             & 0.708                             & 0.120                             & 2.088                             & 8.710                             \\
			350                        & 750                       & 0.470                             & 2.812                             & 0.510                             & 0.300                             & 2.608                             & 0.957                             & 0.570                             & 3.404                             & 24.317                            \\ \hline
			425                        & 50                        & 0.350                             & 2.086                             & 0.129                             & 0.080                             & 1.811                             & 0.235                             & 0.250                             & 2.304                             & 0.183                             \\
			425                        & 250                       & 0.190                             & 1.622                             & 0.370                             & 0.050                             & 1.916                             & 0.706                             & 0.010                             & 2.335                             & 2.330                             \\
			425                        & 500                       & 0.200                             & 2.074                             & 0.473                             & 0.370                             & 2.110                             & 0.909                             & 0.600                             & 2.429                             & 9.331                             \\
			425                        & 750                       & 0.320                             & 1.661                             & 0.711                             & 0.200                             & 1.720                             & 1.362                             & 0.280                             & 2.000                             & 26.281                            \\ \hline
	\end{tabular}}
	\label{tab:est.A(i).t06}
\end{table}

\begin{table}[H]
	\caption{\footnotesize{Simulation A(i): estimation performance of Step 1 ($\h\tau$), AL1 $(\breve\tau)$ and WS methods under Gaussian setting with $\tau^0=\lfloor0.8\cdotp T\rfloor.$ Bias ($|E(\h\tau-\tau^0)|$), and RMSE ($E^{1/2}(\h\tau-\tau^0)^2$) and time (in seconds), approximated with $100$ monte carlo replications.}}
	\resizebox{1\textwidth}{!}{	\begin{tabular}{cclllllllll}
			\hline
			\multicolumn{2}{c}{$\tau^0=\lfloor0.8\cdotp T\rfloor$} & \multicolumn{3}{c}{Step 1}                                                                                & \multicolumn{3}{c}{AL1}                                                                                   & \multicolumn{3}{c}{WS}                                                                                    \\ \hline
			$T$                        & $p$                       & \multicolumn{1}{c}{\textbf{bias}} & \multicolumn{1}{c}{\textbf{RMSE}} & \multicolumn{1}{c}{\textbf{time}} & \multicolumn{1}{c}{\textbf{bias}} & \multicolumn{1}{c}{\textbf{RMSE}} & \multicolumn{1}{c}{\textbf{time}} & \multicolumn{1}{c}{\textbf{bias}} & \multicolumn{1}{c}{\textbf{RMSE}} & \multicolumn{1}{c}{\textbf{time}} \\ \hline
			200                        & 50                        & 2.880                             & 7.347                             & 0.065                             & 0.110                             & 1.658                             & 0.113                             & 0.320                             & 2.195                             & 0.119                             \\
			200                        & 250                       & 1.250                             & 6.298                             & 0.155                             & 0.380                             & 2.272                             & 0.267                             & 1.170                             & 6.268                             & 1.537                             \\
			200                        & 500                       & 1.110                             & 7.077                             & 0.195                             & 0.400                             & 1.908                             & 0.358                             & 1.090                             & 3.936                             & 7.182                             \\
			200                        & 750                       & 1.220                             & 4.572                             & 0.266                             & 0.460                             & 2.608                             & 0.488                             & 3.380                             & 8.470                             & 23.587                            \\ \hline
			275                        & 50                        & 3.520                             & 7.521                             & 0.088                             & 0.690                             & 2.161                             & 0.167                             & 0.110                             & 2.456                             & 0.143                             \\
			275                        & 250                       & 1.960                             & 4.357                             & 0.242                             & 0.870                             & 3.500                             & 0.424                             & 0.930                             & 4.073                             & 1.927                             \\
			275                        & 500                       & 2.100                             & 7.624                             & 0.318                             & 0.040                             & 1.822                             & 0.588                             & 1.220                             & 4.212                             & 7.834                             \\
			275                        & 750                       & 1.190                             & 3.486                             & 0.387                             & 0.150                             & 2.309                             & 0.728                             & 1.470                             & 4.487                             & 23.118                            \\ \hline
			350                        & 50                        & 3.230                             & 6.587                             & 0.097                             & 0.460                             & 3.036                             & 0.185                             & 0.080                             & 2.821                             & 0.164                             \\
			350                        & 250                       & 2.440                             & 5.669                             & 0.249                             & 0.470                             & 2.659                             & 0.452                             & 0.340                             & 2.691                             & 2.012                             \\
			350                        & 500                       & 1.180                             & 3.672                             & 0.410                             & 0.200                             & 1.679                             & 0.763                             & 1.160                             & 4.402                             & 8.761                             \\
			350                        & 750                       & 1.530                             & 4.487                             & 0.489                             & 0.510                             & 2.216                             & 0.890                             & 1.860                             & 5.185                             & 24.224                            \\ \hline
			425                        & 50                        & 3.580                             & 7.254                             & 0.122                             & 0.310                             & 1.752                             & 0.235                             & 0.290                             & 2.504                             & 0.179                             \\
			425                        & 250                       & 1.940                             & 4.459                             & 0.304                             & 0.420                             & 2.400                             & 0.581                             & 0.200                             & 3.000                             & 2.272                             \\
			425                        & 500                       & 1.770                             & 3.814                             & 0.500                             & 0.410                             & 1.688                             & 0.948                             & 0.550                             & 2.452                             & 9.484                             \\
			425                        & 750                       & 2.800                             & 8.656                             & 0.716                             & 0.370                             & 2.071                             & 1.406                             & 1.190                             & 5.381                             & 26.215                            \\ \hline
	\end{tabular}}
	\label{tab:est.A(i).t08}
\end{table}

%%%%%%%%%%%%%%%%%%%%%LAPLACE ESTIMATION TABLES%%%%%%%%%%%%%%%%%%%%%%%%%%%%%%%%%%%%%%%%%%%%%%%%

\begin{table}[H]
	\caption{\footnotesize{Simulation B(i): estimation performance of Step 1 ($\h\tau$), AL1 $(\breve\tau)$ and WS methods under Laplace setting with $\tau^0=\lfloor 0.4\cdotp T\rfloor.$ Bias ($|E(\h\tau-\tau^0)|$), and RMSE ($E^{1/2}(\h\tau-\tau^0)^2$) and time (in seconds), approximated with $100$ monte carlo replications.}}
	\resizebox{1\textwidth}{!}{	\begin{tabular}{cclllllllll}
			\hline
			\multicolumn{2}{c}{$\tau^0=\lfloor0.4\cdotp T\rfloor$} & \multicolumn{3}{c}{Step 1}                                                                                & \multicolumn{3}{c}{AL1}                                                                                   & \multicolumn{3}{c}{WS}                                                                                    \\ \hline
			$T$                        & $p$                       & \multicolumn{1}{c}{\textbf{bias}} & \multicolumn{1}{c}{\textbf{RMSE}} & \multicolumn{1}{c}{\textbf{time}} & \multicolumn{1}{c}{\textbf{bias}} & \multicolumn{1}{c}{\textbf{RMSE}} & \multicolumn{1}{c}{\textbf{time}} & \multicolumn{1}{c}{\textbf{bias}} & \multicolumn{1}{c}{\textbf{RMSE}} & \multicolumn{1}{c}{\textbf{time}} \\ \hline
			200                        & 50                        & 0.320                             & 2.341                             & 0.061                             & 0.030                             & 1.658                             & 0.104                             & 0.010                             & 1.836                             & 0.115                             \\
			200                        & 250                       & 0.250                             & 2.027                             & 0.137                             & 0.240                             & 2.035                             & 0.231                             & 0.310                             & 3.002                             & 1.424                             \\
			200                        & 500                       & 0.150                             & 2.317                             & 0.226                             & 0.190                             & 2.095                             & 0.415                             & 0.010                             & 2.472                             & 7.184                             \\
			200                        & 750                       & 0.010                             & 2.119                             & 0.258                             & 0.070                             & 1.389                             & 0.483                             & 0.170                             & 2.390                             & 21.718                            \\ \hline
			275                        & 50                        & 0.400                             & 1.697                             & 0.100                             & 0.140                             & 1.435                             & 0.168                             & 0.030                             & 2.062                             & 0.134                             \\
			275                        & 250                       & 0.290                             & 1.967                             & 0.228                             & 0.180                             & 1.860                             & 0.400                             & 0.080                             & 3.156                             & 1.907                             \\
			275                        & 500                       & 0.940                             & 2.267                             & 0.317                             & 0.710                             & 1.916                             & 0.586                             & 0.480                             & 2.437                             & 7.603                             \\
			275                        & 750                       & 0.050                             & 2.722                             & 0.407                             & 0.020                             & 2.542                             & 0.744                             & 0.480                             & 2.093                             & 22.957                            \\ \hline
			350                        & 50                        & 0.410                             & 2.189                             & 0.098                             & 0.310                             & 2.166                             & 0.173                             & 0.060                             & 2.581                             & 0.133                             \\
			350                        & 250                       & 0.140                             & 1.449                             & 0.293                             & 0.050                             & 1.404                             & 0.547                             & 0.150                             & 1.825                             & 2.184                             \\
			350                        & 500                       & 0.020                             & 1.649                             & 0.381                             & 0.280                             & 1.822                             & 0.693                             & 0.320                             & 2.358                             & 8.262                             \\
			350                        & 750                       & 0.090                             & 2.095                             & 0.486                             & 0.070                             & 1.841                             & 0.897                             & 0.050                             & 1.889                             & 24.225                            \\ \hline
			425                        & 50                        & 0.270                             & 1.895                             & 0.116                             & 0.150                             & 1.841                             & 0.218                             & 0.230                             & 1.879                             & 0.164                             \\
			425                        & 250                       & 0.360                             & 1.794                             & 0.339                             & 0.340                             & 1.811                             & 0.649                             & 0.280                             & 2.145                             & 2.291                             \\
			425                        & 500                       & 0.030                             & 1.597                             & 0.524                             & 0.070                             & 1.758                             & 1.000                             & 0.050                             & 2.012                             & 9.266                             \\
			425                        & 750                       & 0.310                             & 1.396                             & 0.635                             & 0.050                             & 1.432                             & 1.213                             & 0.070                             & 2.128                             & 25.804                            \\ \hline
	\end{tabular}}
	\label{tab:est.B(i).t04}
\end{table}

\begin{table}[H]
	\caption{\footnotesize{Simulation B(i): estimation performance of Step 1 ($\h\tau$), AL1 $(\breve\tau)$ and WS methods under Laplace setting with $\tau^0=\lfloor 0.6\cdotp T\rfloor.$ Bias ($|E(\h\tau-\tau^0)|$), and RMSE ($E^{1/2}(\h\tau-\tau^0)^2$) and time (in seconds), approximated with $100$ monte carlo replications.}}
	\resizebox{1\textwidth}{!}{	\begin{tabular}{cclllllllll}
			\hline
			\multicolumn{2}{c}{$\tau^0=\lfloor0.6\cdotp T\rfloor$} & \multicolumn{3}{c}{Step 1}                                                                                & \multicolumn{3}{c}{AL1}                                                                                   & \multicolumn{3}{c}{WS}                                                                                    \\ \hline
			$T$                        & $p$                       & \multicolumn{1}{c}{\textbf{bias}} & \multicolumn{1}{c}{\textbf{RMSE}} & \multicolumn{1}{c}{\textbf{time}} & \multicolumn{1}{c}{\textbf{bias}} & \multicolumn{1}{c}{\textbf{RMSE}} & \multicolumn{1}{c}{\textbf{time}} & \multicolumn{1}{c}{\textbf{bias}} & \multicolumn{1}{c}{\textbf{RMSE}} & \multicolumn{1}{c}{\textbf{time}} \\ \hline
			200                        & 50                        & 0.230                             & 1.584                             & 0.061                             & 0.090                             & 1.729                             & 0.110                             & 0.120                             & 2.362                             & 0.110                             \\
			200                        & 250                       & 0.030                             & 1.792                             & 0.139                             & 0.030                             & 1.741                             & 0.234                             & 0.300                             & 1.954                             & 1.442                             \\
			200                        & 500                       & 0.480                             & 2.387                             & 0.242                             & 0.170                             & 2.265                             & 0.415                             & 0.090                             & 2.156                             & 7.175                             \\
			200                        & 750                       & 0.300                             & 2.706                             & 0.245                             & 0.340                             & 2.437                             & 0.451                             & 0.260                             & 3.030                             & 21.480                            \\ \hline
			275                        & 50                        & 0.440                             & 1.783                             & 0.088                             & 0.180                             & 1.755                             & 0.168                             & 0.170                             & 1.597                             & 0.126                             \\
			275                        & 250                       & 0.350                             & 1.863                             & 0.222                             & 0.140                             & 1.523                             & 0.391                             & 0.330                             & 2.594                             & 1.843                             \\
			275                        & 500                       & 0.240                             & 2.191                             & 0.321                             & 0.250                             & 2.193                             & 0.602                             & 0.010                             & 3.012                             & 7.660                             \\
			275                        & 750                       & 0.190                             & 1.221                             & 0.406                             & 0.040                             & 1.530                             & 0.753                             & 0.040                             & 2.214                             & 23.008                            \\ \hline
			350                        & 50                        & 0.610                             & 1.905                             & 0.097                             & 0.250                             & 2.147                             & 0.163                             & 0.170                             & 2.617                             & 0.137                             \\
			350                        & 250                       & 0.090                             & 1.977                             & 0.266                             & 0.030                             & 1.884                             & 0.506                             & 0.110                             & 2.693                             & 1.947                             \\
			350                        & 500                       & 0.210                             & 1.863                             & 0.387                             & 0.150                             & 1.735                             & 0.695                             & 0.090                             & 2.202                             & 8.285                             \\
			350                        & 750                       & 0.390                             & 2.278                             & 0.474                             & 0.040                             & 1.934                             & 0.883                             & 0.360                             & 2.731                             & 24.209                            \\ \hline
			425                        & 50                        & 0.280                             & 2.263                             & 0.116                             & 0.010                             & 2.052                             & 0.219                             & 0.200                             & 2.319                             & 0.158                             \\
			425                        & 250                       & 0.210                             & 1.936                             & 0.343                             & 0.050                             & 1.879                             & 0.651                             & 0.240                             & 1.975                             & 2.316                             \\
			425                        & 500                       & 0.260                             & 2.069                             & 0.506                             & 0.110                             & 2.100                             & 0.969                             & 0.360                             & 2.619                             & 9.183                             \\
			425                        & 750                       & 0.200                             & 2.030                             & 0.657                             & 0.020                             & 1.703                             & 1.253                             & 0.080                             & 2.289                             & 25.939                            \\ \hline
	\end{tabular}}
	\label{tab:est.B(i).t06}
\end{table}

\begin{table}[H]
	\caption{\footnotesize{Simulation B(i): estimation performance of Step 1 ($\h\tau$), AL1 $(\breve\tau)$ and WS methods under Laplace setting with $\tau^0=\lfloor 0.8\cdotp T\rfloor.$ Bias ($|E(\h\tau-\tau^0)|$), and RMSE ($E^{1/2}(\h\tau-\tau^0)^2$) and time (in seconds), approximated with $100$ monte carlo replications.}}
	\resizebox{1\textwidth}{!}{	\begin{tabular}{cclllllllll}
			\hline
			\multicolumn{2}{c}{$\tau^0=\lfloor0.8\cdotp T\rfloor$} & \multicolumn{3}{c}{Step 1}                                                                                & \multicolumn{3}{c}{AL1}                                                                                   & \multicolumn{3}{c}{WS}                                                                                    \\ \hline
			$T$                        & $p$                       & \multicolumn{1}{c}{\textbf{bias}} & \multicolumn{1}{c}{\textbf{RMSE}} & \multicolumn{1}{c}{\textbf{time}} & \multicolumn{1}{c}{\textbf{bias}} & \multicolumn{1}{c}{\textbf{RMSE}} & \multicolumn{1}{c}{\textbf{time}} & \multicolumn{1}{c}{\textbf{bias}} & \multicolumn{1}{c}{\textbf{RMSE}} & \multicolumn{1}{c}{\textbf{time}} \\ \hline
			200                        & 50                        & 3.320                             & 6.684                             & 0.063                             & 0.660                             & 2.881                             & 0.110                             & 0.770                             & 3.045                             & 0.115                             \\
			200                        & 250                       & 2.430                             & 7.997                             & 0.143                             & 1.320                             & 7.122                             & 0.246                             & 1.790                             & 7.115                             & 1.483                             \\
			200                        & 500                       & 0.850                             & 5.469                             & 0.213                             & 0.070                             & 1.916                             & 0.372                             & 1.250                             & 5.944                             & 7.111                             \\
			200                        & 750                       & 1.010                             & 7.072                             & 0.241                             & 0.990                             & 5.733                             & 0.444                             & 3.050                             & 7.584                             & 21.474                            \\ \hline
			275                        & 50                        & 4.010                             & 10.014                            & 0.088                             & 0.890                             & 5.015                             & 0.164                             & 0.250                             & 4.253                             & 0.127                             \\
			275                        & 250                       & 1.650                             & 4.892                             & 0.236                             & 0.380                             & 2.263                             & 0.421                             & 0.600                             & 2.789                             & 1.925                             \\
			275                        & 500                       & 1.720                             & 5.982                             & 0.294                             & 0.240                             & 2.173                             & 0.535                             & 1.370                             & 4.034                             & 7.520                             \\
			275                        & 750                       & 0.280                             & 3.552                             & 0.420                             & 0.210                             & 2.022                             & 0.769                             & 1.860                             & 5.521                             & 22.996                            \\ \hline
			350                        & 50                        & 3.120                             & 7.580                             & 0.102                             & 0.730                             & 3.887                             & 0.168                             & 0.050                             & 3.055                             & 0.135                             \\
			350                        & 250                       & 1.430                             & 3.874                             & 0.262                             & 0.210                             & 2.356                             & 0.474                             & 0.650                             & 3.604                             & 2.046                             \\
			350                        & 500                       & 1.460                             & 4.176                             & 0.389                             & 0.240                             & 1.833                             & 0.711                             & 0.040                             & 2.425                             & 8.273                             \\
			350                        & 750                       & 1.990                             & 6.856                             & 0.486                             & 0.190                             & 2.571                             & 0.924                             & 1.930                             & 6.564                             & 24.302                            \\ \hline
			425                        & 50                        & 2.650                             & 6.489                             & 0.117                             & 0.380                             & 1.661                             & 0.223                             & 0.080                             & 1.908                             & 0.163                             \\
			425                        & 250                       & 1.780                             & 4.459                             & 0.307                             & 0.150                             & 1.962                             & 0.592                             & 0.450                             & 4.009                             & 2.198                             \\
			425                        & 500                       & 1.770                             & 5.811                             & 0.480                             & 0.040                             & 1.685                             & 0.910                             & 0.300                             & 3.165                             & 9.161                             \\
			425                        & 750                       & 2.010                             & 4.288                             & 0.632                             & 0.710                             & 2.629                             & 1.209                             & 1.620                             & 4.539                             & 25.863                            \\ \hline
	\end{tabular}}
	\label{tab:est.B(i).t08}
\end{table}

%%%%%%%%%%%%%%%%%GAUSSIAN INFERENCE TABLES%%%%%%%%%%%%%%%%%%%%%%%%%%%%%%%%%%%%%%%%%%%%%%%%%%%%

\begin{table}[H]
	\caption{\footnotesize{Simulation A(ii): inference using AL1 $(\breve\tau)$ with $\tau^0=\lfloor 0.4\cdotp T\rfloor,$ at significance level $\al=0.05.$ Here, {\bf V:} confidence intervals constructed using Theorem \ref{thm:wc.vanishing} under vanishing regime, {\bf NV:} confidence intervals constructed using Theorem \ref{thm:wc.non.vanishing} under non-vanishing regime (Gaussian parametric assumption). Computation based on 500 monte carlo replications.}}
	\resizebox{1\textwidth}{!}{	\begin{tabular}{cllllll}
			\hline
			\multirow{2}{*}{} & \multicolumn{6}{c}{\textbf{Coverage (average margin of error)}}                                                                                  \\ \cline{2-7}
			& \multicolumn{1}{c}{V} & \multicolumn{1}{c}{NV} & \multicolumn{1}{c}{V} & \multicolumn{1}{c}{NV} & \multicolumn{1}{c}{V} & \multicolumn{1}{c}{NV} \\ \hline
			$p$               & \multicolumn{2}{c}{$n=275$}                    & \multicolumn{2}{c}{$n=350$}                    & \multicolumn{2}{c}{$n=425$}                    \\ \hline
			50                & 0.956 (4.020)         & 0.966 (3.929)          & 0.916 (4.111)         & 0.926 (4.021)          & 0.956 (4.136)         & 0.968 (4.058)          \\
			250               & 0.938 (3.936)         & 0.946 (3.859)          & 0.946 (3.984)         & 0.958 (3.901)          & 0.944 (4.025)         & 0.956 (3.949)          \\
			500               & 0.934 (3.886)         & 0.942 (3.783)          & 0.934 (3.925)         & 0.940 (3.847)          & 0.950 (3.982)         & 0.958 (3.912)          \\
			750               & 0.944 (3.848)         & 0.954 (3.753)          & 0.936 (3.922)         & 0.954 (3.843)          & 0.946 (3.961)         & 0.954 (3.905)          \\ \hline
	\end{tabular}}
	\label{tab:inf.A(ii).t04}
\end{table}

% \usepackage{multirow}
\begin{table}[H]
	\caption{\footnotesize{Simulation A(ii): inference using AL1 $(\breve\tau)$ with $\tau^0=\lfloor 0.6\cdotp T\rfloor,$ at significance level $\al=0.05.$ Here, {\bf V:} confidence intervals constructed using Theorem \ref{thm:wc.vanishing} under vanishing regime, {\bf NV:} confidence intervals constructed using Theorem \ref{thm:wc.non.vanishing} under non-vanishing regime (Gaussian parametric assumption). Computation based on 500 monte carlo replications.}}
	\resizebox{1\textwidth}{!}{	\begin{tabular}{cllllll}
			\hline
			\multirow{2}{*}{} & \multicolumn{6}{c}{\textbf{Coverage (average margin of error)}}                                                                                  \\ \cline{2-7}
			& \multicolumn{1}{c}{V} & \multicolumn{1}{c}{NV} & \multicolumn{1}{c}{V} & \multicolumn{1}{c}{NV} & \multicolumn{1}{c}{V} & \multicolumn{1}{c}{NV} \\ \hline
			$p$               & \multicolumn{2}{c}{$n=275$}                    & \multicolumn{2}{c}{$n=350$}                    & \multicolumn{2}{c}{$n=425$}                    \\ \hline
			50                & 0.956 (4.033)         & 0.970 (3.965)          & 0.932 (4.054)         & 0.952 (4.030)          & 0.938 (4.101)         & 0.952 (4.044)          \\
			250               & 0.918 (3.916)         & 0.936 (3.852)          & 0.946 (3.988)         & 0.958 (3.917)          & 0.938 (4.008)         & 0.942 (3.907)          \\
			500               & 0.948 (3.881)         & 0.950 (3.804)          & 0.944 (3.901)         & 0.962 (3.813)          & 0.946 (4.003)         & 0.950 (3.945)          \\
			750               & 0.930 (3.823)         & 0.944 (3.745)          & 0.940 (3.981)         & 0.948 (3.891)          & 0.958 (3.999)         & 0.964 (3.903)          \\ \hline
	\end{tabular}}
	\label{tab:inf.A(ii).t06}
\end{table}

\begin{table}[H]
	\caption{\footnotesize{Simulation A(ii): inference using AL1 $(\breve\tau)$ with $\tau^0=\lfloor 0.8\cdotp T\rfloor,$ at significance level $\al=0.05.$ Here, {\bf V:} confidence intervals constructed using Theorem \ref{thm:wc.vanishing} under vanishing regime, {\bf NV:} confidence intervals constructed using Theorem \ref{thm:wc.non.vanishing} under non-vanishing regime (Gaussian parametric assumption). Computation based on 500 monte carlo replications.}}
	\resizebox{1\textwidth}{!}{	\begin{tabular}{cllllll}
			\hline
			\multirow{2}{*}{} & \multicolumn{6}{c}{\textbf{Coverage (average margin of error)}}                                                                                  \\ \cline{2-7}
			& \multicolumn{1}{c}{V} & \multicolumn{1}{c}{NV} & \multicolumn{1}{c}{V} & \multicolumn{1}{c}{NV} & \multicolumn{1}{c}{V} & \multicolumn{1}{c}{NV} \\ \hline
			$p$               & \multicolumn{2}{c}{$n=275$}                    & \multicolumn{2}{c}{$n=350$}                    & \multicolumn{2}{c}{$n=425$}                    \\ \hline
			50                & 0.942 (3.831)         & 0.954 (3.758)          & 0.932 (3.861)         & 0.944 (3.789)          & 0.934 (3.951)         & 0.944 (3.884)          \\
			250               & 0.896 (3.417)         & 0.920 (3.349)          & 0.920 (3.602)         & 0.928 (3.537)          & 0.940 (3.716)         & 0.950 (3.644)          \\
			500               & 0.908 (3.307)         & 0.926 (3.235)          & 0.896 (3.476)         & 0.910 (3.392)          & 0.916 (3.612)         & 0.930 (3.491)          \\
			750               & 0.890 (3.252)         & 0.916 (3.158)          & 0.896 (3.408)         & 0.922 (3.324)          & 0.928 (3.542)         & 0.948 (3.461)          \\ \hline
	\end{tabular}}
	\label{tab:inf.A(ii).t08}
\end{table}

%%%%%%%%%%%%%%%%%LAPLACE INFERENCE TABLES%%%%%%%%%%%%%%%%%%%%%%%%%%%%%%%%%%%%%%%%%%%%%%%%%%%%

\begin{table}[H]
	\caption{\footnotesize{Simulation B(ii): inference using AL1 $(\breve\tau)$ with $\tau^0=\lfloor 0.4\cdotp T\rfloor,$ at significance level $\al=0.05.$ Here, {\bf V:} confidence intervals constructed using Theorem \ref{thm:wc.vanishing} under vanishing regime, {\bf NV:} confidence intervals constructed using Theorem \ref{thm:wc.non.vanishing} under non-vanishing regime (Laplace parametric assumption). Computation based on 500 monte carlo replications.}}
	\resizebox{1\textwidth}{!}{	\begin{tabular}{cllllll}
			\hline
			\multirow{2}{*}{} & \multicolumn{6}{c}{\textbf{Coverage (average margin of error)}}                                                                                  \\ \cline{2-7}
			& \multicolumn{1}{c}{V} & \multicolumn{1}{c}{NV} & \multicolumn{1}{c}{V} & \multicolumn{1}{c}{NV} & \multicolumn{1}{c}{V} & \multicolumn{1}{c}{NV} \\ \hline
			$p$               & \multicolumn{2}{c}{$n=275$}                    & \multicolumn{2}{c}{$n=350$}                    & \multicolumn{2}{c}{$n=425$}                    \\ \hline
			50                & 0.946 (4.029)         & 0.956 (4.015)          & 0.932 (4.099)         & 0.940 (4.083)          & 0.946 (4.113)         & 0.950 (4.067)          \\
			250               & 0.932 (3.926)         & 0.942 (3.905)          & 0.950 (3.984)         & 0.964 (3.933)          & 0.936 (4.035)         & 0.958 (3.990)          \\
			500               & 0.938 (3.878)         & 0.948 (3.841)          & 0.940 (3.892)         & 0.952 (3.849)          & 0.942 (4.055)         & 0.960 (4.023)          \\
			750               & 0.938 (3.843)         & 0.952 (3.769)          & 0.930 (3.931)         & 0.948 (3.871)          & 0.946 (3.947)         & 0.962 (3.919)          \\ \hline
	\end{tabular}}
	\label{tab:inf.B(ii).t04}
\end{table}

\begin{table}[H]
	\caption{\footnotesize{Simulation B(ii): inference using AL1 $(\breve\tau)$ with $\tau^0=\lfloor 0.6\cdotp T\rfloor,$ at significance level $\al=0.05.$ Here, {\bf V:} confidence intervals constructed using Theorem \ref{thm:wc.vanishing} under vanishing regime, {\bf NV:} confidence intervals constructed using Theorem \ref{thm:wc.non.vanishing} under non-vanishing regime (Laplace parametric assumption). Computation based on 500 monte carlo replications.}}
	\resizebox{1\textwidth}{!}{	\begin{tabular}{cllllll}
			\hline
			\multirow{2}{*}{} & \multicolumn{6}{c}{\textbf{Coverage (average margin of error)}}                                                                                  \\ \cline{2-7}
			& \multicolumn{1}{c}{V} & \multicolumn{1}{c}{NV} & \multicolumn{1}{c}{V} & \multicolumn{1}{c}{NV} & \multicolumn{1}{c}{V} & \multicolumn{1}{c}{NV} \\ \hline
			$p$               & \multicolumn{2}{c}{$n=275$}                    & \multicolumn{2}{c}{$n=350$}                    & \multicolumn{2}{c}{$n=425$}                    \\ \hline
			50                & 0.950 (4.043)         & 0.958 (4.003)          & 0.952 (4.100)         & 0.964 (4.075)          & 0.930 (4.103)         & 0.944 (4.051)          \\
			250               & 0.940 (3.945)         & 0.952 (3.897)          & 0.912 (3.989)         & 0.944 (3.965)          & 0.948 (4.029)         & 0.960 (4.030)          \\
			500               & 0.940 (3.875)         & 0.948 (3.822)          & 0.940 (3.969)         & 0.950 (3.924)          & 0.944 (4.006)         & 0.958 (3.941)          \\
			750               & 0.928 (3.827)         & 0.950 (3.769)          & 0.950 (3.942)         & 0.964 (3.895)          & 0.930 (4.002)         & 0.948 (3.949)          \\ \hline
	\end{tabular}}
	\label{tab:inf.B(ii).t06}
\end{table}

\begin{table}[H]
	\caption{\footnotesize{Simulation B(ii): inference using AL1 $(\breve\tau)$ with $\tau^0=\lfloor 0.8\cdotp T\rfloor,$ at significance level $\al=0.05.$ Here, {\bf V:} confidence intervals constructed using Theorem \ref{thm:wc.vanishing} under vanishing regime, {\bf NV:} confidence intervals constructed using Theorem \ref{thm:wc.non.vanishing} under non-vanishing regime (Laplace parametric assumption). Computation based on 500 monte carlo replications.}}
	\resizebox{1\textwidth}{!}{	\begin{tabular}{cllllll}
			\hline
			\multirow{2}{*}{} & \multicolumn{6}{c}{\textbf{Coverage (average margin of error)}}                                                                                  \\ \cline{2-7}
			& \multicolumn{1}{c}{V} & \multicolumn{1}{c}{NV} & \multicolumn{1}{c}{V} & \multicolumn{1}{c}{NV} & \multicolumn{1}{c}{V} & \multicolumn{1}{c}{NV} \\ \hline
			$p$               & \multicolumn{2}{c}{$n=275$}                    & \multicolumn{2}{c}{$n=350$}                    & \multicolumn{2}{c}{$n=425$}                    \\ \hline
			50                & 0.918 (3.841)         & 0.934 (3.789)          & 0.934 (3.883)         & 0.950 (3.833)          & 0.924 (3.966)         & 0.944 (3.907)          \\
			250               & 0.930 (3.410)         & 0.940 (3.423)          & 0.944 (3.597)         & 0.952 (3.563)          & 0.932 (3.624)         & 0.938 (3.565)          \\
			500               & 0.904 (3.404)         & 0.916 (3.368)          & 0.936 (3.480)         & 0.954 (3.431)          & 0.916 (3.600)         & 0.942 (3.561)          \\
			750               & 0.898 (3.234)         & 0.924 (3.198)          & 0.900 (3.410)         & 0.916 (3.391)          & 0.924 (3.533)         & 0.942 (3.487)          \\ \hline
	\end{tabular}}
	\label{tab:inf.B(ii).t08}
\end{table}

\bibliographystyle{plainnat}
\bibliography{meanchange}

\end{document}